\documentclass[article]{amsart}

\usepackage[round]{natbib}
\usepackage[a4paper, total={6in, 8in}]{geometry}
\usepackage{rotating}
\usepackage{scrextend}
\usepackage{tablefootnote}
\usepackage{caption}
\usepackage{subcaption}
\usepackage{hyperref}       
\usepackage[english]{babel}
\usepackage{blindtext}
\usepackage{a4wide}
\usepackage[T1]{fontenc}
\usepackage{lscape}
\usepackage{amsaddr}

\usepackage{amsfonts,amsmath,amssymb,amsthm}
\usepackage{latexsym}
\usepackage{siunitx}
\usepackage{layout}
\usepackage{dsfont}
\usepackage{xcolor}
\usepackage{color,soul}
\usepackage{graphicx}
\usepackage{graphics}
\usepackage{multirow}
\usepackage{booktabs}
\usepackage{tikz}
\usepackage{comment}
\usepackage{footmisc}
\usepackage{bm}
\usepackage{float}
\restylefloat{table}
\usepackage[linesnumbered,ruled,vlined]{algorithm2e}
\SetKwInput{KwSetup}{Setup}                
\SetKwInput{KwInput}{Input}                

\SetCommentSty{mycommfont}

\usepackage{graphicx}

\newtheorem{prop}{Proposition}[section]
\newtheorem{lemma}{Lemma}[section]
\newtheorem{corollary}{Corollary}[section]
\newtheorem{assump}{Assumption}[section]
\newtheorem{remark}{Remark}[section]

\newtheorem*{proposition*}{Proposition}
\theoremstyle{definition}

\newcommand{\tN}{\tilde{N}}
\newcommand{\tiT}{\tilde{t}}
\newcommand{\stateSpace}{\mathcal{S}}
\newcommand{\prob}[1]{\mathbb{#1}}

\newcommand{\e}{\mathbb{E}}

\newcommand{\reals}{\mathbb{R}}
\newcommand{\ind}{\mathbf{1}}

\newcommand{\diff}{\,\mathrm{d}}

\newcommand{\tS}{\tilde{S}}

\newcommand{\XmM}{X^{(m,M)}}
\newcommand{\Rm}{R^{(m)}}
\newcommand{\diag}{\text{diag}}
\renewcommand{\tilde}{\widetilde}
\renewcommand{\hat}{\widehat}
\allowdisplaybreaks
\title[A Unifying Approach for the Pricing of Debt Securities]{A Unifying Approach for the Pricing of Debt Securities} 
	
\author{Marie-Claude Vachon$^{\ast\ddag}$ and Anne MacKay$^{\dag\ddag}$}
\address{$\dag$Universit\'e de Sherbooke, Sherbrooke, Qu\'ebec, Canada\\
$\ddag$ Universit\'e du Qu\'ebec \`a Montr\'eal, Montr\'eal, Qu\'ebec, Canada\\
}
\thanks{$^\ast$Corresponding author. Email: vachon.marie-claude.2@courrier.uqam.ca\\}

\date{\today}

\begin{document}
	
	\begin{abstract}
		We propose a unifying framework for the pricing of debt securities under general time-inhomogeneous short-rate diffusion processes. The pricing of bonds, bond options, 
callable/putable bonds, and convertible bonds (CBs) is covered. Using continuous-time Markov chain (CTMC) approximations, we obtain closed-form matrix expressions to approximate the price of bonds and bond options under general one-dimensional short-rate processes. A simple and efficient algorithm is also developed to price callable/putable debt. The availability of a closed-form expression for the price of zero-coupon bonds allows for the perfect fit of the approximated model to the current market term structure of interest rates, regardless of the complexity of the underlying diffusion process selected. 
        We further consider the pricing of CBs under general bi-dimensional time-inhomogeneous diffusion processes to model equity and short-rate dynamics. Credit risk is also incorporated into the model using the approach of \cite{tsiveriotis1998valuing}. 
        Based on a two-layer CTMC method, an efficient algorithm is developed to approximate the price of convertible bonds. When conversion is only allowed at maturity, a closed-form matrix expression is obtained. 
        Numerical experiments show the accuracy and efficiency of the method across a wide range of model parameters and short-rate models. 
		
	\end{abstract}

	\maketitle	
\newpage
\section{Introduction}
 
This paper proposes a unifying framework based on continuous-time Markov chain (CTMC) approximations to price debt securities under general time-inhomogenous short-rate models. 
 Over the last few years, CTMC methods have garnered attention in option pricing literature, see for instance \cite{cui2018general}, \cite{ding2021markov}, and \cite{kirkby2022hybrid}, among others. In particular, \cite{cui2018general} developed a two-layer CTMC technique to price European, barrier, Bermudian, Asian, and occupation time derivatives under general stochastic volatility models, while \cite{ding2021markov} discussed the extension of the method to time-inhomogeneous processes for the pricing of European and barrier options. Recently, \cite{kirkby2022hybrid} approximated time-homogeneous bi-dimensional diffusion processes to model short-rates and equity for the pricing of hybrid securities such as equity swap and cap via CTMC approximation. The framework outlined below is an extension to the work of \cite{cui2018general} to time-inhomogeneous processes for the pricing of debt securities such as bonds, bond options (loan commitments and deposit instruments), callable/putable bonds, and convertible bonds (CBs).  
 
 The advantage of CTMC approximations over other numerical techniques resides in their ability to easily adapt to various diffusion processes (homogeneous as well as inhomogeneous processes), and their extension to higher dimension models is also relatively straightforward \cite{cui2018general} (for two-dimension), \cite{kirkby2020general} (for higher dimension). More importantly, they generally allow for an explicit formulation of expectations (resp. conditional expectations), see, for instance, \cite{cui2019continuous}, which facilitate the pricing of European (resp. American)-type derivatives. 
 In particular, the method allows for a closed-form matrix expression for the price of zero-coupon bonds regardless of the complexity of the short-rate dynamics selected, simplifying the calibration of the approximated model to the current market-term structure for a wide range of short-rate models. 
Calibration to the current market term structure is essential for practitioners since small deviations in the current short rates can result in significant differences in the value of the derivatives, see \cite{brigoMercurio2006}. Moreover, the method generally exhibits a second-order convergence rate; see, for instance, \cite{li2018error} and \cite{zhang2019analysis} for a one-dimensional setting, and \cite{ma2022convergence} for a two-dimensional setting. In this paper, we develop an easy and efficient algorithm to calibrate the approximated model to the current market term structure of interest rates under general one-dimensional time-inhomogeneous short-rate processes. A closed-form matrix expression is also obtained for the price of bond options (with coupons). 

Debt securities often include embedded options such as call and put options, also known as redeemable and retractable bonds, respectively. A callable bond grants the issuer the right to pay back the bond at a predetermined price in the future (the call or the strike price). This type of provision protects the issuer and reduces the value of the bond. 
Another type of common embedded option in a bond is a put option. A putable bond grants the bondholder the right to sell back the bond to the issuer at a predetermined price (the put or the strike price). These types of options can usually be exercised at any time during a given exercise period (exercise window), which makes them similar to American options with time-dependent strikes. 
A bond can have both a call and put options embedded simultaneously.
Because of the American style of the embedded options, the valuation of callable/putable debt does not admit a closed-form expression; thus, numerical procedures are necessary to price them.  

Classical numerical techniques\footnote{In this paper, we do not consider advance notices, that is, exercise decisions prior to exercise benefits, which complexify the problem significantly, see for instance \cite{buttler1996pricing}, \cite{dHalluin2001numerical}, \cite{ben2007dynamic} and \cite{ding2012pricing} for numerical techniques in that particular context.} for pricing interest rate derivatives include tree 
methods, numerical solutions to partial differential equations (PDEs),  
and Monte Carlo simulation, see for instance \cite{brigoMercurio2006}, Section 3.11. 
Trees are particularly interesting for valuing American-type derivatives because the continuation value can be calculated explicitly at each node of the tree. 
However, building trees for time-inhomogeneous mean reverting diffusion processes, such as those used for short-rate dynamics, is not always straightforward, and extension to higher dimensions may be difficult. The construction of the tree depends on the diffusion model selected (see among others, \cite{HoLee1986term}, \cite{blackDermanToy1990one}, \cite{blackKarasinski1991bond}, \cite{hull1994numerical}, \cite{hull1996using}, \cite{mercurio2001family} and \cite{brigoMercurio2006}, Appendix F, for a general formulation), and extra care is required for the transition probabilities to stay positive. Fitting the approximated model to the current market term structure can also be tricky \cite{hull1994numerical}, \cite{hull1996using}. Trees can also make options with path-dependent payoffs, such as Asian options, difficult to value. PDE approaches encounter similar challenges to trees in terms of flexibility in the modelization of the underlying diffusion process. They also make path-dependent payoff valuation challenging, and techniques are generally less intuitive than trees; see, for instance, \cite{duffy2013finite} for a general review of the approach for derivatives pricing.
In contrast, simulation methods can be implemented for various diffusion models and adapt easily to higher dimensions; however, this type of procedure is usually less efficient computationally than those based on trees or PDEs because the continuation value cannot be calculated explicitly and needs to be approximated. 
Different techniques are proposed in the literature to price American-type derivatives using simulation methods; see, for instance, \cite{fu2001pricing} for a comparison of these approaches, and \cite{glasserman2003monte}, Chapter 8, for a review.

In this paper, a simple and efficient algorithm for pricing callable/putable debt under general one-dimensional time-inhomogeneous short-rate processes is developed. Our methodology overcomes several of the drawbacks of other classical methods.
In particular, it is intuitive (similar to tree methods) and adapts to a wide range of diffusion processes (homogeneous as well as inhomogeneous processes). Moreover, because conditional expectations have closed-form matrix expressions, the continuation value can be calculated explicitly, making the approximation extremely efficient and accurate for the pricing of American-type derivatives. 
Using the methodology of \cite{cui2018general}, described in Section \ref{subsectCTMC_S} of this paper, the extension of these procedures to two-factor short rate models, see for instance \cite{brigoMercurio2006}, Chapter 4, is also straightforward.

Next, the pricing of CBs is considered. CBs are hybrid securities that possess features of both debt and equity. They are similar to bonds except that the investor has the right to convert the bond for a predetermined number of shares, known as the conversion ratio, of the issuing company during a certain exercise window prior to maturity. At maturity, if conversion is allowed and the bond has not been converted to shares, the holder has the right to convert the bond or receive its face value. In practice, additional features such as call and put options are also generally embedded in CBs, so numerical procedures are required to value these securities. 

Over the last 40 years, CB pricing has been studied extensively in the literature. Under the standard Black--Scholes setting, 
\cite{ingersoll1977contingent} proposes a structural approach under which the firm value is the underlying state variable.
Then follows the work of \cite{brennan1977convertible}, in which finite difference methods are used to solve a PDE. 
The structural approach has multiple drawbacks, \cite{batten2014convertible}. In particular, since the firm value is not a tradable asset, the calibration of the model may be challenging. Thus, \cite{mcconnell1986lyon} propose a reduced-form approach under which the issuing company stock price is the underlying state variable. Working in the Black--Scholes setting with a constant risk-free rate and volatility, \cite{mcconnell1986lyon} compensate for the credit risk by adding a constant (the credit spread) over the risk-free rate to discount the cash flows. Since then, multiple authors have incorporated credit risk into the valuation framework adequately. \cite{tsiveriotis1998valuing} split the bond value into two components: equity and debt. The equity part (when the debt is converted to stock) is discounted at a risk-free rate, whereas the debt part is discounted at a risky rate, where the risky rate can be deduced from the market-observed credit spreads. This approach to model credit risk is still widely used among practitioners for its simplicity and ability to incorporate the main feature of CBs with limited market information, \cite{gushchin2008pricing}. 
Following the approach of \cite{jarrow1995pricing} to model the credit risk, \cite{hung2002pricing} and \cite{chambers2007tree} use a binomial tree method and incorporate stochastic risk-free rates in the valuation model, whereas \cite{ayache2003valuation} and \cite{milanov2012binomial} incorporate default risk by modeling the stock price by a jump-diffusion process in a constant risk-free rate environment.

Diverse numerical methods have been proposed to value convertible debt, ranging from the classical tree methods (\cite{hung2002pricing},\cite{chambers2007tree}, \cite{milanov2012binomial}, among others),
 to finite-difference and finite element approaches (\cite{tsiveriotis1998valuing}, \cite{ayache2003valuation}, \cite{barone2003two}, among others) and simulation (\cite{ammann2008simulation}, \cite{batten2018pricing}).
Recently, \cite{lu2017simple} and \cite{ma2020valuation} developed a two-factor willow-tree approach to price CBs under stochastic interest rates and used the approach of \cite{jarrow1995pricing} as in \cite{hung2002pricing} and \cite{chambers2007tree} to include credit risk. On the other hand, \cite{lin2020numerically} propose a predictor-corrector scheme to solve a PDE under stochastic volatility or interest rate models, whereas \cite{lin2022pricing} use an integral approach under the Black--Scholes setting. Both \cite{lin2020numerically} and \cite{lin2022pricing} ignored credit risk in their valuation framework. 
This paper considers the CB pricing problem under general bi-dimensional time-inhomogeneous diffusion processes, where equity and risk-free rates are the two risk factors. 
 Default/credit risk is also included in the valuation model using the approach of \cite{tsiveriotis1998valuing}. An efficient algorithm to approximate the value of CBs using a two-layer continuous-time Markov chain approximation 
 is developed. When conversion is only permitted at maturity (European-style or European CBs), a closed-form matrix expression is obtained. Numerical experiments reveal that the method is highly efficient and accurate. 
 The advantage of the CTMC approximation over other classical methods resides in its ability to adapt to a wide range of time-inhomogeneous bi-dimensional models while preserving the simplicity of one-dimensional valuation models. The method is intuitive (similar to trees), and it is worth reiterating that it allows for the perfect fit of the current market term structure in a straightforward manner, regardless of the short-rate diffusion process selected.   

The CB pricing problem is also studied from a theoretical perspective.
In particular, when there is no credit risk, no dividends and other features such as call and put options are ignored, we show that early conversion is sub-optimal such that the problem is reduced to the pricing of a European-style CB. 
This result also holds for coupon-bearing convertible debt. On the other hand, when credit risk is considered, we show that the value of American-style CBs\footnote{The term American-style CB (or just CB) is used to refer to a bond under which the conversion option can be exercised at any time prior to maturity.} is bounded from below and above by those of European--style CBs with and without credit risk, respectively. 

Finally, numerical experiments demonstrate the high level of accuracy of CTMC methods across a large range of model parameters and short-rate models. 
The efficiency and numerical convergence of the CTMC methodology in pricing debt securities are also studied empirically, and theoretical convergence is discussed.  

The main contributions of this paper are as follows:
\vspace{-3pt}
\begin{enumerate}
    \item 
    This paper extends the results of \cite{cui2018general} to a time-inhomogeneous framework for the pricing of debt securities, such as callable, putable, and convertible bonds.   
    \item A closed-form matrix expression is obtained to approximate the price of bonds under general time-inhomogeneous short-rate processes. 
    The availability of a closed-form expression to approximate the price of zero-coupon bonds, regardless of the complexity of the short-rate dynamics selected, makes the method attractive for practitioners since it allows to perfectly calibrate the approximated model to the current market term structure. 
    \item A closed-form matrix expression is also obtained for the price of bond (with or without coupons) options under general time-inhomogeneous short-rate processes, providing an alternative approximation formula to \cite{kirkby2022hybrid}, Proposition 12, which involved an integral to be solved numerically. 
    \item Efficient procedures are developed to approximate the price of convertible debt under general bi-dimensional time-inhomogeneous diffusion processes, and a closed-form matrix expression is obtained for the price of European-style CBs. 
    \item The pricing of convertible bonds is also considered from a theoretical perspective. When there is no credit risk and no dividend yield, we show that early conversion is sub-optimal. When credit risk is considered, lower and upper bounds are obtained. 
\end{enumerate} The remainder of the paper is organized as follows. In Section \ref{sectFinancialSett}, we introduce the market model. A brief introduction to CTMC approximation methods for two-dimensional time-inhomogeneous diffusion processes is provided in Section \ref{sectNonHomoCTMC}. In Section \ref{sectionApplicationDebtSecurities}, CTMC methods are used to approximate the price of bonds, bond options, and debt securities such as callable and putable bonds under general one-dimensional time-inhomogeneous short-rate processes. Section \ref{sectionConvBond} discusses the application of CTMC approximation to convertible debt valuation under bi-dimensional time-inhomogeneous diffusion processes. Section \ref{sectNumResults} provides numerical results and shows the high efficiency of CTMC methods over other common numerical techniques. Section \ref{sectConclu} concludes the paper. 
	\section{Financial Setting}\label{sectFinancialSett}
	\subsection{Market Model}\label{sectionMarketModel}
	We consider a filtered probability space $(\Omega,\,\mathcal{F},\,\mathbb{F},\prob{Q})$, where $\mathbb F$ denotes a complete and right-continuous filtration and where
	$\prob{Q}$ denotes the pricing measure for our market. We consider a stochastic short-rate process $R$ correlated with the price of a risky asset (or stock) $S$, which can be described by a two-dimensional process $(S, R)=\{(S_t, R_t)\}_{t\geq 0}$ satisfying 
    \begin{equation}
        \begin{aligned}
         \diff S_t &=(R_t-q_t) S_t\diff t+\sigma_S(R_t)S_t \diff W^{(1)}_t,\\
	   \diff R_t &=\mu_R(t,R_t)\diff t+\sigma_R(R_t) \diff W^{(2)}_t,\label{eqEDS_S}
         \end{aligned}
	\end{equation}
    with $S_0>0$ and $R_0\in\mathcal{S}_R$, where $\mathcal{S}_R$ denotes the state-space of $R$ (generally $\reals$ or $\reals_+$\footnote{Throughout this paper, $\reals_+^\star$ denotes the strictly positive real numbers, 
    and $\reals_+$ denotes the non-negative real numbers.} 
    depending on the model, see Tables \ref{tblmodelsHomo} and \ref{tblmodelsNonHomo} for details), $q:\reals_+\rightarrow [0,1]$ is a continuous function representing a time-deterministic dividend yield, and $W=\{(W^{(1)}_t, W^{(2)}_t)\}_{t\geq0}$ is a two-dimensional correlated Brownian motion with cross-variation $[W^{(1)},W^{(2)}]_t=\rho t$, with $\rho\in[-1,1]$. 
    We assume that $\mu_R: \reals_+\times \mathcal{S}_R\rightarrow\reals$ is continuous and that $\sigma_R, \sigma_S:\mathcal{S}_R\rightarrow \reals_+$ are continuously differentiable with $\sigma_R(\cdot),\sigma_S(\cdot)>0$ on $\mathcal{S}_R$. 
	Further, we suppose that $\mu_R$, $\sigma_R$ and $\sigma_S$ are defined such that \eqref{eqEDS_S} has a unique-in-law weak solution. 
 
    The function $\sigma_S$ is often set to a constant $\sigma_S(\cdot)=\sigma>0$ ( \cite{lu2017simple}, \cite{ma2020valuation}, \cite{kirkby2022hybrid}), such that the stock price follows a geometric Brownian motion with stochastic interest rate. A list of common short-rate models is provided in Tables \ref{tblmodelsHomo}, \ref{tblmodelsNonHomo}, and \ref{tblmodelsNonHomoBrigo}.
	\begin{table}[h]
		\centering
        \resizebox{\columnwidth}{!}{
		\begin{tabular}{c c c}
        \hline
		\multicolumn{3}{c}{\textbf{Time-Homogeneous Models}}\\
			\hline
			\textbf{Model} &\textbf{Dynamics} & \textbf{Parameters}\\
		\hline
		\hline
       Vasicek, \cite{vasicek1977equilibrium} & $\diff R_t=\kappa(\theta -R_t) \diff t+\sigma \diff W_t$ & $\kappa,\theta,\sigma>0$, $R_0\in\reals$\\
       & & \\
       Cox--Ingersoll--Ross (CIR),& \multirow{2}{*}{ $\diff R_t = \kappa(\theta- R_t)\diff t +\sigma \sqrt{R_t}\diff W_t$} & $\kappa,\theta,\sigma, R_0>0$,\\
       \cite{CIR1985} & & with $2\kappa\theta>\sigma^2$\\
        & &\\
        Dothan, \cite{dothan1978term} &  $\diff R_t=\kappa R_t \diff t+\sigma R_t \diff W_t$ & $\sigma$, $\kappa \in \reals$, $R_0>0$\\
         & & \\
         Exponential Vasicek (EV) & \multirow{2}{*}{ $\diff R_t = R_t(\eta- \alpha \ln R_t)\diff t +\sigma R_t\diff W_t$} & \multirow{2}{*}{$\eta,\alpha,\sigma, R_0>0$}\\
         \cite{brigoMercurio2006}, Section 3.2.5 & &\\
         \hline
       \end{tabular}}
		\caption{Example of time-homogeneous short-rate models}
		\label{tblmodelsHomo}
    	\end{table}
	\begin{table}[h!]
		\centering
  \resizebox{\columnwidth}{!}{
		\begin{tabular}{c c c}
        \hline
		\multicolumn{3}{c}{\textbf{Time-Inhomogeneous Models}}\\
			\hline
			\textbf{Model} &\textbf{Dynamics} & \textbf{Parameters}\\
		\hline
		\hline
        Ho--Lee (HL),& $\diff R_t=\theta(t) \diff t+\sigma \diff W_t$& $\sigma>0$, $R_0\in \reals$\\
        \cite{HoLee1986term}& & \\
        &&\\
        Black--Derman--Toy (BDT), &\multirow{2}{*}{ $\diff R_t=\theta(t) R_t\diff t+\sigma R_t \diff W_t$} & \multirow{2}{*}{$\sigma>0$, $R_0>0$}\\
        \cite{blackDermanToy1990one} & & \\
        & & \\
       Hull--White (HW), & \multirow{2}{*}{$\diff R_t=\left(\theta(t) -\kappa R_t \right)\diff t+\sigma \diff W_t$} & $\kappa,\sigma>0$, and\\
        \cite{hull1990pricing}& & $R_0\in\reals$ \\
        & & \\
        Black--Karasinski (BK), & \multirow{2}{*}{$\diff R_t=R_t\left(\theta(t) -\kappa \ln R_t\right) \diff t+\sigma R_t\diff W_t$} & \multirow{2}{*}{$\kappa,\sigma, R_0>0$}\\
        \cite{blackKarasinski1991bond}& & \\
        & & \\
        Mercurio and Moraleda (MM), & $\diff R_t=R_t\left[\theta(t)-\left(\lambda-\frac{\gamma}{1+\gamma t}\ln R_t\right)\right]\diff t+\sigma R_t\diff W_t$ & $\lambda,\gamma\in\reals_+,\text{ and}$\\
        \cite{mercurio2001family} & &$\sigma,R_0>0$\\
        & & \\
         Extended CIR (CIR+), & $ \multirow{2}{*}{$\diff R_t=\theta(t) -\kappa R_t \diff t+\sigma\sqrt{R_t} \diff W_t$}$
            
      & $\kappa,\sigma, R_0>0$\tablefootnote{\label{tblFootnote}The short rates are positive under some conditions on function $\theta$.} \\
        \cite{hull1990pricing} &  & \\
        & & \\
        \hline
       \end{tabular}}
		\caption{Example of time-inhomogeneous short-rate models}
		\label{tblmodelsNonHomo}
\end{table}
        \begin{table}[h!]
		\centering
  \resizebox{\columnwidth}{!}{
		\begin{tabular}{c c c}
        \hline
		\multicolumn{3}{c}{\textbf{Time-Inhomogeneous Models}}\\
			\hline
			\textbf{Model} &\textbf{Dynamics} & \textbf{Parameters}\\
		\hline
		\hline
         Extended Vasicek, (EV+) & $\diff Y_t = \kappa(\alpha-Y_t)\diff t +\sigma \diff W_t$    
      & $\kappa,\alpha,\sigma>0$,  \\
        \cite{brigoMercurio2006}, Section 3.8.4& $R_t=Y_t+\theta(t)$  & $Y_0,R_0\in\reals$\\
        &&\\
          Extended CIR, (CIR++) & $\diff Y_t=\kappa(\alpha-Y_t)\diff t+\sigma\sqrt{Y_t} \diff W_t$    
      & $\kappa,\alpha,\sigma>0$, \\
        \cite{brigoMercurio2006}, Section 3.9& $R_t=Y_t+\theta(t)$  & $Y_0,R_0>0$\footref{tblFootnote}\\
        &&\\
         Extended Exponential Vasicek, (EEV+) & $\diff Y_t = Y_t(\eta- \alpha \ln Y_t)\diff t +\sigma Y_t\diff W_t$    
      & $\eta,\alpha,\sigma\in \reals$, \\
         \cite{brigoMercurio2006}, Section 3.8& $R_t=Y_t+\theta(t)$  & $Y_0,R_0>0$\footref{tblFootnote}\\ 
        \hline
       \end{tabular}}
		\caption{Extended time-homogeneous models of \cite{brigoMercurio2006}, Section 3.8}
		\label{tblmodelsNonHomoBrigo}
    	\end{table}
 
The time-homogeneous models listed in Table \ref{tblmodelsHomo} are popular because of their analytical tractability. However, they are less used by practitioners because they cannot adequately replicate the term structure of interest rates. Indeed, to be able to capture the discount curve appropriately, models need to have at least one time-dependent parameter. This important feature of interest rate dynamics gives rise to the time-inhomogeneous models; see Table \ref{tblmodelsNonHomo} for examples. In these models, the yield curve is provided exogenously (as an input to the model).
The extended models of \cite{brigoMercurio2006}, listed in Table \ref{tblmodelsNonHomoBrigo}, allow fitting of the initial term structure and reproduce important stylized facts while preserving the analytical tractability of the model through the auxiliary homogeneous process $Y$.

In Tables \ref{tblmodelsNonHomo} and \ref{tblmodelsNonHomoBrigo}, the initial term structure of interest rates is captured through the time deterministic function $\theta$. When using CTMC approximation, an easy recursive procedure is used to find the function $\theta$ that makes the approximated models fit the initial discount curve perfectly. This will be discussed further in Section \ref{subsecFitTermStruc}.
\section{Continuous-Time Markov Chain Approximation (CTMC) of Nonhomogeneous Processes}\label{sectNonHomoCTMC}
The CTMC framework outlined in this section was first proposed by \cite{cui2018general} for the pricing of exotic equity options under general stochastic local volatility models. Subsequently, the technique has been extended to time-inhomogeneous processes by \cite{ding2021markov}.
\subsection{Approximation of the Short-Rate Process $\{R_t\}_{t\geq0}$}	\label{subsectCTMC_R}
\begin{sloppypar}
 The objective is to construct a continuous-time Markov chain $\{R_t^{(m)}\}_{t\geq 0}$ taking values on a finite state-space ${\stateSpace_R^{(m)}:=\{r_1,r_2,\dots\, r_{m}\}}$, with $r_i\in\mathcal{S}_R$, $m\in\mathbb{N}$ and a time-dependent generator ${\tilde{\mathbf{Q}}^{(m)}(t)=[\tilde{q}_{ij}(t)]_{m\times m}}$ that converges weakly to $\{R_t\}_{t\geq 0}$ as $N,m\rightarrow\infty$. To denote the weak convergence of $R^{(m)}$ to $R$, we write $R^{(m)}\Rightarrow R$.   
\end{sloppypar}
 The first step is to approximate the state-space of the short-rate process. Several approaches are available in the literature to construct the finite state-space $\stateSpace_R^{(m)}$ of $R^{(m)}$, from simple uniform to non-uniform grids (see for instance \cite{tavellapricing}, \cite{mijatovic2013continuously}, and \cite{lo2014improved} for examples of non-uniform grids). The specific grid choice for this work is discussed in greater detail in Section \ref{sectNumResults}. 

   The next step is to construct the time-dependent generator $\tilde{\mathbf{Q}}^{(m)}(t)$. 
   For analytical tractability, we suppose further that $\tilde{\mathbf{Q}}^{(m)}(t)$ is piecewise constant in time\footnote{This assumption ensures that the transition probability matrix has a simple expression in terms of its generator, see \cite{rindos1995exact} Equation (8.4) for the general formulation.}, that is, 
\begin{equation}
	\tilde{\mathbf{Q}}^{(m)}(t)=\sum_{n=1}^{N}\mathbf{Q}_n^{(m)}\ind_{[t_{n-1},t_n)}(t). \label{eqQform} 
\end{equation}
for some time partition ${0=t_0<t_1<\ldots<t_N=T}$ of $[0,T]$, with $T>0$, $t_n:=n\Delta_N$ and $\Delta_N=T/N$, and where $\mathbf{Q}_n^{(m)}=[q_{ij}^{(n)}]_{m\times m}$ denotes the generator on the time interval $[t_{n-1},t_n)$, whose elements $q_{ij}^{(n)}$, $ 1\leq i,j\leq m$, satisfy $q_{ij}^{(n)}\geq 0$ when $i\neq j$, and $q_{ij}^{(n)}\leq 0$ when $i=j$, $n=1,2,\ldots N$. Under this assumption, the transition probability matrix $\mathbf{P}(s,t)$ from time $s$ to $t$ has the following matrix representation
\begin{equation}
	\mathbf{P}(s,t)=e^{\mathbf{Q}_{i+1}^{(m)}(t_{i+1}-s)}e^{\mathbf{Q}_{i+2}^{(m)}(t_{i+2}-t_{i+1})}\cdots e^{\mathbf{Q}_{j+1}^{(m)}(t-t_j)},\quad t_{i}\leq s<t_{i+1},\, t_{j}\leq t<t_{j+1},\, s< t,\label{eqTransProb}
\end{equation}
where
	\begin{equation}
	\exp\{\mathbf{Q}_n^{(m)} t\}=\sum_{k=0}^\infty \frac{(\mathbf{Q}_n^{(m)} t)^k}{k!}, \quad  0 \leq t\leq T,\label{EqP}
	\end{equation}
	see \cite{rindos1995exact}, p.123--124 for details.
 
Following the work of \cite{lo2014improved} and \cite{ding2021markov}, the generator $\mathbf{Q}_n^{(m)}=[q_{ij}^{(n)}]_{m\times m}$ on the time interval $[t_{n-1},t_n)$ is constructed as follows
	\begin{equation}\label{eq:generatorR}
		q_{ij}^{(n)}=\begin{cases}
			\frac{\sigma_R^2(r_i)-\delta_i\mu_R (t_{n-1}, r_i)}{\delta_{i-1}(\delta_{i-1}+\delta_i)}, & \text{$j=i-1$},\\
			-q_{i,i-1}^{(n)}-q_{i,i+1}^{(n)}, & \text{$j=i$},\\
			\frac{\sigma_R^2(r_i)+\delta_{i-1}\mu_R(t_{n-1}, r_i)}{\delta_{i}(\delta_{i-1}+\delta_i)}, & \text{$j=i+1$},\\
			0,&\text{$j\neq i,\,i-1,\,i+1$},
		\end{cases}
	\end{equation}
 \begin{sloppypar}
   for ${n\in{1,2,\ldots,N}}$, ${2\leq i\leq m-1}$, ${1\leq j\leq m}$, and where ${\delta_i=r_{i+1}-r_i}$, ${i=1,2,\ldots, m-1}$. On the borders, we set ${q_{12}^{(n)}=\frac{|\mu_R(t_{n-1}, r_1)|}{\delta_1}}$, ${q_{11}^{(n)}=-q_{12}^{(n)}}$, ${q_{m,m-1}^{(n)}=\frac{|\mu_R(t_{n-1}, r_m)|}{\delta_{m-1}}}$, ${q_{m,m}^{(n)}=-q_{m,m-1}^{(n)}}$, and $0$ elsewhere. At endpoints, other schemes could have also been employed, see \cite{chourdakis2004non} and \cite{mijatovic2013continuously} for examples. However, we note that these schemes are equivalent numerically.
 \end{sloppypar}
\begin{remark}[Weak convergence of the approximation]\label{remakWeakConvergenceR}
		Such a construction ensures that the process $R^{(m)}$ converges weakly to $R$ as $N,m\rightarrow\infty$, 
        see \cite{mijatovic2009continuously}, Section 4 and Corollary 2 for details, and \cite{ding2021markov}, Section 2.1.
  \end{remark}
\subsection{Approximation of the Stock Process $\{S_t\}_{t\geq0}$}\label{subsectCTMC_S}
The idea behind the CTMC approximation of two-dimensional processes is similar to the approximation of one-dimensional processes. The first step is to replace the CTMC approximation of the short-rate process $R^{(m)}$ into an auxiliary process \eqref{eqEDS_X2} obtained from a transformation of \eqref{eqEDS_S}. This results in a regime-switching diffusion process. The regime-switching diffusion process is then approximated by a regime-switching CTMC. The final step and the key to the approximation consists of mapping the two-dimensional regime-switching CTMC onto a one-dimensional CTMC defined on an enlarged state-space. Thus, working with an approximation of two-dimensional processes is similar to working with an approximation of one-dimensional processes.

The next lemma allows the removal of the correlation between the two Brownian motions in \eqref{eqEDS_S}.
\begin{lemma}[\cite{cui2018general}, Lemma 1]\label{lemmaStoX}
\begin{sloppypar}
        Let $(S,R)$ be defined as in \eqref{eqEDS_S}. Define $f(r):=\int_{\cdot}^r \frac{\sigma_S(u)}{\sigma_R(u)} \diff u$, and $X_t:=\ln (S_t)-\rho f(R_t)$ for $t\geq 0$. Then, $X$ satisfies
    \begin{equation}
    \begin{split}\label{eqEDS_X2}
        \diff X_t & =\mu_X(t,R_t)\diff t+\sigma_X(R_t)\diff W^\star_t\\
        \diff R_t & = \mu_R(t,R_t)\diff t+\sigma_R(R_t)\diff W^{(2)}_t,
    \end{split}
    \end{equation}
    where $W^\star:=\frac{W^{(1)}_t-\rho W^{(2)}_t}{\sqrt{1-\rho^2}}$ denotes a standard Brownian motion independent of $W^{(2)}$,
    ${\sigma_X(r):=\sigma_S(r)\sqrt{1-\rho^2}}$ and ${\mu_X(t,r):=r-q_t-\frac{\sigma_S^2(r)}{2}-\rho\psi(t,r)}$, and
    $$\psi(t,r):=\mu_R(t,r)\frac{\sigma_S(r)}{\sigma_R(r)}+\frac{1}{2}\left[\sigma_S'(r)\sigma_R(r)-\sigma_S(r)\sigma_R'(r)\right],$$
    for $r\in\mathcal{S}_R$.
\end{sloppypar}
\end{lemma}
By replacing the short-rate process in \eqref{eqEDS_X2} by its CTMC approximation $R^{(m)}$, we obtain the following regime-switching diffusion process $\{X^{(m)}_t\}_{t\geq 0}$ satisfying
\begin{equation}\label{eqEDS_Xm}
    \diff X^{(m)}_t=\mu_X(t,R^{(m)}_t) \diff t+ \sigma_X(R^{(m)}_t) \diff W^\star_t,
\end{equation}
\begin{sloppypar}
    where the regimes correspond to the state of the approximated short rate process, ${\mathcal{S}_R^{(m)}:=\{r_1,r_2,\ldots,r_m\}}$.

The regime-switching diffusion process $(X^{(m)},R^{(m)})$ is then approximated by a regime-switching CTMC $(X^{(m,M)}, R^{(m)})$. This is done by fixing a state for the short-rate process $R^{(m)}$ and then constructing a CTMC approximation for $X^{(m)}$ given $R^{(m)}$ is in that state. For this step, the same procedure as that described in Section \ref{subsectCTMC_R} can be used. More precisely, let $X^{(m,M)}$ be the CTMC approximation of $X^{(m)}$ taking values on a finite state-space ${\mathcal{S}_X:=\{x_1,x_2,\ldots,x_M\}}$, $M\in\mathbb{N}$. 

For each $r_k\in\mathcal{S}_R^{(m)}$, we define the time-dependant generator $\mathbf{\tilde{\Lambda}}^{(N,M)}_k(t):=\sum_{n=1}^N \mathbf{\Lambda}^{(n,M)}_{k}\ind_{[t_{n-1},t_n)}(t),$ where
$\mathbf{\Lambda}^{(n,M)}_{k}=[\lambda_{ij}^{(n,k)}]_{M\times M}$ and
\end{sloppypar}
\begin{equation}\label{eq:generatorX}
		\lambda_{ij}^{(n,k)}=\begin{cases}
			\frac{\sigma_X^2(r_k)-\delta_i^X\mu_X (t_{n-1}, r_k)}{\delta^X_{i-1}(\delta^X_{i-1}+\delta_i^X)} & \text{$j=i-1$}\\
			-\lambda_{i,i-1}^{(n,k)}-\lambda_{i,i+1}^{(n,k)} & \text{$j=i$}\\
			\frac{\sigma_X^2(r_k)+\delta^X_{i-1}\mu_X(t_{n-1}, r_k)}{\delta^X_{i}(\delta^X_{i-1}+\delta^X_i)} & \text{$j=i+1$}\\
			0&\text{$j\neq i,\,i-1,\,i+1$},
		\end{cases}
\end{equation}
for $2\leq i\leq M-1$, $1\leq j\leq M$, where $\delta^X_{i}=x_{i+1}-x_i$, $i=1,2, \ldots, M-1$. On the boundaries, we set
$\lambda_{12}^{(n,k)}=\frac{|\mu_X(t_{n-1},r_k)|}{\delta_1^X}$, $\lambda_{11}^{(n,k)}=-\lambda_{12}^{(n,k)}$, $\lambda_{M,M-1}^{(n,k)}=\frac{|\mu_X(t_{n-1},r_k)|}{\delta_{M-1}^X}$, $\lambda_{M,M}^{(n,k)}=-\lambda_{M,M-1}^{(n,k)}$, and $0$ elsewhere.

Using the regime-switching approximation of $(X, R)$ and the relation between $X$ and $S$ provided in Lemma \ref{lemmaStoX}, the CTMC approximation of the stock process $S^{(m,M)}$ is defined by
\begin{equation}\label{eq:SapproxCTMC}
    S^{(m,M)}_t:=\exp{\left\{X^{(m,M)}_t+\rho f(R_t^{(m)})\right\}},\quad t\geq 0.
\end{equation}
The final step consists in mapping the two-dimensional regime-switching CTMC onto a one-dimensional CTMC process $Z^{(mM)}$ on an enlarged state-space $\mathcal{S}^{(mM)}_Z:=\{1,2,\ldots, mM\}$. This is done in Proposition \ref{propCaiCTMC}, reproduced from Proposition 2 of \cite{ding2021markov}. 

\begin{prop}\label{propCaiCTMC}[Proposition 2 of \cite{ding2021markov}]
\begin{sloppypar}
    Consider a regime-switching CTMC $(X^{(m,M)}, R^{(m)})$ taking values in $\mathcal{S}_X^{(M)}\times\mathcal{S}_R^{(m)}$, where $\mathcal{S}_X^{(M)}=\{x_1,x_2,\ldots,x_M\}$ and $\mathcal{S}_R^{(m)}=\{r_1,r_2,\ldots,r_m\}$, and another one-dimensional CTMC, $\{Z_t^{(mM)}\}_{t\geq 0}$, taking values in $\mathcal{S}_Z^{(mM)}:=\{1,2,\ldots, mN\}$ and its time-dependent generator defined by $\mathbf{\tilde{G}}^{(mM)}(t):=\sum_{n=1}^N \mathbf{G}^{(mM)}_{n}\ind_{[t_{n-1},t_n)}(t),$ where
\end{sloppypar}
\begin{equation}
	\mathbf{G}^{(mM)}_n:=\left(\begin{array}{cccc}
		q_{11}^{(n)}\mathbf{I}_M+\mathbf{\Lambda}_{1}^{(n,M)} & q_{12}^{(n)}\mathbf{I}_M & \cdots & q_{1m}^{(n)}\mathbf{I}_M\\
		q_{21}^{(n)}\mathbf{I}_M & q_{22}^{(n)}\mathbf{I}_M+\mathbf{\Lambda}_{2}^{(n,M)} & \cdots & q_{2m}^{(n)}\mathbf{I}_M\\
		\vdots & \vdots & \ddots & \vdots\\
		q_{m1}^{(n)}\mathbf{I}_M & q_{m2}^{(n)}\mathbf{I}_M & \cdots & q_{mm}^{(n)}\mathbf{I}_M+\mathbf{\Lambda}_{m}^{(n,M)}\\
	\end{array}\right),\label{eq:GeneratorY}
\end{equation}
\begin{sloppypar}
    where $\mathbf{I}_M$ denotes the $M\times M$ identity matrix, ${\mathbf{Q}^{(m)}_n=[q_{ij}^{(n)}]_{m\times m}}$ and ${\mathbf{\Lambda}_{k}^{(n,M)}=[\lambda_{ij}^{(n,k)}]_{M\times M}}$, ${k=1,2,\ldots, m}$, ${n=1,2,\ldots, N}$ denote the generators defined in \eqref{eq:generatorR} and \eqref{eq:generatorX}, respectively. Define the function $\psi:\mathcal{S}_X^{(M)}\times\mathcal{S}_R^{(m)}\rightarrow \mathcal{S}_Z^{(mM)}$ by $\psi(x_l,r_k)=(k-1)M+l$ and its inverse $\psi^{-1}:\mathcal{S}_Z^{(mM)}\rightarrow \mathcal{S}_X^{(M)}\times \mathcal{S}_R^{(m)}$ by $\psi^{-1}(n_z)= (x_l,r_k)$ for $n_z\in\mathcal{S}_Z^{(mM)}$, and $k=\lceil n_z/M\rceil$, $l=n_z-(k-1)M$, where $\lceil x\rceil$ denotes the largest integer less than $x$. Then, we have
$$
\e\left[\Psi(X^{(m,N)}, R^{(m)})\Big|X_0^{(m,N)}=x_i,R_0^{(m)}=r_j\right]=\e\left[\Psi(\psi^{-1}(Z^{(mN)}))\Big|Z_0^{(mN)}=(j-1)M+i\right],
$$
for any path-dependent function $\Psi$ such that the expectation on the left-hand side is finite.
\end{sloppypar}
\end{prop}

\begin{remark}[Weak convergence of the approximation]\label{remakWeakConvergenceS}
\begin{sloppypar}
		Such a construction ensures that ${(X^{(m,M)},R^{(m)}) \Rightarrow (X,R)}$ and ${S^{(m,M)}\Rightarrow S}$ as $N,m,M\rightarrow\infty$,
        see \cite{ding2021markov}, Proposition 3, for details.
\end{sloppypar} 
  \end{remark}
\section{Application to the Pricing of Interest Rate Securities}\label{sectionApplicationDebtSecurities}
This section provides closed-form matrix expressions for the prices of zero-coupon bonds and European bond options. We also develop an efficient recursive procedure for the pricing of American-type financial instruments such as callable and putable bonds. Calibration to the current term structure of interest rates is also discussed. 
\begin{sloppypar}
    Let $0=t_0<t_1<\ldots<t_N =T$ be a time partition of $[0,T]$, where $T>0$ denotes the maturity of the financial instrument, $t_n=n\Delta_N$, $n=0,1,2,\ldots, N$, and $\Delta_N=T/N$, $N\in\mathbb{N}$. Recall that $R^{(m)}$ is the CTMC approximation of $R$ taking values in a finite state-space ${\mathcal{S}_R^{(m)}=\{r_1,r_2,\ldots,r_m\}}$, $m\in\mathbb{N}$, and its time dependent generator, ${\mathbf{\tilde{Q}}^{(m)}(t)=\sum_{n=1}^N \mathbf{Q}_n^{(m)}\ind_{[t_{n-1},t_n)}(t)}$, is defined in \eqref{eq:generatorR}.
\end{sloppypar}

Throughout this section, we denote by $\{\mathbf{e}_{k}\}_{k=1}^{m}$ the standard basis in $\reals^{m}$, that is, $\mathbf{e}_{k}$ represents a row vector of size $1\times m$ with a value of $1$ in the $k$-th entry and $0$ elsewhere, $\mathbf{1}_{m\times1}$ denotes an $m\times 1$ unit vector, and ${\mathbf{D}_m:=\diag(\bm{r})}$ denotes an $m\times m$ diagonal matrix with the vector $\bm{r}=(r_1,r_2,\ldots,r_m)$ on its diagonal. 

The following results often require Assumption \ref{assumpStateSpaceR} to hold. 
\begin{assump}\label{assumpStateSpaceR}
    There exist a $r^\star\in\reals$ such that $R_t\geq r^\star$ for all $t\geq 0$.
\end{assump}
This assumption restricts the state-space of the short-rate process for the discount factor to be bounded. This allows the use of some convergence theorems. Note that Vasicek, Ho--Lee, Hull--White, and EV+ models, listed in Tables \ref{tblmodelsHomo}, \ref{tblmodelsNonHomo}, or \ref{tblmodelsNonHomoBrigo}, do not satisfy this condition. However, numerical results in Section \ref{sectNumResults} and Appendix \ref{appendixSupplMatNumExperiment} (available online as supplemental material) show that the theoretical results of this section still hold for the Hull--White and Vasicek models, suggesting that Assumption \ref{assumpStateSpaceR} can be relaxed under a certain set of parameters. In that case, theoretical results must be shown on a case-by-case basis for each particular model. 

\subsection{Zero-Coupon Bond}\label{sectZeroBond}
The CTMC approximation of zero-coupon bond prices has been previously examined in the literature for time-homogeneous short-rate processes; see, for instance, \cite{kirkby2022hybrid}, Proposition 3.  In this section, we extend these findings to time-inhomogeneous processes. 
The first result, presented in Lemma \ref{lemmaZeroCouponBondApprox}, concerns the Laplace transform of some additive functions, extending Proposition 8 of \cite{cui2018general} to time-inhomogeneous processes. This result will be used thereafter to obtain a closed-form matrix expression for the price of zero-coupon bonds. 

\begin{lemma}\label{lemmaZeroCouponBondApprox}
	Consider $0=\tiT_0<\tiT_1<\ldots<\tiT_{\tN}=T$ a partition of $[0,T]$, with $\tN=kN$ for some $k\in\mathbb{N}$, $\Delta_{\tN}=T/{\tN}$ and $\tilde{t}_n=n\Delta_{\tN}$, 
      and let ${R_{t_i}^{(m)}=R_{t_i}=r_j\in \mathcal{S}_R^{(m)}}$, for some $i\in\{0,1,2,\ldots,N\}$. It holds that
    \begin{equation}
       \e\left[e^{-\sum_{n=ki+1}^{\tN} R_{\tilde{t}_n}^{(m)}\Delta_{\tN}} \big| R_{\tilde{t}_{ki}}^{(m)}=r_j\right]
 = \mathbf{e}_{j} \left(\prod_{n=i+1}^{N} \left(e^{\mathbf{Q}_{n}^{(m)}\Delta_{\tN}}e^{\mathbf{-D}\Delta_{\tN}}\right)^k\right)\mathbf{1}_{m\times1},\label{eqExpSumApprox}
  \end{equation}
  and 
  \begin{equation}
	    \e\left[e^{-\sum_{n=ki}^{{\tN}} R_{\tilde{t}_n}^{(m)}\Delta_{\tN}} \big| R_{\tilde{t}_{ki}}^{(m)}=r_j\right]
     = \mathbf{e}_{j}\left(\prod_{n=i+1}^{{N}} \left(e^{-\mathbf{D}_m\Delta_{\tN}}e^{\mathbf{Q}_{n}^{(m)}\Delta_{\tN}}\right)^k \right)e^{-\mathbf{D}_m\Delta_{\tN}}\mathbf{1}_{m\times1}.\label{eqExpSumApproxInAdvance}
	\end{equation}
\end{lemma}
The proof provided in Appendix \ref{appendixProof} is intuitive and follows essentially by noticing that \eqref{eqExpSumApprox} and \eqref{eqExpSumApproxInAdvance} are matrix representations of the conditional expectation of a function of a discrete one-dimensional random process whose conditional probabilities are given by \eqref{eqTransProb}. 

\begin{remark}\label{rmkCommutZeroCoupon}
    Using arguments similar to those of the proof of Lemma \ref{lemmaZeroCouponBondApprox}, we can also find that matrices $e^{\mathbf{D}_m}$ and $e^{\mathbf{Q}_n^{(m)}}$ commute under multiplication in \eqref{eqExpSumApproxInAdvance}. More precisely, using the notation of Lemma \ref{lemmaZeroCouponBondApprox}, we have that
    \begin{multline*}
        \e\left[e^{-\sum_{n=ki}^{\tN} R_{\tilde{t}_n}^{(m)}\Delta_{\tN}} \big| R_{\tilde{t}_{ki}}^{(m)}
        =r_j\right]= \mathbf{e}_{j} \left(\prod_{n=i+1}^{N} \left(e^{-\mathbf{D}_m\Delta_{\tN}} e^{\mathbf{Q}_{n}^{(m)}\Delta_{\tN}}\right)^{k}\right) e^{-\mathbf{D}_m\Delta_{\tN}} \times\mathbf{1}_{m\times1} \\
        =\mathbf{e}_{j} e^{-\mathbf{D}_m\Delta_{\tN}}\left(\prod_{n=i+1}^{N} \left(e^{\mathbf{Q}_{n}^{(m)}\Delta_{\tN}}e^{-\mathbf{D}_m\Delta_{\tN}}\right)^k\right)\times\mathbf{1}_{m\times1}.
    \end{multline*}
\end{remark}

Proposition \ref{propZeroCouponBondApprox} provides a closed-form matrix expression for the price of a zero-coupon bond under general time-inhomogeneous CTMCs. The result is a natural extension of Proposition 3 of \cite{kirkby2022hybrid} for time-inhomogeneous diffusion processes.

\begin{prop} \label{propZeroCouponBondApprox}
\begin{sloppypar}
    Let Assumption \ref{assumpStateSpaceR} hold. Given that ${R_{t_i}^{(m)}=R_{t_i}=r_j\in \mathcal{S}_R^{(m)}}$, for some $i\in\{0,1,2,\ldots,N\}$, the price at time $t_i$ of a zero-coupon bond with maturity $T\geq t_i$ can be approximated by
	\begin{equation}
	    P_j^{(m)}(t_i,T):=\e\left[e^{-\int_{t_i}^T R_{s}^{(m)}\diff s} \big| R_{t_i}^{(m)}=r_j\right]
     = \mathbf{e}_{j}\left(\prod_{n=i+1}^{N} e^{\left(\mathbf{Q}_{n}^{(m)}-\mathbf{D}_m\right)\Delta_N}\right) \mathbf{1}_{m\times1}.\label{eqZeroCouponApprox}
	\end{equation}
 \end{sloppypar}
\end{prop}
The proof of Proposition \ref{propZeroCouponBondApprox}, detailed in Appendix \ref{appendixProof}, relies on Lemma \ref{lemmaZeroCouponBondApprox}, the dominated convergence theorem, and the Lie product formula for the limit of matrix exponentials. For time-homogeneous models, an elegant proof can also be found in \cite{kirkby2022hybrid}. The proof presented in this paper employs straightforward and intuitive probabilistic arguments, making it applicable to both time-homogeneous and time-inhomogeneous models. However, it requires Assumption \ref{assumpStateSpaceR} to be satisfied for the use of dominated convergence, which is also necessary to ensure the convergence of the approximated price to the true price (as in Proposition 3 (iii) of \cite{kirkby2022hybrid}), as discussed in Remark \ref{rmkConvBondPrices} below. 
\begin{remark}[Convergence of zero-coupon bonds]\label{rmkConvBondPrices}
  Under the condition of Proposition \ref{propZeroCouponBondApprox}, the convergence of the approximated bond prices
  follows from the weak convergence of $R^{(m)}\Rightarrow R$, see Remark \ref{remakWeakConvergenceR}. Indeed, from there, we have that $\int_0^t R_s^{(m)}\diff s\Rightarrow \int_0^t R_s\diff s$ from Proposition 4 of \cite{cui2021efficient}. Then, the convergence of the expectation follows directly from the Portmanteau theorem and Assumption \ref{assumpStateSpaceR} since, for $x\geq r$, the function $e^{-x}$ is continuous and bounded. 

Detailed error and convergence analysis for CTMC methods applied to option pricing in one-dimensional settings are discussed in \cite{li2018error} and \cite{zhang2019analysis}. Convergence behavior and rates are discussed further in Appendix \ref{appendixNumZBpriceHomo} (available online as supplemental material).
\end{remark}

\begin{remark}\label{rmk:ConditionZeroBondPriceApprox}
    Assumption \ref{assumpStateSpaceR} ensures that the dominated convergence and Portmanteau theorems can be used in the proof Proposition \ref{propZeroCouponBondApprox} and Remark \ref{rmkConvBondPrices}, respectively. However, numerical results in Section \ref{sectNumResults} and Appendix \ref{appendixSupplMatNumExperiment} (available online as supplemental material) show these results still hold when $\mathcal{S}_R=\reals$ under some specific models and set of parameters.
\end{remark}

\subsection{Bond Option}
Proposition \ref{propCallPutInhomogeneousModels} provides an explicit closed-form matrix expression to approximate the prices of call and put options on zero-coupon bonds under general time-inhomogeneous short-rate models. To the best of the author's knowledge, this method of approximating the price of zero-coupon bond options is a novel contribution to the literature.

\begin{prop}\label{propCallPutInhomogeneousModels}
\begin{sloppypar}
    Let Assumption \ref{assumpStateSpaceR} hold. Given that ${R_{t_{n_1}}^{(m)}=R_{t_{n_1}}=r_j\in \mathcal{S}_R^{(m)}}$, for some ${n_1}\in\{0,1,2,\ldots,N\}$, the price at $t_{n_1}\geq 0 $ of a European call (resp. put) option with maturity $t_{n_2}>t_{n_1}$ on a zero-coupon bond maturing at time $T>t_{n_2}$ with strike $K>0$ can be approximated by
    \begin{equation}\label{eqCallPutZero}
        \e\left[e^{-\int_{t_{n_1}}^{t_{n_2}} R_s^{(m)}\diff s }h\left(P^{(m)}(t_{n_2},T)\right)\Big|R^{(m)}_{t_{n_1}}=r_j\right]= \mathbf{e}_{j}\left(\prod_{n=n_1+1}^{n_2} e^{\left(\mathbf{Q}^{(m)}_{n}-\mathbf{D}_m\right)\Delta_N}\right)\mathbf{H},
    \end{equation}
       where $h(x)=\max(x-K,0)$ (resp. $h(x)=\max(K-x,0)$) denotes the payoff function, ${P^{(m)}(t_{n_2},T):=\e\left[e^{-\int_{t_{n_2}}^TR_s^{(m)}\diff s}\Big|R_{t_{n_2}}^{(m)}\right]}$ denotes the approximated zero-coupon bond price at $t_{n_2}$, and $\mathbf{H}$ denotes a column vector of size $m\times 1$, whose $k$-th entry is given by 
           $h_k= h\left(P_k^{(m)}(t_{n_2},T)\right),$ 
       with $P_k^{(m)}(t_{n_2},T)$ defined in \eqref{eqZeroCouponApprox}.  
    \end{sloppypar}
\end{prop}
The proof follows using arguments similar to that of the proof of Proposition \ref{propZeroCouponBondApprox}. 

\begin{remark}[Convergence of zero-coupon bond options]
\begin{sloppypar}
 From Remark \ref{rmkConvBondPrices}, we have that ${\int_t^T R^{(m)}_s\diff s \Rightarrow \int_t^T R^{(m)}\diff s}$ for all $t\leq T$, as $m,N\rightarrow\infty$, and we conclude that ${P_j^{(m)}(t,T)\rightarrow P_j(t,T):=\e\left[e^{-\int_t^T R_s\diff s}|R_t=r_j\right]}$ for all $r_j\in\mathcal{S}_R^{(m)}$. However, this is not sufficient to prove the convergence of zero-coupon bond option prices. Consequently, one first needs to show that $P^{(m)}(t,T)=\e\left[e^{-\int_t^T R^{(m)}_s\diff s}|R^{(m)}_t\right] \Rightarrow P(t,T):=\e\left[e^{-\int_t^T R_s\diff s}|R_t\right]$, that is, weak convergence of random variables also implies weak convergence of conditional expectations. This has been studied in the context of filtering theory by \cite{goggin1994convergence}, \cite{kouritzin2005weak}, and \cite{crimaldi2005convergence}. 
 However, their results are inapplicable in that particular context since it requires finding a measure under which processes $e^{-\int_t^T R_s \diff s}$ and $R_t$ are independent. 
 An alternative way to prove the weak convergence of conditional expectations is to show that the estimation errors converge; see \cite{goggin1994convergence}, Lemma 2.2. Showing this property is, however, out of the scope of this paper. Numerical experiments in the next section demonstrate the accuracy and efficiency of Proposition \ref{propCallPutInhomogeneousModels} empirically. 

 Detailed error and convergence analysis for CTMC methods applied to option pricing in one-dimensional settings are discussed in \cite{li2018error} and \cite{zhang2019analysis}. Extensions to zero-bond option pricing are left for future investigations. Convergence behavior and rates are studied empirically in the next section.
\end{sloppypar}
\end{remark}
 \begin{remark}[Extension to coupon-bearing bonds]\label{rmkCouponBearingBond}
  The extension of \eqref{eqCallPutZero} to coupon-bearing bonds is straightforward. Indeed, suppose that the underlying bond pays a periodic coupon $\alpha>0$ at time $t_{n_2+z}<t_{n_2+2z}<\ldots<t_{n_2+\tilde{N}z}=T$, with $z=(N-n_2)/\tilde{N}$. Then, it suffices to replace $P_k^{(m)}(t_{n_2},T)$ in Proposition \ref{propCallPutInhomogeneousModels} by
  $$P_k^{(m)}(t_{n_2},T)+\sum_{n=1}^{\tilde{N}} \alpha P_k^{(m)}(t_{n_2},t_{n_2+nz}),$$
  with $P_k^{(m)}(\cdot,\cdot)$, defined in \eqref{eqZeroCouponApprox}. 
\end{remark}
\subsection{Callable/Putable Bond}\label{subsectionCallablePutableBond}
In the following sections, we develop simple and efficient algorithms for pricing callable and putable debt under general one-dimensional time-inhomogeneous short-rate processes. To the author's knowledge, this type of approximation for callable and putable bond pricing is novel in the literature.

Let $K_{t}^p, K_t^c\geq 0$ be constants representing the put and call prices (or strike prices) at time $t\leq T$, respectively, and let $F>0$ be the face value of the bond, and $T>0$ the maturity of the bond. In the following, we assume that there is no coupon. However, adjustment to coupon-bearing bonds is straightforward, and it is discussed in greater detail below. Furthermore, we denote by $\mathcal{T}_{t,T}$, the (admissible) set of all stopping times taking values on the interval $[t,T]$. 

When the put option can be exercised at any time prior to maturity, the value of a putable debt 
is equivalent to solving the following optimal stopping problem
\begin{equation}\label{eq:valueFctputableDebt}
    v_p(t,r):=\sup_{\tau\in\mathcal{T}_{t,T}}\e\left[e^{-\int_t^\tau R_u\diff u}\varphi_p(\tau)\Big|R_t=r\right], 
\end{equation}
where the reward (or payoff) function $\varphi_p:[0,T]\rightarrow \reals_+$ is defined by
\begin{equation}\label{eq:rewardFctputableDebt}
\varphi_p(t) =\left \{\begin{array}{ll}
                 K^p_t & \textrm{if } t<T,\\
                 \max (F,K^p_T ) & \textrm{if } t=T.
                \end{array}\right.
\end{equation}
On the other hand, the value of callable debt is given by
\begin{equation}\label{eq:valueFctCallableDebt}
    v_c(t,r):=\inf_{\tau\in\mathcal{T}_{t,T}}\e\left[e^{-\int_t^\tau R_u\diff u}\varphi_c(\tau)\Big|R_t=r\right], 
\end{equation}
where the reward function $\varphi_c:[0,T]\rightarrow \reals_+$ is defined by
\begin{equation}\label{eq:rewardFctCallableebt}
\varphi_p(t) =\left \{\begin{array}{ll}
                 K^c_t & \textrm{if } t<T,\\
                 \min (F,K^c_T ) & \textrm{if } t=T.
                \end{array}\right.
\end{equation}

Typically, a bond can have both a call and put options embedded; thus, the two problems in \eqref{eq:rewardFctputableDebt} and \eqref{eq:valueFctCallableDebt} need to be solved simultaneously. We denote by $v_{cp}:[0,T]\times\mathcal{S}_R\rightarrow \reals_+$ the value function of the problem when \eqref{eq:rewardFctputableDebt} and \eqref{eq:valueFctCallableDebt} are solved together.  Often, options are only exercisable during a certain period (exercise period or window). This is discussed further below. Numerical techniques are thus required to solve the problem.  Commonly used techniques, such as trees
, are based on the Bermudan\footnote{The Bermudan contract refers to a contract under which the embedded options can be exercised on a finite number of predetermined dates, whereas an American contract refers to a contract under which the embedded options can be exercised at any time from the inception to the maturity date.\label{footnoteBermContract}} approximation of $v_{cp}$ 
and the dynamic programming principle (see, for instance, \cite{lamberton1998}, Theorem 10.1.3). 
Proposition \ref{propCallableputableDebt} relies on the same ideas. 

\begin{prop}\label{propCallableputableDebt}
 Let Assumption \ref{assumpStateSpaceR} hold. The value of a callable and putable bond with maturity $T>0$ and face value $F>0$ can be approximated recursively by
\begin{equation}\label{eq:CallablePutableDebt}
    \left\{\begin{array}{lll}
        \mathbf{V}_N^{(m)} & = \max\left(\min\left(\mathbf{F},\mathbf{K}_N^c\right),\mathbf{K}_N^p\right)&  \\
         \mathbf{V}_n^{(m)}& = \max\left(\min\left(\mathbf{K}_n^c, e^{(\mathbf{Q}^{(m)}_{n+1}- \mathbf{D}_m)\Delta_N}\mathbf{V}_{n+1}^{(m)}\right),\mathbf{K}_n^p\right) & 0\leq n\leq N-1.
    \end{array}\right.
\end{equation}
for a sufficiently large $N\in\mathbb{N}$, and where $\mathbf{K}_n^a=K_{t_n}^a\mathbf{1}_{m\times 1}$, $a\in\{p,c\}$, $\mathbf{F}=F\mathbf{1}_{m\times 1}$, and the maximum (resp. minimum) is taken element by element (also known as the parallel maxima (resp. minima)). 
Specifically, given $R_0^{(m)}=R_0=r_j$, the approximated price 
of a callable and putable debt is given by
$$v^{(m)}_{cp}(0,R_0)=\mathbf{e}_j\mathbf{V}_0^{(m)}.$$
\end{prop}
\begin{sloppypar}
The proof of Proposition \ref{propCallableputableDebt} follows from the dynamic programming principle and by noting that $\exp\left(\mathbf{Q}^{(m)}_{n+1}- \mathbf{D}_m)\Delta_N\right)\mathbf{V}_{n+1}^{(m)}$ (the continuation value) is the matrix representation of the conditional expectation of a function of a discrete random variable whose conditional probability mass function is given by the transitional probability $p_{ij}(t_n,t_{n+1})$, $1\leq i,j\leq m$, with ${\mathbf{P}(t_n,t_{n+1})=[p_{ij}(t_n,t_{n+1})]_{m\times m}=\exp\left({\mathbf{Q}_{n+1}^{(m)}\Delta_N}\right)}$ as per \eqref{eqTransProb}. The remainder of the proof follows the same reasoning used in the proofs of Lemma \ref{lemmaZeroCouponBondApprox} and Proposition \ref{propZeroCouponBondApprox}, detailed in Appendix \ref{appendixProof}.
\end{sloppypar}
The accuracy of \eqref{eq:CallablePutableDebt} in pricing callable and putable debt is demonstrated numerically in Section \ref{sectNumResults}. 
To price a callable only bond, it suffices to set $K^p_{t}=0$, and for a putable only bond, one must let $K^c_{t}\rightarrow\infty$, $0\leq t\leq T$. Different exercise windows can also be incorporated using a similar logic. 

\begin{remark}[Extension to coupon-bearing bonds]\label{rmk:ExtensionCouponCallable}
Proposition \ref{propCallableputableDebt} is set up for zero-coupon debt. However, extension to coupon-bearing bonds is straightforward. 
Indeed, when a coupon $\alpha>0$ is paid $t_{n+1}$, it just needs to be discounted back at time $t_n$ with the value of the bond at $t_{n+1}$. More precisely, \eqref{eq:CallablePutableDebt} becomes
$$\begin{array}{lll}
         \mathbf{V}_n^{(m)}& = \max\left(\min\left(\mathbf{K}_n^c, e^{(\mathbf{Q}^{(m)}_{n+1}- \mathbf{D}_m)\Delta_N}\left(\mathbf{V}_{n+1}^{(m)}+\alpha\mathbf{1}_{m\times 1}\right)\right),\mathbf{K}_n^p\right)
    \end{array}$$
    for $n\in\{1,2,\ldots,N-1\}$.
\end{remark}
 The results of Lemma \ref{lemmaZeroCouponBondApprox}, and Propositions \ref{propZeroCouponBondApprox}, \ref{propCallPutInhomogeneousModels}, and \ref{propCallableputableDebt} can be simplified when the short-rate process is time-homogeneous such as for the models listed in Table \ref{tblmodelsHomo}, or when it can be expressed as a sum of an auxiliary time-homogeneous process and a deterministic function of time as for the models listed in Table \ref{tblmodelsNonHomoBrigo}. This is discussed further in Appendices \ref{appendixTimeHomoModels} and \ref{appendixBrigoModels}.

Extension of these results to two-factor short-rate models can be accomplished using the procedure of Section \ref{sectNonHomoCTMC}, along with Proposition \ref{propCaiCTMC}. This is left as future research.
\subsection{Calibration to the Initial Term Structure of Interest Rates}\label{subsecFitTermStruc}
Using the closed-form formula for the price of a zero-coupon bond in \eqref{eqZeroCouponApprox}, we can develop an efficient algorithm such that the zero-bond curve\footnote{The term zero-bond curve refers to the term structure of discount factors (or zero-coupon bonds).} of the approximated model fits the market curve. 

To do so, we choose a time partition of $[0,T]$, $0=t_0<t_1<\ldots<t_N=T$, with $N\in\mathbb{N}$, $T>0$, $t_n=n\Delta_N$, $n\leq N$ and $\Delta_N=T/N$. We suppose that there is a time deterministic function $\theta$ that appears in the drift of  \eqref{eqEDS_S} such that ${\mu_R(t,r)=\tilde{\mu}_R(\theta(t),r)}$, for $(t,r)\in [0,T]\times \mathcal{S}_R$, as it may often be the case for time-inhomogeneous short-rate models, see the models listed in Table \ref{tblmodelsNonHomo} for examples. Moreover, we assume that $\theta$ is piecewise constant in time, such that
\begin{equation}
    \theta(t)= \sum_{n=1}^N\theta_n\ind_{[t_{n-1},t_n)}(t).\label{eqFctTheta}
\end{equation}
for some $\bm{\theta}=(\theta_1,\theta_2,\ldots,\theta_N)\in \reals^N$.

Let $t\mapsto P^\star(0,t)$ represent the current market zero-bond curve.
The objective is to find the parameters $\bm{\theta}$ that make the zero-coupon bond prices under the approximated model equal to the market zero-coupon bond prices. Henceforth, we denote these calibrated parameters by $\bm{\theta}^\star$. 
Note that matrix $\mathbf{Q}_n$ in \eqref{eq:generatorR} depends on $\theta_n$ via function $\mu_R$, $n=1,2,\ldots, N$. In this subsection, we write $\mathbf{Q}_n^{(m)}(\theta_n)$ for $\mathbf{Q}^{(m)}_n$ to make this relation clearer. By inspecting \eqref{eqZeroCouponApprox}, we also note that the zero-coupon bond price at $t_1$ only depends on $\mathbf{Q}^{(m)}_1(\theta_1$), and the price at $t_2$ depends on $\mathbf{Q}^{(m)}_1(\theta_1)$ and $\mathbf{Q}_2^{(m)}(\theta_2)$; and so on. Thus, the calibrated parameters $\bm{\theta}^\star$, which make $P^{(m)}_j(0,t_i)=P^\star(0,t_i)$, $i=1,2\ldots,N$, can be obtained recursively starting from $t_1$ to $t_N$. Algorithm \ref{algoThetasCalibration} provides an efficient recursive procedure to find $\bm{\theta}^\star$. In Algorithm \ref{algoThetasCalibration}, $\mathbf{I}_{m \times m}$ denotes the identity matrix of size $m\times m$.

\begin{algorithm}[h!]
   \caption{Calibration of $\bm{\theta}$ to the Current Market Term-Structure}
	\label{algoThetasCalibration}
	\DontPrintSemicolon
	\KwInput{Let $\mathbf{Q}^{(m)}_n(\theta_n)$ be defined as in \eqref{eq:generatorR} and $t\mapsto P^\star(0,t)$ be the current market zero-bond curve, $n=1,2,\ldots,N$\;
		$N\in\mathbb{N}$, the number of time steps \;
		$\Delta_N\leftarrow T/N$, the size of a time step}
    Set $t_n=n\Delta_N$, $n=1,2,\ldots,N$\;
    Set $\mathbf{D}_m\leftarrow \diag(\bm{r})$ with $\bm{r}=(r_1,r_2,\ldots,r_m)$, $r_k\in\mathcal{S}_R^{(m)}$, $k=1,2,\ldots,m$\;
	\tcc{Calibration to the current market zero-bond curve $t\mapsto P^\star(0,t)$}
    $\mathbf{A}^\star\leftarrow\mathbf{I}_{m\times m }$ 
	\For{$n=1,\ldots N$}{
        Find $\theta_n^\star$ such that $P^\star(0,t_n)-\mathbf{e}_j \mathbf{A}^\star \times e^{(\mathbf{Q}_n^{(m)}(\theta_n)-\mathbf{D}_m)\Delta_N}\mathbf{1}_{m\times 1}=0$\;
        $\mathbf{A}^\star\leftarrow\mathbf{A}^\star \times e^{(\mathbf{Q}^{(m)}_n(\theta_n^\star)-\mathbf{D}_m)\Delta_N}$\;
	}	
	\KwRet $\{\theta_n^\star\}_{n=1}^N$\; 
\end{algorithm}
When the short-rate process can be modeled as a deterministic shift of a homogeneous process, such as the models listed in Table \ref{tblmodelsNonHomoBrigo}, the calibrated parameters $\bm{\theta}^\star$ have an explicit closed-form expression. This is discussed further in Appendix \ref{appendixBrigoThetaFit}.
\section{Application to the Pricing of Convertible Bonds}\label{sectionConvBond}
In this section, we develop efficient algorithms for pricing CBs using CTMC approximations. When the conversion feature is only permitted at maturity, a closed-form matrix expression is obtained. In this paper, we use the term European (resp. American)-style CB to refer to a CB under which the investor has the right to convert the bond at maturity only (resp. at any time prior to maturity). The use of CTMC approximation for pricing convertible debt is a novel contribution to the literature.

The frameworks outlined below consider two risk factors: equity and risk-free rate. Default/credit risk is incorporated into the model using the methodology of \cite{tsiveriotis1998valuing}. Their approach consists of splitting the debt into two components: a cash-only and an equity part. The \textit{cash-only part} consists of coupons and principal payments, whereas the \textit{equity part} consists of equity payments (when the debt is converted to stock). Each part is subject to different credit risks. Indeed, the cash-only part can be seen as a standard bond and is subject to the issuer default risk. Cash-flows are thus discounted at a risky rate. The equity part can be interpreted as an equity derivative and must thus be discounted at the risk-free rate.

In the following, we consider zero-coupon CBs since the extension to coupon-bearing convertible debt is straightforward. Indeed, when conversion can only occur at maturity (European-style), coupons can be added to the price. For American-style CBs, the procedure is similar to callable and putable bonds. That is, when a coupon is paid at time $t_{n+1}$, then it just needs to be discounted back to time $t_n$ with the value of the cash-only part of the bond at $t_{n+1}$, $0\leq n\leq N-1$. This is discussed further in Remarks \ref{rmk:EuroCBsCoupons} and \ref{rmkAmCBcoupon}. Further, we suppose that the risky rate $\{\tilde{R}_t\}_{0\leq t\leq T}$ is obtained by adding a time-deterministic credit spread, $c:[0,T]\rightarrow [0,1]$, over the risk-free rate, that is, $\tilde{R}_t=R_t+c_t$, $0\leq t\leq T$. Finally, the face value of the bonds is denoted by $F>0$, and $\eta>0$ represents the conversion ratio.
\begin{sloppypar}
Recall that $0=t_0<t_1<\ldots<t_N =T$ is a time partition of $[0,T]$, where $T>0$ denotes the maturity of the financial instrument, $t_n=n\Delta_N$, $n=0,1,2,\ldots, N$, and $\Delta_N=T/N$, $N\in\mathbb{N}$. $(X^{(m,M)},R^{(m)})$ denotes the regime-swiching CTMC approximation of $(X,R)$, see Section \ref{subsectCTMC_S}, taking values of a finite state-space $\mathcal{S}_X^{(M)}\times\mathcal{S}_R^{(m)}$ with $\mathcal{S}_X^{(M)}=\{x_1,x_2,\ldots,x_M\}$ and ${\mathcal{S}_R^{(m)}=\{r_1,r_2,\ldots,r_m\}}$, $m,M\in\mathbb{N}$. We have also defined $S^{(m,M)}$ in terms of $(X^{(m,M)},R^{(m)})$ in \eqref{eq:SapproxCTMC}, and the generator $\mathbf{G}_n^{(mM)}$ is defined in \eqref{eq:GeneratorY}, $n=1,2, \ldots, N$.
\end{sloppypar}

Throughout this section, we denote by $\{\mathbf{e}_{kl}\}_{k,l=1}^{m,M}$ the standard basis in $\reals^{mM}$, that is, $\mathbf{e}_{kl}$ represents a row vector of size $1\times mM$ with a value of 1 in the $(k-1)M+l$-th entry and $0$ elsewhere. ${\mathbf{D}_{mM}:=\diag\left(\bm{d}\right)}$ is an $mM\times mM$ diagonal matrix with vector $\bm{d}=(d_1,d_2,\ldots,d_{mM})$ on its diagonal, where ${d_{(k-1)M+l}=r_k\in\mathcal{S}_R^{(m)}}$, $k=1,2,,\ldots,m$, $l=1,2,,\ldots,M$.
\subsection{European-Style Convertible Bond}
Under the approach of \cite{tsiveriotis1998valuing}, the risk-neutral value of a European-style convertible debt, $v_e:[0,T]\times\reals_+^\star\times\mathcal{S}_R\rightarrow \reals_+$, is given by
\begin{equation}\label{eq:CBeuroTF}
\begin{split}
     v_e(t,x,r)  =\e \Big[e^{-\int_t^T R_u\diff u} \eta S_T \ind_{\{S_T\geq F/\eta\}} 
     + e^{-\int_t^T R_u+c_u\diff u} F  \ind_{\{S_T< F/\eta\}} \big| S_t=x, R_t=r\Big]. 
 \end{split}
\end{equation}
The cash-only $v_e^{CO}:[0,T]\times\reals_+^\star\times\mathcal{S}_R\rightarrow\reals_+$ and equity $v_e^{E}:[0,T]\times\reals_+^\star\times\mathcal{S}_R\rightarrow\reals_+$ parts of the debt can then be defined as 
\begin{equation}
     v_e^{E}(t,x,r)  :=\e \Big[e^{-\int_t^T R_u\diff u} \eta S_T \ind_{\{S_T\geq F/\eta\}}| S_t=x, R_t=r\Big],
\end{equation}
and
\begin{equation}
     v_e^{CO}(t,x,r)  := \e \Big[ e^{-\int_t^T R_u+c_u\diff u} F  \ind_{\{S_T< F/\eta\}} \big| S_t=x, R_t=r\Big],
\end{equation}
respectively.

Under the assumption that the volatility parameter of $S$ in \eqref{eqEDS_S} is constant, $\sigma_S(r)=\tilde{\sigma}_S>0$ for all $r\in\mathcal{S}_R$, and the short-rate process is Gaussian; we can find an explicit expression for \eqref{eq:CBeuroTF}. This is the case for the Vasicek, Ho--Lee, and Hull--White models (see Tables \ref{tblmodelsHomo} and \ref{tblmodelsNonHomo} for details). This is discussed further in Appendix \ref{appendixCBeuroTF}, available online as supplemental material.

\begin{prop}\label{prop:vanillaCB_CTMC}
    Let Assumption \ref{assumpStateSpaceR} hold. Given that $S_0>0$, and $X^{(m,M)}=\ln(S_0)-\rho f(R_0)=x_i\in\mathcal{S}_X^{(M)}$, with $R_0=r_j\in\mathcal{S}_R^{(m)}$, the value of the European-style CB with maturity $T>0$, face value $F>0$, and conversion ratio $\eta>0$ can be approximated by
    \begin{align}
        v_e^{(m,M)}(0,S_0,R_0)&:=\e \Big[e^{-\int_0^T R^{(m)}_u\diff u} \eta S_T^{(m,M)} \ind_{\{S_T^{(m,M)}\geq F/\eta\}} \nonumber\\
     & \quad\quad\quad 
     + e^{-\int_0^T R_u^{(m)}+c_u\diff u} F  \ind_{\{S_T^{(m,M)}< F/\eta\}} \big| S_0^{(m,M)}=S_0, R^{(m)}_0=R_0\Big] \nonumber \\
    & = \mathbf{e}_{ji} \prod_{n=1}^N e^{\left(\mathbf{G}^{(mM)}_n-\mathbf{D}_{mM}\right)\Delta_N}\mathbf{H}.\label{eq:vanillaCB_CTMC} 
    \end{align}
    \begin{sloppypar}
         where $\mathbf{H}$ denotes a column vector of size $mM\times 1$ whose $(k-1)M+l$-th entry is given by
    \end{sloppypar}
    \begin{equation}\label{eq:hlk_ConvBondEuroTF}
        h_{(k-1)M+l}=\eta e^{x_l+\rho f(r_k)}\ind_{\{e^{x_l+\rho f(r_k)}\geq F/\eta\}}+e^{-\int_0^Tc_u\diff u}F\ind_{\{e^{x_l+\rho f(r_k)}< F/\eta\}},
    \end{equation}
    for $k=1,2,,\ldots,m$, $l=1,2,,\ldots,M$. 
\end{prop}
\begin{sloppypar}
The proof follows by noting that \eqref{eq:vanillaCB_CTMC} is the matrix representation of the conditional expectation of a function of a discrete one-dimensional random variable whose conditional probability mass function is given by the transitional probability $p_{kl}(t_n,t_{n+1})$, $1\leq k,l\leq mM$, with ${\mathbf{P}(t_n,t_{n+1})=[p_{kl}(t_n,t_{n+1})]_{mM\times mM}}$ as defined in  \eqref{eqTransProb}, with the generators $\mathbf{Q}^{(m)}_n$ replaced by $\mathbf{G}^{(mM)}_n$. The remainder of the proof follows the same reasoning used in the proofs of Lemma \ref{lemmaZeroCouponBondApprox} and Proposition \ref{propZeroCouponBondApprox}, detailed in Appendix \ref{appendixProof}.
\end{sloppypar}
\begin{remark}[Extension to coupon-bearing European CBs]\label{rmk:EuroCBsCoupons}
To extend the valuation to coupon-bearing bonds, it suffices to add the discounted value of the future coupons to the value obtained in \eqref{eq:vanillaCB_CTMC}, similar to the approach in Remark \ref{rmkCouponBearingBond}. More precisely, assume that a periodic coupon, $\alpha>0$, is paid at times $t_{z}<t_{2z}<\ldots<t_{\tilde{N}z}=T$, with $z=N/\tilde{N}$. The present value of future coupons is given by
$$\sum_{n=1}^{\tilde{N}} \alpha P_j^{(m)}(0,t_{nz}),$$
 with $P_j^{(m)}(\cdot,\cdot)$ defined in \eqref{eqZeroCouponApprox}. Adding this to the value of the European-style CB without coupons obtained in \eqref{eq:vanillaCB_CTMC} completes the extension.
\end{remark}
\begin{remark}[Convergence of European-style CBs]
    From Remark \ref{remakWeakConvergenceS}, we have that $S^{(m,M)}\Rightarrow S$. The convergence of derivatives with a continuous and bounded payoff function then follows directly from the Portmanteau theorem, see for instance \cite{billingsley1999convergence}, Theorem 2.1. However, for discontinuous and unbounded payoff functions such as that involved in \eqref{eq:CBeuroTF}\footnote{The payoff function exhibits a discontinuity in the state variable because of the difference between the risky and the risk-free rates.}, the convergence of the prices is not as straightforward. For a continuous and unbounded payoff function $g$, \cite{mijatovic2009continuously}, Remark 3 and \cite{cui2018general}, Remark 5, suggest replacing the original payoff function by a truncated payoff $g\wedge L$, with a constant $L>0$ chosen to be sufficiently large such that the numerical results are not altered. \cite{kirkby2022hybrid}, Proposition 6, shows the convergence of the derivative prices for continuous bounded payoffs, such as equity cap and floor, that is, when the payoff function is bounded from above and below. 
    
    Detailed error and convergence analysis in the context of European option pricing under two-dimensional stochastic local volatility models are discussed in \cite{ma2022convergence}. Extensions to European-style CBs under two-dimensional stochastic interest rate models are left for future research. Numerical experiments in Section \ref{sectNumResults} demonstrate the accuracy and efficiency of the approximation empirically.
\end{remark}
The closed-form matrix expression in \eqref{eq:vanillaCB_CTMC} can be implemented in a straightforward manner. However, as highlighted by \cite{mackay2023analysis}, several numerical issues can be encountered when dealing with medium/long time-horizon derivatives because of the size of generator $\mathbf{G}^{(mM)}_n$. Hence, based on Propostion 4.3 of \cite{mackay2023analysis}, a new algorithm that speeds up the pricing of European-style CBs is developed. This fast version of Proposition \ref{prop:vanillaCB_CTMC} is presented in Appendix \ref{appendixAlgoEuroCBfast}. 
\subsection{Convertible Bond (American-Style)}
When the conversion option can be exercised at any time prior to maturity (and call and put features are ignored), the valuation of CBs is equivalent to solving the following optimal stopping problem
\begin{equation}\label{eq:valueFctCBamerican}
    v(t,x,r)=\sup_{\tau\in\mathcal{T}_{t,T}}\e\left[e^{-\int_t^\tau R_u +c_u \ind_{\left\{\tau=T,S_T< F/\eta\right\}}\diff u}\varphi(\tau,S_{\tau})\Big|S_t=x, R_t=r\right], 
\end{equation}
where $\mathcal{T}_{t,T}$ denotes the (admissible) set of all stopping times taking values on the interval $[t,T]$, and the reward (or gain) function $\varphi:[0,T]\times \reals_+^\star\rightarrow \reals_+^\star$ is defined by
\begin{equation}\label{eq:rewardFctCB}
\varphi(t,x) =\left \{\begin{array}{ll}
                 \eta x & \textrm{if } t<T,\\
                 \max (\eta x, F) & \textrm{if } t=T.
                \end{array}\right.
\end{equation}
\begin{remark}\label{rmkCBtimeDiscountinuity}
    When $x<F/\eta$, the reward function is discontinuous at $T$ since
    $$\lim_{t\rightarrow T^-}\varphi(t,x)=\eta x < F=\varphi(T,x).$$
\end{remark}
Assuming that an optimal stopping time\footnote{An admissible stopping time $\tau_t^\star\in\mathcal{T}_{t,T}$ is said to be optimal for \eqref{eq:valueFctCBamerican} if ${v(t,x,r)=\e\left[e^{-\int_t^{\tau_t^\star} R_u +c_u \ind_{\{\tau_t^\star=T,\eta S_T< F\}}\diff u}\varphi(\tau_t^\star,S_{\tau_t^\star})\Big|S_t=x, R_t=r\right]}$.} $\tau_t^\star$ exists, the cash-only part  $v^{CO}:[0,T]\times \reals_+^\star\times\mathcal{S}_R\rightarrow \reals_+$ and equity part $v^E:[0,T]\times \reals_+^\star\times\mathcal{S}_R\rightarrow \reals_+$ of the CB can be defined by
\begin{equation*}\label{eq:vCO_CB}
\begin{split}
    v^{CO}(t,x,r) &=\e\left[e^{-\int_t^{T} R_u +c_u \diff u}F\ind_{A}\Big|S_t=x, R_t=r\right], \text{ and}\\
     v^{E}(t,x,r) &= \e\left[e^{-\int_t^{\tau_t^\star} R_u \diff u}\eta S_{\tau_t^\star}\ind_{A^c}\Big|S_t=x, R_t=r\right],
\end{split}
\end{equation*}
respectively, where $A:=\left\{\tau^\star_t=T,S_T< F/\eta\right\}$, and $A^c:=\left\{\tau^\star_t< T\right\}\cup\left\{\tau^\star_t=T, S_T\geq F/\eta\right\}$ denotes the complement of $A$. It follows that $v(t,x,r)=v^{CO}(t,x,r)+v^{E}(t,x,r)$.

When no dividends are paid, and credit risk is assumed to be nil ($q_t=c_t=0$ for all $t\in[0,T]$), the value of the CB in \eqref{eq:valueFctCBamerican}, which can be exercised at any time prior to maturity, is equal to the value of the European-style CB \eqref{eq:CBeuroTF}, meaning that an optimal stopping time for \eqref{eq:valueFctCBamerican} is at the maturity of the bond. On the other hand, when credit risk is considered, the value of American-style CBs is bounded from below and above by those of European-style CBs with and without credit risk, respectively. This is formalized in the following.  
\begin{prop}\label{propCBtrivialTF}
\begin{sloppypar}
    Assume that $q_t=c_t=0$ for all $t\in[0,T]$ and $\sigma_S(r)=\tilde{\sigma}_S>0$ for all $r\in\mathcal{S}_R$. We have that ${v(t,x,r)=v_e(t,x,r)}$ for all $(t,x,r)\in [0,T]\times \reals_+^\star\times\mathcal{S}_R$. 
\end{sloppypar}
\end{prop}
\begin{proof}
We first show that the discounted reward process $\{e^{-\int_0^t R_u\diff u} \varphi(t,S_t)\}_{0\leq t\leq T}$ is a submartingale. For $0\leq s\leq t <T$, we have that
$$\e\left[e^{-\int_0^t R_u\diff u} \varphi(t,S_t) \Big|\mathcal{F}_s\right]=\e\left[e^{-\int_0^t R_u\diff u} \eta S_t|\mathcal{F}_s\right] =e^{-\int_0^s R_u\diff u} \eta S_s=e^{-\int_0^s R_u\diff u} \varphi(s,S_s),$$
where the second equality follows from the martingale property of the discounted stock process under the risk-neutral measure (see Remark \ref{rmkMartingalePropS}).
On the other hand, if $0\leq s< t =T$, we have that
\begin{align*}
    \e\left[e^{-\int_0^T R_u\diff u} \varphi(T,S_T) \Big|\mathcal{F}_s\right]&=\e\left[e^{-\int_0^T R_u\diff u}\max (\eta S_T, F) \Big|\mathcal{F}_s\right]\\
    &\geq \e\left[e^{-\int_0^T R_u\diff u}\eta S_T \Big|\mathcal{F}_s\right]
    =e^{-\int_0^s R_u\diff u} \varphi(s,S_s).
\end{align*}
The final assertion follows for well-known results in optimal stopping theory, which states that if the discounted reward process is a submartingale, then the maturity of the contract is an optimal stopping time, see \cite{bjork2009}, Proposition 21.2.
\end{proof}
When periodic coupons are paid, the results of Proposition \ref{propCBtrivialTF} still hold. This is demonstrated in Corollary \ref{corCBtrivialTFcoupon}.

Let $\tilde{v}_e:[0,T]\times\reals_+^\star\times\mathcal{S}_R\rightarrow \reals_+$ denote the value function of European-style CBs when $c_t=0$ for all $t\in[0,T]$. Using \eqref{eq:CBeuroTF}, it follows that
\begin{equation}\label{eqEuroCBnoCreditRisk}
    \begin{split}
        \tilde{v}_e(t,x,r)&:=\e \Big[e^{-\int_t^T R_u\diff u} \max\left(\eta S_T,F\right) \big| S_t=x, R_t=r\Big]\\
        &=\e \Big[e^{-\int_t^T R_u\diff u} \varphi(T,S_T) \big| S_t=x, R_t=r\Big].
    \end{split}
\end{equation}
\begin{corollary}\label{corrCBsUpperBound}
\begin{sloppypar}
        Assume that $q_t=0$ for all $t\in[0,T]$ and $\sigma_S(r)=\tilde{\sigma}_S>0$ for all $r\in\mathcal{S}_R$. We have that ${v_e(t,x,r)\leq v(t,x,r)\leq \tilde{v}_e(t,x,r)}$ for all $(t,x,r)\in[0,T]\times\reals_+^\star\times\mathcal{S}_R$.
\end{sloppypar}
\end{corollary}
\begin{proof}
    The first inequality follows directly since $T\in \mathcal{T}_{t,T}$. For the second inequality, it suffices to note that
    \begin{align*}
        v(t,x,r)& = \sup_{\tau\in\mathcal{T}_{t,T}}\e\left[e^{-\int_t^\tau R_u +c_u \ind_{\{\tau=T,\S_T<F/\eta\}}\diff u}\varphi(\tau,S_{\tau})\Big|S_t=x, R_t=r\right]\\
        &\leq \sup_{\tau\in\mathcal{T}_{t,T}}\e\left[e^{-\int_t^\tau R_u \diff u}\varphi(\tau,S_{\tau})\Big|S_t=x, R_t=r\right] = \tilde{v}_e(t,x,r)
    \end{align*}
    where the last equality follows from Proposition \ref{propCBtrivialTF}.
\end{proof}
When periodic coupons are paid, the result of Corollary \ref{corrCBsUpperBound} still holds. This is discussed in Corollary \ref{corrCBsUpperBoundcoupon}.

\begin{remark}\label{rmkMartingalePropS}
    Condition $\sigma_S(r)=\tilde{\sigma}_S>0$ for all $r\in\mathcal{S}_R$ in Propositions \ref{propCBtrivialTF} can be relaxed provided that the discounted stock process remains a true martingale under the risk-neutral measure. Indeed, when no additional condition is added, the discounted stock process $\{e^{-\int_0^t R_u\diff u} S_t\}_{t\geq 0}$ is a local martingale. For $\{e^{-\int_0^t R_u\diff u} S_t\}_{t\geq 0}$ to be a true martingale, some restrictions must be added to the parameters of the short-rate dynamics. This has been studied in stochastic volatility models for some specific time-homogeneous diffusion processes; see \cite{sin1998complications}, \cite{jourdain2004loss}, and \cite{cui2013martingale}. Conditions under which the discounted stock process is a true martingale for the particular short-rate models listed in Tables \ref{tblmodelsHomo}, \ref{tblmodelsNonHomo}, and \ref{tblmodelsNonHomoBrigo} are left as future research.

   Note also that the aforementioned results are applicable only to stocks that do not pay dividends, that is $q_t=0$ for all $t\geq 0$. When dividends are distributed, the discounted stock process becomes a supermartingale, which makes the arguments in the proof of Proposition \ref{propCBtrivialTF} invalid.
\end{remark}

When the conditions of Proposition \ref{propCBtrivialTF} are not satisfied, or the debt includes other specific features such as call and/or put options, as is often the case in practice, numerical techniques are required to solve the optimal stopping problem in \eqref{eq:valueFctCBamerican}. Commonly used methods, such as trees or least-squares Monte Carlo (see \cite{longstaff2001valuing}), are based on the Bermudan approximation\footref{footnoteBermContract} of $v$
and the dynamic programming principle. The same ideas are used in Proposition \ref{propConvBondCTMC}.

Consequently, we define $\mathbf{H}^{CO}$ (resp. $\mathbf{H}_n^{E}$, $0\leq n\leq N$) as column vectors of size $mM\times 1$ representing the cash-only (resp. equity) part of the reward, whose $(k-1)M+l$-entry are respectively given by
\begin{equation}\label{eq:defH_CO_ctmc}
    h^{CO}_{(k-1)M+l} = F\ind_{\{e^{x_l+\rho f(r_k)}< F/\eta\}},
\end{equation}
and
\begin{equation}
    h^{E}_{(k-1)M+l,n} = \left\{\begin{array}{ll} 
                                 \eta e^{x_l+\rho f( r_k)},& \textrm{ if } 0\leq n\leq N-1,\\
                                 \eta e^{x_l+\rho f(r_k)}\ind_{\{e^{x_l+\rho f(r_k)}\geq F/\eta\}}, &                                                                                                         \textrm{ if } n=N,\\ 
                                \end{array}\right.
\label{eq:defH_SO_ctmc}
\end{equation}
for $1 \leq k\leq m$, $1\leq l\leq M$.
Furthermore, let $\mathbf{B}_n^{CO}$, $\mathbf{B}^{E}_n$, and $\mathbf{B}_n$, $0\leq n\leq N$, be column vectors of size $mM\times 1$, representing the cash-only part, the equity part, and the total value of the CB at time $t_n$, respectively. 
We denote by $b_{k,n}$ the $k$-th entry of $\mathbf{B}_n$, $1\leq k\leq mM$, and define the indicator vector $\mathbf{1}_{\{\mathbf{B}_n=\mathbf{H}_n^{E}\}}$, where the $k$-th entry of $\mathbf{1}_{\{\mathbf{B}_n=\mathbf{H}_n^{E}\}}$, denoted by $\mathbf{1}_{\{\mathbf{B}_n=\mathbf{H}_n^{E}\}}(k)$, is given by
$\mathbf{1}_{\{\mathbf{B}_n=\mathbf{H}_n^{E}\}}(k)= \ind_{\{b_{k,n}=h^{E}_{k,n}\}},$
for each $k\in\mathcal{S}_Z^{mN}$. Finally, $\bm{1}_{mM\times 1}$ denotes the unit vector of size $mM\times 1$.

\begin{prop}\label{propConvBondCTMC}  Let Assumption \ref{assumpStateSpaceR} hold. The value of a CB with maturity $T>0$, face value $F>0$, and conversion ratio $\eta>0$, can be approximated recursively by
\begin{align*}
\begin{array}{llll}
    \multirow{2}{*}{$\hat{\mathbf{B}}_n^{CO}$},&\multirow{2}{*}{$=\Bigg\{$}
                                                 & \mathbf{H}^{CO}, & \textrm{ if } n=N,\\ 
                                                 &                           &  e^{-\int_{t_n}^{t_{n+1}}c_u\diff u}\exp\left\{\left(\mathbf{G}^{(mN)}_{n+1}-\mathbf{D}_{mM}\right)\Delta_N\right\}\mathbf{B}^{CO}_{n+1}, &  \textrm{ if }0\leq n\leq N-1,\\[6pt]
\multirow{2}{*}{$\mathbf{B}_n^{CO}$}&\multirow{2}{*}{$=\Bigg\{$}
                                                 &  \widehat{\mathbf{B}}_N^{CO}, & \textrm{ if } n=N,\\ 
                                    &            &  \hat{\mathbf{B}}^{CO}_n\left(\mathbf{1}_{mM\times 1}-\ind_{\{\mathbf{B}_n=\mathbf{H}_n^{E}\}} \right), & \textrm{ if }0\leq n\leq N-1,\\[6pt]
\multirow{2}{*}{$\hat{\mathbf{B}}_n^{E}$}&\multirow{2}{*}{$=\Bigg\{$}
                                                 & \mathbf{H}^{E}_N, & \textrm{ if } n=N,\\ 
                                    &            &  \exp\left\{\left(\mathbf{G}^{(mN)}_{n+1}-\mathbf{D}_{mM}\right)\Delta_N\right\}\mathbf{B}^{E}_{n+1}, & \textrm{ if }0\leq n\leq N-1,\\[6pt] 
\multirow{2}{*}{$\mathbf{B}_n^{E}$}&\multirow{2}{*}{$=\Bigg\{$}
                                                 & \hat{\mathbf{B}}_N^{E}, & \textrm{ if } n=N,\\ 
                                    &            &  \mathbf{B}_n-\mathbf{B}_n^{CO}, & \textrm{ if }0\leq n\leq N-1,
                                    \\[6pt] 
\multirow{2}{*}{$\mathbf{B}_n$}&\multirow{2}{*}{$=\Bigg\{$}
                                                 & \hat{\mathbf{B}}_N^{CO}+ \hat{\mathbf{B}}_N^{E}, & \textrm{ if } n=N,\\ 
                                    &            &  \max\left(\mathbf{H}_n^{E}, \hat{\mathbf{B}}_n^{CO}+ \hat{\mathbf{B}}_n^{E} \right) & \textrm{ if }0\leq n\leq N-1.
\end{array}
\end{align*}
\begin{sloppypar}
for a sufficiently large $N\in\mathbb{N}$, and where the maximum is taken element by element. 
Specifically, given that ${X_0^{(m,M)}=\ln(S_0)-\rho f(R_0)=x_i\in\mathcal{S}^{(M)}_X}$ and ${R_0^{(m)}=R_0=r_j\in\mathcal{S}^{(m)}_R}$, the the value of an American-style CB can be approximated by
$$v^{(m,M)}(0,S_0, R_0)=\mathbf{e}_{ji} \mathbf{B}_0^{(mM)}.$$
\end{sloppypar}
\end{prop}
Proposition \ref{propConvBondCTMC} is presented as an algorithm in Appendix \ref{appendixCallablePutableBondAlgo}. Similar to the European-style CB, the performance of the procedure in Proposition \ref{propConvBondCTMC} can be significantly increased using the technique of \cite{mackay2023analysis}, Proposition 4.3. This new fast version of the procedure is also reported in Appendix \ref{appendixCallablePutableBondAlgo}. 
\begin{remark}[Extension to coupon-bearing bonds]\label{rmkAmCBcoupon}
    Recall that Proposition \ref{propConvBondCTMC} is set up for zero-coupon CBs. However, as mentioned previously, adding coupons to the previous procedure is straightforward. Indeed, when a coupon is paid at time $t_{n+1}$, it just needs to be discounted back to time $t_n$ with the cash-only part of the CB. More precisely, when a coupon $\alpha>0$ is paid at time $t_{n+1}$, then the continuation value of the cash-only part at $t_n$, $\hat{\mathbf{B}}_n^{CO}$, must be calculated as follow
    
    $$\hat{\mathbf{B}}_n^{CO}= e^{-\int_{t_n}^{t_{n+1}}c_u\diff u}\exp\left\{\left(\mathbf{G}^{(mN)}_{n+1}-\mathbf{D}_{mM}\right)\Delta_N\right\}\left(\mathbf{B}^{CO}_{n+1}+\alpha\mathbf{1}_{mN\times 1}\right),
    $$  
    for $n\in\{0,1,\ldots,N-1\}$.
\end{remark}
Additional features, such as call and put options, can also be added, similarly as in tree methods (see, for instance, \cite{hung2002pricing}, Exhibit 2), by modifying the value of the convertible debt, the equity part, and the cash-only part accordingly at each time step.
\begin{remark}[Convergence of convertible bonds (American-style)]\label{rmk:ConvergenceCB}
    When there is no credit risk, the convergence follows as in Remark 4.8 of \cite{mackay2023analysis}. 
    Indeed, in that particular case, we can rely on the continuous-reward representation of \cite{mackay2023optimal}, Theorem 3.4, and use the results of \cite{song2013weak}, Theorem 9, to establish the convergence. When credit risk is considered, the convergence is less clear because of the discontinuity in the discounted reward process created by the difference between the risky and risk-free rates. 
    
    Detailed error and convergence analysis for American-style CBs are left for future research.
    The accuracy and efficiency of Proposition \ref{propConvBondCTMC} in pricing American-style CBs is demonstrated empirically in Section \ref{sectNumResults}. 
\end{remark}

\section{Numerical Experiments}\label{sectNumResults}
In this section, numerical experiments are conducted to analyze the performance of the methodology proposed in the previous sections under models listed in Tables \ref{tblmodelsNonHomo} and \ref{tblmodelsNonHomoBrigo}. For the testing, we selected the Hull--White model, which is widely used in practice, and the extended CIR model (CIR++) for its analytical traceability. 
    More precisely, we analyze the accuracy and efficiency\footnote{The term ``efficiency'' refers to the ratio of the computation time of a procedure to the precision of its numerical result.} of CTMC approximations in valuing different debt securities. Numerical convergence is also investigated. Note that Assumption \ref{assumpStateSpaceR} is not respected under the Hull--White model.

In Appendix \ref{appendixSupplMatNumExperiment} (available online as supplemental material), a similar analysis is performed for the Vasicek and CIR models, two time-homogeneous short-rate processes of Tables \ref{tblmodelsHomo}. The accuracy and efficiency of the CTMC methods in approximating zero-bond prices, Proposition \ref{propZeroCouponBondApprox}, are also investigated, and numerical convergence is analyzed. 
Additional examples with Dothan, exponential Vasicek, EV+, and EEV+ models are also available upon request. Results under these models are similar to those obtained under the Hull--White and CIR++ short-rate processes documented below.

All the numerical experiments are conducted with Matlab R2015a on a Core i7 desktop with 16GB RAM and a speed of 2.40 GHz. Matrix exponentials are calculated using the function fastExpm for Matlab, see \cite{matlabFastExpm}, which is designed to accelerate the calculation of large (sparse and full) matrices. Column ``CTMC'' reports the CTMC approximated value calculated using the results of Sections \ref{sectionApplicationDebtSecurities} or \ref{sectionConvBond}. Column ``Benchmark'' presents the benchmark value, column ``Abs. error'' documents the absolute error\footnote{The \textit{absolute error} is defined as the absolute value of the difference between the CTMC approximated value and the benchmark value.}, whereas column ``Rel. error'' provides the relative error. The convergence rate about the number of grid points $m$ is approximated using the following formula: 
$$\textrm{Rate}\approx \frac{\log \left(e_{m_2}/e_{m_1}\right)}{\log \left(m_1/m_2\right)},$$
where $e_m$ is the absolute error using the number of grid points $m$. Throughout this paper, log refers to the natural logarithm.

In all the following numerical experiments, the model is calibrated to the market risk-free discount curve\footnote{The market discount curve is obtained from Bloomberg and corresponds to the US Dollar curve 23 as of March 31, 2023, from the Swap Curve Builder (ICVS) page.} reported in Table \ref{tblmarketZBcurve}. The calibration to the market curve is performed using Algorithm \ref{algoThetasCalibration} for the Hull--White model and Algorithm \ref{algoThetasCalibrationBrigo} for the CIR++ model.

\begin{table}[h!]
\scalebox{0.7}{
	\begin{tabular}{c cccccccccccccccc}
		\hline
		\textbf{t} &  0.26 & 	 0.47& 	 0.72& 	 0.97 &	 1.22 &	 1.47 	& 1.72 &	2 &	3 &	4 \\
       $\mathbf{P^\star(0,t)}$ & 0.986944 &	 0.976019 &	 0.964123 	& 0.953152 &	 0.943283 &	 0.934357 &	 0.926202 &	                            0.917553 &	 0.888740 &	 0.861950\\ 
		\hline
	\end{tabular}}
	\caption{Market zero-bond curve, $t\mapsto P^\star(0,t)$.}
	\label{tblmarketZBcurve}
\end{table}
\begin{sloppypar}
For the Hull-White model, the state-space of the approximated short-rate process, ${\mathcal{S}^{(m)}_R=\{r_1,r_2,\ldots,r_m\}}$ with $m\in\mathbb{N}$, is constructed using the non-uniform grid proposed by Tavella and Randall (\cite{tavellapricing}, Chapter 5.3). That is, we first select the grid lower and upper bounds, $r_1,r_m\in \mathcal{S}_R$, and set the other grid points as follows
	$r_k=R_0+\tilde{\alpha}_R\sinh\left(c_2 \frac{k}{m}+c_1\left[1-\frac{k}{m}\right]\right),\quad k=2,\,\ldots,\,m-1,$
where $c_1=\sinh^{-1}\left(\frac{r_1-R_0}{\tilde{\alpha}}\right)$, $c_2=\sinh^{-1}\left(\frac{r_m-R_0}{\tilde{\alpha}}\right),$ and $\tilde{\alpha}_R\geq 0$, controls the degree of non-uniformity of the grid. For the CIR++ models, the same procedure is applied to the auxiliary process $Y^{(m)}$, and we denote by $\tilde{\alpha}_Y$ the grid non-uniformity parameter. Note that when $R_0$ (or $Y_0$ for the CIR++ model) is not in the grid, it is inserted (see, for instance, \cite{cui2019continuous}, Section 2.3 for details).
Unless stated otherwise, all model experiments are conducted using the model and the CTMC parameters summarized in Table \ref{tblModelandCTMCparamNonHomo}. 
\end{sloppypar}
\begin{table}[h!]
		\centering
	\begin{tabular}{ccccc | ccccc}
		\hline
		     & ${R_0}$ & $\alpha$  &${\kappa}$ & ${\sigma}$ & ${m}$  & ${r_1}$ & ${r_m}$ & $\tilde{\alpha}_R$& $\Delta_N$ \\
		\hline 
		Hull--White  & $0.04$ & N/A & $1$  & $0.20$ &$160$ & $-30 r_0$& $25 r_0$ & $0.5$ & 1/252 \\
        \hline\hline
        & ${Y_0=R_0}$ & $\alpha$ & ${\kappa}$ & ${\sigma}$ & ${m}$  & ${y_1}$ & ${y_m}$ & $\tilde{\alpha}_Y$& $\Delta_N$ \\
		\hline 
        CIR++       & 0.04 & 0.035 & 2 & 0.20 &$160$ & $y_0/100$& $7 y_0$ & $0.5$  & 1/252 \\
		\hline
\end{tabular}
\caption[Model and CTMC parameters Hull--White and CIR++]{Model and CTMC parameters}\label{tblModelandCTMCparamNonHomo}
\end{table}
\subsection{Approximation of Zero-Coupon Bond Option Prices}\label{sectionNumExpBond}
In this section, we study the accuracy and efficiency of \eqref{eqCallPutZero}, as well as the numerical convergence of the approximated zero-coupon bond option prices, under the Hull--White and CIR++ models\footnote{For the CIR++ model, \eqref{eqCallPutZero} can be greatly simplified using the time-homogeneous property of the auxiliary process $Y^{(m)}$, see \eqref{eqCallPutBrigo} in Appendix \ref{appendixBrigoModels} for details.}, respectively. Under these particular models, the price of zero-coupon bond options have a closed-form expression, which can be found in \cite{brigoMercurio2006}, Section 3, and thus, can serve as a benchmark in our analysis. 

We test the accuracy of the approximated option prices for different levels of moneyness and volatilities.  
The results are summarized in Table \ref{tblAccuracyZBoptionStrikesNonHomo}. 
Column ``price-to-strike'' shows the price-to-strike ratio, calculated as the actual zero-coupon bond price over the option strike price $K>0$. The price-to-strike ratio is a measure indicating the degree of moneyness of an option. A ratio above (resp. below) one shows that the call option is in the money (resp. out of the money), whereas a value of one indicates that the option is at the money\footnote{The price-to-strike ratios differ between the Hull-White and CIR++ models because the short rate in the CIR++ model cannot become negative, limiting the zero-bond price to $1$.}. 
\begin{table}[h]
	\begin{subtable}[c]{0.495\linewidth}
		\centering
		\scalebox{0.70}{
			\begin{tabular}{ccccc}
		\hline
		$\mathbf{\sigma}$&\textbf{price-to-strike}  & \textbf{CTMC} & \textbf{Benchmark}  & \textbf{Abs. error} \\
		\hline 
		\multirow{5}{*}{$\mathbf{0.1}$}& 1.67&	 0.38741911& 	0.38741911&	4.95E-14\\
        &1.25&	 0.22924215 &	0.22924215&	3.43E-10 \\
        &1.00&	 0.07281370 &	0.07281784&	4.14E-06 \\
        &0.83&	 0.00130678& 	0.00130552&	1.26E-06\\
        &0.71&	 0.00000021 &	0.00000023&	2.49E-08\\
		\hline
	\end{tabular}}
	\end{subtable}
	\begin{subtable}[c]{0.495\linewidth}
		\centering
		\scalebox{0.70}{
			\begin{tabular}{ccccc}
		\hline
		$\mathbf{\sigma}$&\textbf{price-to-strike} & \textbf{CTMC} & \textbf{Benchmark} & \textbf{Abs. error} \\
		\hline 
		\multirow{5}{*}{$\mathbf{0.1}$} &1.67&	 0.38741911& 	 0.38741911& 	2.46E-11\\
                                        &1.25&	 0.22924215& 	 0.22924215& 	2.18E-11\\
                                        &1.00&	 0.07106519& 	 0.07106519& 	1.90E-11\\
                                        &0.95&	 0.03153223& 	 0.03153224& 	9.24E-09\\
                                        &0.92&	 0.00150899& 	 0.00150815& 	8.40E-07\\
		\hline
	\end{tabular}}
	\end{subtable}
 \begin{subtable}[c]{0.49\linewidth}
		\centering
		\scalebox{0.70}{
			\begin{tabular}{ccccc}
		\hline
		$\mathbf{\sigma}$&\textbf{price-to-strike} & \textbf{CTMC} & \textbf{Benchmark} & \textbf{Abs. error} \\
		\hline 
		\multirow{5}{*}{$\mathbf{0.2}$}&1.67&	 0.38741912& 	0.38741912&	4.50E-10\\
                                       &1.25&	 0.22939397 &	0.22939439&	4.21E-07\\
                                       &1.00&	 0.08510628 &	0.08509821&	8.07E-06\\
                                        &0.83&	 0.01328299 &	0.01328294&	4.45E-08\\
                                        &0.71&	 0.00083528 &	0.00083714&	1.86E-06\\
		\hline
	\end{tabular}}
	\end{subtable}
 \begin{subtable}[c]{0.49\linewidth}
		\centering
		\scalebox{0.70}{
			\begin{tabular}{ccccc}
		\hline
		$\mathbf{\sigma}$&\textbf{price-to-strike} & \textbf{CTMC} & \textbf{Benchmark} & \textbf{Abs. error} \\
		\hline 
		\multirow{5}{*}{$\mathbf{0.2}$}&1.67	& 0.38741911 &	 0.38741911 &	3.58E-10\\
                                       & 1.25	 &0.22924215 &	 0.22924215 &	3.12E-10\\
                                       & 1.00	 &0.07106519 &	 0.07106519 &	2.86E-10\\
                                        &0.95	& 0.03153223 &	 0.03153224 &	9.24E-09\\
                                        &0.92	 &0.00299047 &	 0.00299063 &	1.62E-07\\
		\hline
	\end{tabular}}
	\end{subtable}
  \begin{subtable}[c]{0.49\linewidth}
		\centering
		\scalebox{0.70}{
			\begin{tabular}{ccccc}
		\hline
		$\mathbf{\sigma}$&\textbf{price-to-strike} & \textbf{CTMC} & \textbf{Benchmark} & \textbf{Abs. error} \\
		\hline 
		\multirow{5}{*}{$\mathbf{0.3}$}&1.67&	 0.38743466& 	0.38743476&	9.94E-08\\
                                       &1.25&	 0.23168156 &	0.23168115&	4.13E-07\\
                                       &1.00&	 0.10193411 &	0.10192793&	6.17E-06\\
                                       &0.83&	 0.03096382 &	0.03096201&	1.81E-06\\
                                       &0.71&	 0.00681115 &	0.00681259&	1.44E-06\\

		\hline
	\end{tabular}}
	\end{subtable}
  \begin{subtable}[c]{0.49\linewidth}
		\centering
		\scalebox{0.70}{
			\begin{tabular}{ccccc}
		\hline
		$\mathbf{\sigma}$&\textbf{price-to-strike} & \textbf{CTMC} & \textbf{Benchmark} & \textbf{Abs. error} \\
		\hline 
		\multirow{5}{*}{$\mathbf{0.3}$}& 1.67&	 0.38741912 &	 0.38741911 &	1.61E-10\\
                                       & 1.25&	 0.22924215 &	 0.22924215& 	1.39E-10\\
                                       & 1.00&	 0.07106791 &	 0.07106800& 	8.86E-08\\
                                       & 0.95&	 0.03153223 &	 0.03153224& 	9.24E-09\\
                                       & 0.92&	 0.00432070 &	 0.00431963& 	1.07E-06\\
		\hline
	\end{tabular}}
	\end{subtable}
  \begin{subtable}[c]{0.49\linewidth}
		\centering
		\scalebox{0.70}{
			\begin{tabular}{ccccc}
		\hline
		$\mathbf{\sigma}$&\textbf{price-to-strike} & \textbf{CTMC} & \textbf{Benchmark} & \textbf{Abs. error} \\
		\hline 
		\multirow{5}{*}{$\mathbf{0.4}$}& 1.67&	 0.38776664& 	0.38776489&	1.75E-06\\
                                        &1.25&	 0.23777915& 	0.23777104&	8.11E-06\\
                                        &1.00&	 0.12017293& 	0.12017104&	1.89E-06\\
                                        &0.83&	 0.05053076& 	0.05052852&	2.24E-06\\
                                        &0.71&	 0.01840799& 	0.01841005&	2.05E-06\\

		\hline
	\end{tabular}}
		\subcaption{Hull--White model}
	\end{subtable}
  \begin{subtable}[c]{0.49\linewidth}
		\centering
		\scalebox{0.70}{
			\begin{tabular}{ccccc}
		\hline
		$\mathbf{\sigma}$&\textbf{price-to-strike} & \textbf{CTMC} & \textbf{Benchmark} & \textbf{Abs. error} \\
		\hline 
		\multirow{5}{*}{$\mathbf{0.4}$}& 1.67	& 0.38741912& 	 0.38741911& 	1.55E-10\\
                                       &  1.25	 &0.22924215& 	 0.22924215& 	1.23E-10\\
                                       &  1.00	 &0.07110811& 	 0.07111282& 	4.71E-06\\
                                       &  0.95	 &0.03153223& 	 0.03153224& 	9.24E-09\\
                                       &  0.92	 &0.00545130& 	 0.00544519& 	6.11E-06\\
		\hline
	\end{tabular}}
		\subcaption{CIR++ model}
	\end{subtable}
	\caption[Accuracy of the approximated price of zero-coupon bond call options under the Hull--White and CIR++ models]{\small{Accuracy of the approximated price of zero-coupon bond call options, Proposition \ref{propCallPutInhomogeneousModels} and Corollary \ref{corrCallPutBrigo}, under the Hull--White and CIR++ models, respectively. Benchmark prices are calculated using closed-form analytical formulas. Except for the number of grid points set to $m=200$, model and CTMC parameters are as listed in Table \ref{tblModelandCTMCparamNonHomo}. Zero-coupon bond call option parameters using the notation of Proposition \ref{propCallPutInhomogeneousModels}: $t_{n_1}=0$, $t_{n_2}=2$, and $T=4$.}}
	\label{tblAccuracyZBoptionStrikesNonHomo}
\end{table} 

We observe that the two models achieve high degrees of accuracy across all parameters and strikes, with absolute errors below $8.11\text{E-}06$. It is also worth noticing the high precision of the approximation for deep-out-of-the-money options, indicating good approximations of the left tails of the underlying short-rate process. Such a high level of accuracy can be difficult to attain when using other numerical techniques, particularly for out-of-the-money options. 

Analogous experiments have been conducted with put options
, with similar results for out-of-the-money put options, indicating good approximations of the right tails of the short-rate diffusion process. The accuracy of the approximation across different model parameter values has also been tested. The methods exhibit a high level of precision across all parameters. Results are available upon request. 

\begin{figure}[t]
	\centering
		\includegraphics[scale=0.35]{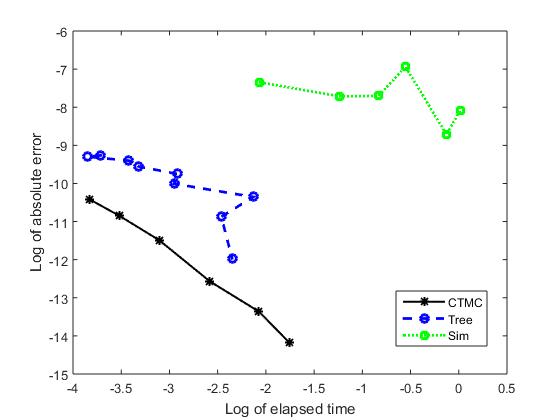} 
	\caption[Efficiency of the CTMC approximations compared to trees and simulation methods in approximating the price of zero-coupon bond call options under the Hull--White model]{\small{Efficiency of the CTMC method \eqref{eqCallPutZero} compared to trees and simulation methods in approximating the price of zero-coupon bond call options under the Hull--White model. Except for the number of grid points $m$, which range from $100$ to $350$, model and CTMC parameters are as listed in Table \ref{tblModelandCTMCparamNonHomo}. For the tree, we use between 300 and 700 time steps per year. Zero-coupon bond call option parameters using the notation of Proposition \ref{propCallPutInhomogeneousModels}: $t_{n_1}=0$, $t_{n_2}=2$, $T=4$, and $K=0.9$.}}\label{figEfficiencyHW}
\end{figure}
\begin{figure}[t]
	\centering
	\begin{tabular}{cc}
		\includegraphics[scale=0.3]{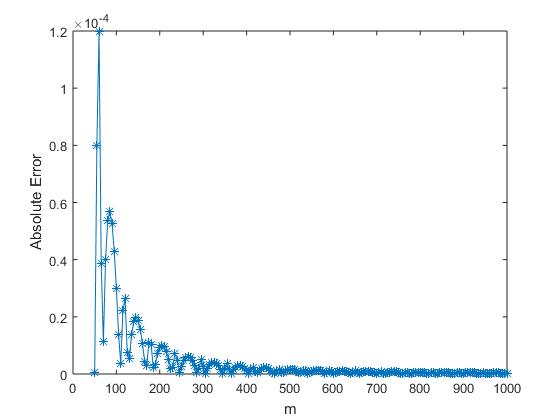} & \includegraphics[scale=0.3]{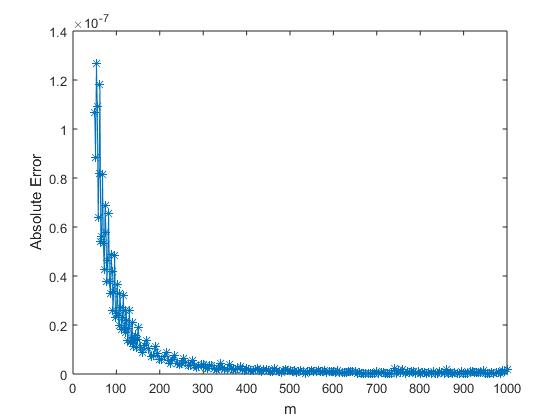}\\
        \small{\textsc{(a)} Hull--White} & \small{\textsc{(b)} CIR++}
	\end{tabular}
	\caption[Convergence pattern of the approximated zero-coupon bond call prices under Hull--White and CIR++ models]{\small{Convergence pattern of the approximated zero-coupon bond call prices under Hull--White and CIR++ models. Except for the number of grid points $m$, model and CTMC parameters are as listed in Table \ref{tblModelandCTMCparamNonHomo}. Zero-coupon bond call option parameters using the notation of Proposition \ref{propCallPutInhomogeneousModels}: $t_{n_1}=0$, $t_{n_2}=2$, $T=4$, and $K=0.9$.}}\label{figConvergenceZBoptionNonHomo}
\end{figure}

\begin{table}[h]
	\begin{subtable}[c]{0.495\linewidth}
		\centering
			\begin{tabular}{ccccc}
		\hline
		\textbf{m}  & \textbf{Abs. error} & \textbf{Rate}  \\
		\hline 
		60&	1.20E-04 &	- \\
        80&	5.36E-05&	2.791 \\
        100&	2.98E-05&	2.637 \\
        120&	2.64E-05&	0.661\\
        160&	1.07E-05&	3.130\\
		\hline
	\end{tabular}
    \subcaption{Hull--White model}
	\end{subtable}
	\begin{subtable}[c]{0.495\linewidth}
		\centering
			\begin{tabular}{ccccc}
		\hline
		\textbf{m}  & \textbf{Abs. error} & \textbf{Rate}  \\
        \hline 
		60	&8.19E-08&	- \\
        80	&4.64E-08&	1.977\\
        100	&2.45E-08&	2.854\\
        120	&1.71E-08&	1.973\\
        160	&8.61E-09&	2.387\\
		\hline
	\end{tabular}
    \subcaption{CIR++ model}
	\end{subtable}
 
	\caption[Convergence Rate Approximation Zero-Coupon Bond Call Options]{\small{Approximation of the convergence rate of the price of zero-coupon bond call options, Proposition \ref{propCallPutInhomogeneousModels} and Corollary \ref{corrCallPutBrigo}, under the Hull--White and CIR++ models, respectively. Benchmark prices are calculated using closed-form analytical formulas. Except for the number of grid points, model and CTMC parameters are as listed in Table \ref{tblModelandCTMCparamNonHomo}. Zero-coupon bond call option parameters using the notation of Proposition \ref{propCallPutInhomogeneousModels}: $t_{n_1}=0$, $t_{n_2}=2$, and $T=4$.}}
	\label{tblConvRateZBoptionNonHomo}
\end{table} 

We compare the efficiency of the CTMC methodology \eqref{eqCallPutZero} to the trinomial tree method (``Tree'') of \cite{hull1994numerical} and \cite{hull1996using}, and Monte Carlo simulation (``Sim'') using the procedure of \cite{Ostrovski2013Efficient}. For the Monte Carlo simulation, we used an Euler discretization scheme, with a number of simulations ranging from $10,000$ to $100,000$ and $252$-time steps per year. Figure \ref{figEfficiencyHW} shows the results. The calibration time is not included in the log elapsed time in Figure \ref{figEfficiencyHW}, that is, for the tree, the calculation time does not include the construction of the interest rate tree that perfectly fits the market data, and for the CTMC approximation, the calculation time does not include the calibration of the model to the market zero-bond curve using Algorithm \ref{algoThetasCalibration}. Note that the construction of the interest rate tree is generally much faster than the CTMC calibration process with Algorithm \ref{algoThetasCalibration}. When using 252-time steps per year, the tree is built in a fraction of a second, whereas the calibration process with Algorithm \ref{algoThetasCalibration} takes on average 3.2 seconds with $m=160$. 
Figure \ref{figEfficiencyHW} shows the high efficiency of CTMC methods compared to other methods. CTMC approximation clearly outperforms these other techniques in terms of both calculation time and precision. The speed of the approximation for the extended models of \cite{brigoMercurio2006}, such as the CIR++ model, is similar to that obtained for the homogeneous models with an average calculation time of less than $0.015$ seconds (excluding calibration).    

The convergence patterns of the value of the call option to the analytical price as the number of grid points $m$ increases are displayed in Figure \ref{figConvergenceZBoptionNonHomo}, whereas Table \ref{tblConvRateZBoptionNonHomo} shows the convergence rate. We observe that the approximation achieves superquadratic convergence on average. The CTMC approximation is known to achieve a theoretical quadratic convergence rate (rate = 2) for barrier and European options in one-dimensional diffusion models with particular grid designs; see \cite{zhang2019analysis} for details. We also note that the two models converge rapidly to their analytical values but exhibit a sawtooth pattern. 
Such oscillatory behavior has also been observed by \cite{zhang2019analysis} in the context of double barrier knock-out options pricing within the CTMC approximation framework. In particular, they observe that constructing a grid with the strike placed precisely in the middle of two grid points removes oscillations (see Section 4.7 of their paper for details). However, their grid design is not directly applicable to the present context since, in this paper, the grid represents the state-space of the short-rate process $R^{(m)}$, and the option (and the strike) depends on the zero-bond price $P^{(m)}(\cdot,T)$, whose approximation also depends on the grid design. In the context of tree methods approximation, a detailed study of the oscillatory behavior of European vanilla options has been performed in \cite{diener2004asymptotics}, whereas \cite{tavellapricing}, Chapter 5, observes that convergence oscillation of the finite difference method can be reduced in a non-uniform grid design when the strike is placed midway between two grid points. Further investigation into how grid design can improve convergence is left for future research.
Finally, since Assumption \ref{assumpStateSpaceR} is not satisfied under the Hull-White model, the preceding experiments show that the results of Section \ref{sectionApplicationDebtSecurities} can hold under less restrictive conditions for a specific set of parameters\footnote{Some testing has also been performed for the EV+ model, and the results are available upon request. These experiments also indicate that Assumption \ref{assumpStateSpaceR} could potentially be relaxed.}. 
\subsection{Approximation of Callable/Putable Bond Prices}\label{sectionNumExpCallableBond}
We now examine the accuracy of Proposition \ref{propCallableputableDebt} in approximating callable/putable bonds under the Hull--White model. Accordingly, we consider a coupon-bearing bond with semi-annual coupons that mature in $4$ years, $T=4$. The coupon rate, denoted by $\alpha$, is set to $5\%$ per annum compounded semi-annually. The notional of the debt is set to $F=100$, and we assume that it can be called at any time between the second and fourth year for no additional cost, that is, $K^c_t=100$ for $2\leq t\leq T$, and we let $K_t^c\rightarrow\infty$ when $t<2$ (as exercise is not allowed). Moreover, since there is no put feature, $K^p:=K^p_t=0$ for $0\leq t\leq T$. Finally, we assume that accrued interest is paid to the bondholder upon redemption\footnote{This is a standard assumption in practice, meaning that the call price $K_t^c$ is increased by accrued interest upon redemption.}.
The results are summarized in Table \ref{tblAccuracyCallableDebtHW}. The column ``CTMC'' shows the approximated value of the debt using CTMCs with $m=160$ as specified in Table \ref{tblModelandCTMCparamNonHomo}, whereas the column ``Benchmark'' shows the CTMC approximated value with $m=350$. The value using the tree method of \cite{hull1994numerical} is reported in the column ``Tree''.  
\begin{table}[h]
	\begin{subtable}[c]{0.9\linewidth}
		\centering
		\scalebox{0.9}{
			\begin{tabular}{cc|ccc|ccc}
		\hline
		\multirow{2}{*}{$\mathbf{\kappa}$} & \multirow{2}{*}{\textbf{Benchmark}}&\multirow{2}{*}{\textbf{CTMC}} & \textbf{Rel.} & \textbf{Elapsed} & \multirow{2}{*}{\textbf{Tree}} & \textbf{Rel.} & \textbf{Elapsed}\\
          &   &  &\textbf{Error}& \textbf{Time (sec)} & & \textbf{Error} & \textbf{Time (sec)}\\
		\hline 
        \textbf{0.5}&	  \textbf{91.6418214}&	 91.6426071 &	8.57E-06&   0.0936 	 &91.6070243& 	3.80E-04 &   0.0411 \\
        \textbf{1}&	     \textbf{95.6073132}&	 95.6072144 &	1.03E-06 &	0.0897   &95.5614168& 	4.80E-04 &	 0.0301 \\
        \textbf{2}&	     \textbf{98.5647947}&	 98.5648058 &	1.13E-07&  0.0961 	 &98.5108693& 	5.47E-04 &  0.0167 \\
        \textbf{3}&	     \textbf{99.7010757}&	 99.7004280 &	6.50E-06&  0.0916    &99.6443817& 	5.69E-04 &  0.0157 \\
		\hline
	\end{tabular}}
	\end{subtable}
 

     \begin{subtable}[c]{0.9\linewidth}
		\centering
		\scalebox{0.9}{
			\begin{tabular}{cc|ccc|ccc}
		\hline
		\multirow{2}{*}{$\mathbf{\sigma}$} & \multirow{2}{*}{\textbf{Benchmark}}&\multirow{2}{*}{\textbf{CTMC}} & \textbf{Rel.} & \textbf{Elapsed} & \multirow{2}{*}{\textbf{Tree}} & \textbf{Rel.} & \textbf{Elapsed}\\
          &   &  &\textbf{Error}& \textbf{Time (sec)} & & \textbf{Error} & \textbf{Time (sec)}\\
		\hline 
		\textbf{0.1}&	 \textbf{98.9838248 }&	 98.9818310 &	2.01E-05&  0.0780  &98.9454895 &	3.87E-04&  0.0238 \\
        \textbf{0.2}&	 \textbf{95.6073132 }&	 95.6072144 &	1.03E-06&  0.0927 	 &95.5614168 &	4.80E-04&  0.0258 \\
        \textbf{0.3}&	 \textbf{92.2381361 }&	 92.2382928 &	1.70E-06&	 0.1011  &92.1871120 &	5.53E-04&  0.0244 \\
        \textbf{0.4}&	 \textbf{88.9338108 }&	 88.9340669 &	2.88E-06&  0.1198 	 &88.8822619 &	5.80E-04&  0.0289 \\
		\hline
	\end{tabular}}
	\end{subtable}
	\caption[Accuracy CTMC methods in approximating the price of callable bonds under the Hull--White model]{\small{Accuracy of Proposition \ref{propCallableputableDebt} in approximating the price of callable bonds under the Hull--White model. Benchmark is calculated using CTMC approximation with $m=350$. Model and CTMC parameters are as listed in Table \ref{tblModelandCTMCparamNonHomo}. For the tree method, we use $252$-time steps per year. Contract specifications are $F=100$, $\alpha=0.05$, $T=4$, $K_t^c=100$, and $K^p=0$, with the call option exercise window starting from $t=2$ to $T=4$.}}
	\label{tblAccuracyCallableDebtHW}
\end{table} 

\begin{figure}[h]
	\centering
		\includegraphics[scale=0.3]{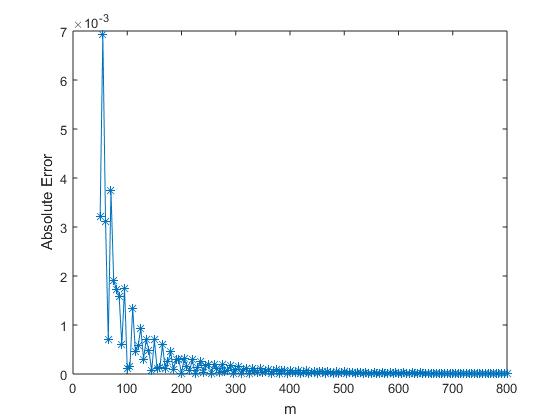}
	\caption[Convergence pattern of the approximated price of callable bonds under the Hull--White model]{\small{Convergence pattern of the approximated price of callable bonds, Proposition \ref{propCallableputableDebt}, as the number of grid points $m$ increases. Benchmark is calculated using CTMC approximation with $m=1,000$. Except for the number of grid points $m$, model and CTMC parameters are as listed in Table \ref{tblModelandCTMCparamNonHomo}.  Contract specifications are $F=100$, $\alpha=0.05$, $T=4$, $K_t^c=100$, and $K^p=0$, with the call option exercise window starting from $t=2$ to $T$.}}\label{figConvergenceCallableHW}
\end{figure}
\begin{table}[h]
		\centering
			\begin{tabular}{ccccc}
		\hline
		\textbf{m}  & \textbf{Abs. error} & \textbf{Rate}  \\
		\hline 
        50 &	3.21E-03 &	- \\
        80 &	1.73E-03 &	 1.31\\ 
        120 &	5.81E-04 &	 2.70\\ 
        140	&4.75E-04	 &1.31\\ 
        175	&2.47E-04	 &2.93\\ 
		\hline
	\end{tabular}
	\caption[Approximation of the convergence rate of price of callable bonds under the Hull--White model]{\small{ - Approximation of the convergence rate of the price of callable bonds, Proposition \ref{propCallableputableDebt}, as the number of grid points $m$ increases. Benchmark is calculated using CTMC approximation with $m=1,000$. Except for the number of grid points $m$, model and CTMC parameters are as listed in Table \ref{tblModelandCTMCparamNonHomo}.  Contract specifications are $F=100$, $\alpha=0.05$, $T=4$, $K_t^c=100$, and $K^p=0$, with the call option exercise window starting from $t=2$ to $T$.}}
	\label{tblConvRateCallableDebtHW}
\end{table} 

We observe that the approximated values from the CTMC and tree methods are close to each other, confirming the adequacy of Proposition \ref{propCallableputableDebt}. Additionally, the relative error of the CTMC approximation is lower than that of the tree method. However, the tree method is shown to be 2 to 5 times faster than the CTMC approximation. The efficiency of the two approaches is shown in the context of zero-coupon bond option pricing in Figure \ref{figEfficiencyHW}.

The convergence pattern of the approximated callable price as $m$ grows is displayed in Figure \ref{figConvergenceCallableHW}, whereas Table \ref{tblConvRateCallableDebtHW} shows the convergence rates. 
The absolute error decreases rapidly to $0$ but exhibits a sawtooth pattern. As mentioned in Section \ref{sectionNumExpBond}, previous work in different contexts has shown that grid design can improve convergence and that placing the strike midway between two grid points can reduce or remove oscillatory behavior (see \cite{tavellapricing}  and \cite{zhang2019analysis} for details). However, the methodologies proposed in these studies are not directly applicable to the present context. Therefore, further investigation into how grid design can improve convergence is left for future research.

\subsection{Approximation of Convertible Bond Prices}\label{sectConvBondNum}
We now investigate the accuracy and efficiency of Proposition \ref{propConvBondCTMC} in approximating CB prices under the Black--Scholes--Hull--White model, as well as the numerical convergence of the price estimates. That is, we suppose that the stock price dynamics follow a geometric Brownian motion with stochastic interest rate satisfying
 \begin{equation}
        \begin{aligned}
         \diff S_t &=(R_t-q_t) S_t\diff t+\tilde{\sigma}_S S_t \diff W^{(1)}_t,\\
	   \diff R_t &=(\theta(t)-\kappa R_t)\diff t+\tilde{\sigma}_R\diff W^{(2)}_t,\label{eqEDS_BSvasicek}
         \end{aligned}
	\end{equation}
  with $\kappa$, $\tilde{\sigma}_S$, $\tilde{\sigma}_R$>0, and $[W^{(1)},W^{(2)}]_t=\rho t$, $\rho\in[-1,1]$. 
  
  From Lemma \ref{lemmaStoX}, we find that $f(r)=\frac{\tilde{\sigma}_S}{\tilde{\sigma}_R}r$. The dynamics of the auxiliary process ${X_t=\ln(S_t)-\rho f(R_t)}$ can then be derived as
\begin{equation}
    \begin{split}\label{eqEDS_XBSvasicek}
        \diff X_t & =\mu_X(t,R_t)\diff t+\sigma_X(R_t)\diff W^\star_t\\
        \diff R_t & = (\theta(t)-\kappa R_t)\diff t+\tilde{\sigma}_R\diff W^{(2)}_t,
    \end{split}
\end{equation}
with $\mu_X(t,R_t)=R_t-q_t-\frac{\tilde{\sigma}_S^2}{2}-\rho\frac{\tilde{\sigma}_S}{\tilde{\sigma}_R}(\theta(t)-\kappa R_t)$, $\sigma_X=\tilde{\sigma}_S\sqrt{1-\rho^2}$, and $X_0=\ln(S_0)-\rho f(R_0)$.

  Unless stated otherwise, the model parameters for the short-rate process are the same as those used in previous examples, summarized in Table \ref{tblModelandCTMCparamNonHomo}. Recall also that function $\theta$ is calibrated to the market zero-bond curve in Table \ref{tblmarketZBcurve} using Algorithm \ref{algoThetasCalibration}, as explained in Section \ref{subsecFitTermStruc}. Unless stated otherwise, we suppose that $\tilde{\sigma}_S=0.2$, $q_t=0.02$ for all $t\in[0,T]$, and $\rho=-0.2$. The model parameters are summarized in Table \ref{tblModelParamBSHW}.
  \begin{table}[h!]
	\begin{tabular}{ccccccccc}
		\hline
		Model & ${R_0}$ & $q_t$ & ${\kappa}$  & ${\tilde{\sigma}_R}$ & $S_0$ & $\tilde{\sigma}_S$ & $\rho$\\
		\hline 
		Black--Scholes--Hull--White  & $0.04$ & 0.02 & $1$ & $0.20$  &$100$& $0.2$ & $-0.2$\\
		\hline
	\end{tabular}
	\caption[Model parameters Black--Scholes--Hull--White]{Model parameters}
	\label{tblModelParamBSHW}
\end{table} 
  
  The grid used to approximate the short-rate process, $\mathcal{S}^{(m)}_R=\{r_1,r_2,\ldots,r_m\}$, and the auxiliary process $\mathcal{S}^{(M)}_X=\{x_1,x_2,\ldots,x_M\}$, are constructed using the methodology of Tavella and Randall (\cite{tavellapricing}, Chapter 5), as explained at the beginning of this section, with $\tilde{\alpha}_R$ (resp.  $\tilde{\alpha}_X$) representing the non-uniformity parameter of the grid of $R^{(m)}$ (resp. $X^{(m)}$). Unless otherwise indicated, all numerical experiments are conducted using the CTMC parameters listed in Table \ref{tblCTMCParamConvBondBSHW}. Note also that the fast versions of Propositions \ref{prop:vanillaCB_CTMC} and \ref{propConvBondCTMC} reported in Appendix \ref{appendixCallablePutableBondAlgo} have been used in all numerical examples.

  \begin{table}[h]
	\begin{tabular}{cccccccccc}
		\hline
		 Model & ${m}$  & $M$ & ${r_1}$ & ${r_m}$ & $\tilde{\alpha}_R$& ${x_1}$ & ${x_M}$ & $\tilde{\alpha}_X$& $\Delta_N$ \\
       \hline
		Black--Scholes--Hull--White  &$160$ & $160$ & $-30 R_0$& $25 R_0$ & $0.5$ & $0.64 X_0$ & $1.42 X_0$& $2$ &$1/252$\\
		\hline
	\end{tabular}
	\caption[CTMC parameters Black--Scholes--Hull--White]{CTMC parameters}
	\label{tblCTMCParamConvBondBSHW}
\end{table}
For the testing, we consider a convertible bond that pays semi-annual coupons with an annual rate $\alpha=0.05$ and has a notional value $F=100$. We suppose that the bond can be converted at any time from inception to maturity ($T=1$) at a conversion rate $\eta=1$. 
Under this set of parameters and assuming that both the dividend yield and credit risk are nil ($q_t=c_t=0$ for all $t\in[0,T]$), the valuation of American-style CBs simplifies to the valuation of European-style convertible bonds, as stated in Proposition \ref{propCBtrivialTF} and generalized to coupon paying bonds in Proposition \ref{corCBtrivialTFcoupon}. In that particular case, the results of Proposition \ref{prop:ExactFormulaVanillaCB_TF}\footnote{The expected present value of future coupons should be added to the formula obtained in  Proposition \ref{prop:ExactFormulaVanillaCB_TF}.}, available online as supplemental material, can thus serve as a benchmark in our analysis. When $q_t,c_t>0$ for some $t\in [0,T]$, the benchmark is calculated using CTMC approximation with $M=300$. All other parameters are as stated in Table \ref{tblCTMCParamConvBondBSHW}. The results are summarized in Table \ref{tblAccuracyCB_BSHW2}. 
\begin{table}[t]
	\begin{subtable}[c]{0.495\linewidth}
		\centering
		\scalebox{0.90}{
			\begin{tabular}{cccc}
		\hline
		$\mathbf{S_0}$ & \textbf{CTMC} & \textbf{Benchmark} & \textbf{Rel. error} \\
		\hline 

           \textbf{90}	& 105.15488 &	 105.15732 	&2.32E-05\\
           \textbf{95}	& 107.56087 &	 107.57350 	&1.17E-04\\
           \textbf{100}& 110.48565 &	 110.50458 	&1.71E-04\\
           \textbf{105}& 113.87792 &	 113.89940 	&1.89E-04\\
           \textbf{110}& 117.66876 &	 117.68977 	&1.79E-04\\

		\hline
	\end{tabular}}
	\end{subtable}
	\begin{subtable}[c]{0.495\linewidth}
		\centering
		\scalebox{0.90}{
			\begin{tabular}{cccc}
		\hline
		$\mathbf{S_0}$ & \textbf{CTMC} & \textbf{Benchmark} & \textbf{Rel. error} \\
		\hline 
            \textbf{90}	&101.96277 &	 101.93530 &	2.69E-04\\
            \textbf{95}	&104.84800 &	 104.82384 &	2.30E-04\\
            \textbf{100}&108.25070 &	 108.21547 &	3.26E-04\\
            \textbf{105}&112.07490 &	 112.03748 &	3.34E-04\\
             \textbf{110}&116.24312& 	 116.19981& 	3.73E-0\\
		\hline
	\end{tabular}}
	\end{subtable}
 \begin{subtable}[c]{0.495\linewidth}
		\centering
		\scalebox{0.90}{
			\begin{tabular}{cccc}
		\hline
		$\mathbf{\sigma_S}$ & \textbf{CTMC} & \textbf{Benchmark} & \textbf{Rel. error} \\
		\hline 
        \textbf{0.1}  &107.36113 &	 107.38133 &	1.88E-04\\
         \textbf{0.15}& 108.82431& 	 108.84084& 	1.52E-04\\
        \textbf{0.2}  &110.48565 &	 110.50458 &	1.71E-04\\
        \textbf{0.3}  &114.04675 &	 114.07719 &	2.67E-04\\
        \textbf{0.4}  &117.72017 &	 117.76496 &	3.80E-04\\

		\hline
	\end{tabular}}
	\end{subtable}
 \begin{subtable}[c]{0.49\linewidth}
		\centering
		\scalebox{0.90}{
			\begin{tabular}{cccc}
		\hline
		$\mathbf{\sigma_S}$ & \textbf{CTMC} & \textbf{Benchmark} & \textbf{Rel. error} \\
		\hline 
         \textbf{0.1} &105.52046 &	 105.48338 &	3.52E-04\\
        \textbf{0.15}&106.75571 &	 106.71897 	&3.44E-04\\
        \textbf{0.2}  &108.25070 &	 108.21547 	&3.26E-04\\
        \textbf{0.3}  &111.58183 &	 111.55168 	&2.70E-04\\
        \textbf{0.4}  &115.08975 &	 115.06539 	&2.12E-04\\
		\hline
	\end{tabular}}
	\end{subtable}
  \begin{subtable}[c]{0.49\linewidth}
		\centering
		\scalebox{0.90}{
			\begin{tabular}{cccc}
		\hline
		$\mathbf{\rho}$ & \textbf{CTMC} & \textbf{Benchmark} & \textbf{Rel. error} \\
		\hline 
        \textbf{-0.3}&110.20950 &	 110.22594 &	1.49E-04\\
        \textbf{-0.2}&110.48565 &	 110.50458 &	1.71E-04\\
        \textbf{0.2}&111.50772 	&   111.53335 	&2.30E-04\\
        \textbf{ 0.3}&111.74601 &	 111.77262 	&2.38E-04\\
		\hline
	\end{tabular}}
		\subcaption{$q_t=c_t= 0$ for all $t\in[0,T]$}
	\end{subtable}
  \begin{subtable}[c]{0.49\linewidth}
		\centering
		\scalebox{0.90}{
			\begin{tabular}{cccc}
		\hline
		$\mathbf{\rho}$ & \textbf{CTMC} & \textbf{Benchmark} & \textbf{Rel. error} \\
		\hline 
        \textbf{-0.3}&107.99824 &	 107.96134 	&3.42E-04\\
         \textbf{-0.2}&108.25070& 	 108.21547 &	3.26E-04\\
        \textbf{0.2}&109.19534 	& 109.16259 	&3.00E-04\\
        \textbf{ 0.3}&109.41008 &	 109.38793 	&2.03E-04\\
		\hline
	\end{tabular}}
		\subcaption{$q_t=0.02$, $c_t= 0.05$ for all $t\in[0,T]$}
	\end{subtable}
	\caption[Accuracy of the approximation of American-style CB prices under Black--Scholes--Hull--White model.]{\small{Accuracy of the approximation of American-style CB prices, Algorithm \ref{algoCBpriceTF_CTMCfast}, under Black--Scholes--Hull--White model. Model and CTMC parameters are as listed in Tables \ref{tblModelParamBSHW} and \ref{tblCTMCParamConvBondBSHW}, respectively. Contract specifications are $F=100$, $T=1$, and $\eta=1$, with an annual coupon rate $\alpha=0.05$ paid semi-annually.}}\label{tblAccuracyCB_BSHW2}
\end{table} 
We note that the model achieves a high level of accuracy across all model parameters, 
with an average calculation time of less than 10 seconds. 
When the short-rate process is time-homogeneous, the matrix exponential can be calculated only once at the beginning of the procedure, which speeds up the procedure significantly. For instance, under the Black--Scholes--Vasicek model, the average calculation time for the CTMC approximated prices is less than 1.7 seconds. Numerical results for that particular model are reported in Appendix \ref{appendixSupplMatNumExperiment}, available online as supplemental material.

\begin{figure}[h]
	\centering
	\begin{tabular}{cc}
		\includegraphics[scale=0.3]{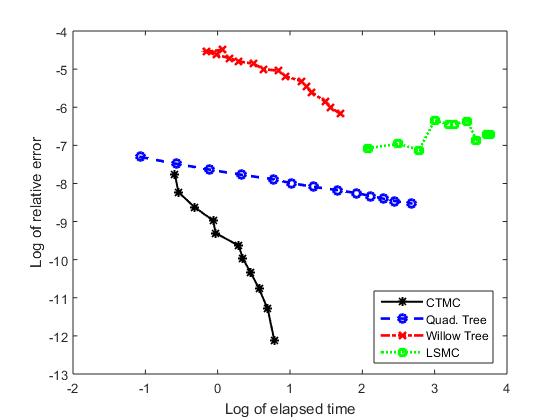} & \includegraphics[scale=0.3]{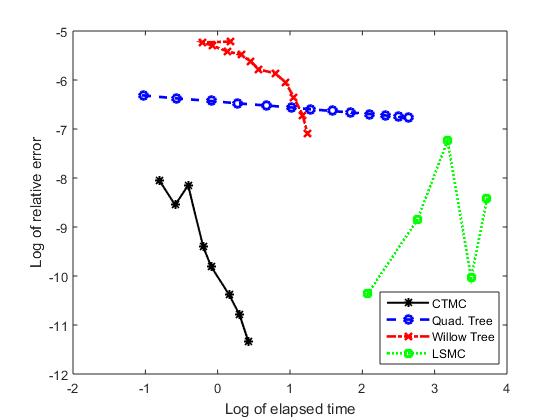}\\
  \small{\textsc{(a)} $q_t=c_t=0$ for all $t\in[0,T]$} & \small{\textsc{(b)} $q_t=0,c_t=0.05$ for all $t\in[0,T]$}  \\
	\end{tabular}
	\caption[Efficiency of the CTMC method in approximating CB prices under the Black—Scholes--Vasicek model]{\small{Efficiency of the CTMC method in approximating CB prices, Algorithm \ref{algoCBpriceTF_CTMCfast}, under the Black—Scholes--Vasicek model. Except for the number of grid points $M$, $\Delta_N=1/100$ and $\theta(t)=\kappa R_0$ for all $t\in[0,T]$, the model and CTMC parameters are as listed in Tables \ref{tblModelParamBSHW} and \ref{tblCTMCParamConvBondBSHW}, respectively. Contract specifications are $F=100$, $T=1$, $\eta=1$, and $\alpha=0$.}}\label{figEfficiencyCBpricesV2}
\end{figure}

Figure \ref{figEfficiencyCBpricesV2} shows the efficiency of the methodology compared to other recently developed numerical approaches. For this testing, we considered zero-coupon CBs ($\alpha=0$), and we set $\theta(t)=\kappa R_0$ for all $t\in[0,T]$, such that the short-rate process collapses to the Vasicek model with a long-term mean level equal to $R_0$. We compare the CTMC approximated prices to the Willow tree approach (``Willow Tree'') of \cite{lu2017simple}, the quadrinomial tree (``Quad. Tree'') of \cite{battauz2022american}, and the LSMC method of \cite{longstaff2001valuing}. All methodologies listed above have been adapted to incorporate credit risk as in the work of \cite{tsiveriotis1998valuing} for a better comparison. When both credit spread and dividend yield are set to nil, the benchmark is obtained using the closed-form formula derived in Proposition \ref{prop:ExactFormulaVanillaCB_TF} (available online as supplemental material); otherwise, CTMC approximation is used as a benchmark with $M=1,000$, and all other CTMC parameters are as listed in Table \ref{tblCTMCParamConvBondBSHW}.
Figure \ref{figEfficiencyCBpricesV2} clearly shows the high efficiency of the CTMC methodologies compared to other methods. CTMC approximation significantly outperforms these other techniques in terms of both precision and calculation time.

The convergence pattern of the approximation as $M$ increases is illustrated in Figure \ref{figConvergenceCB_BSHW}, whereas Table \ref{tblConvRateCBNonHomo} shows the convergence rate.
\begin{figure}[t]
	\centering
	\begin{tabular}{cc}
		\includegraphics[scale=0.3]{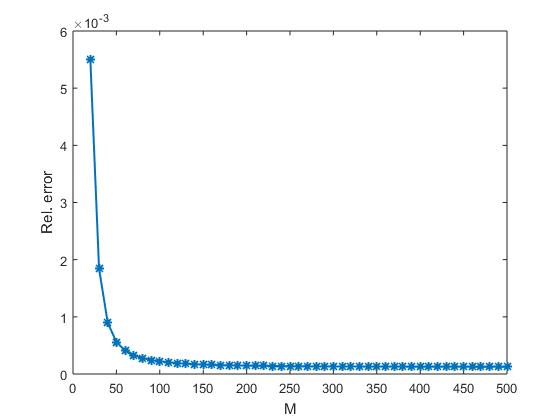} & \includegraphics[scale=0.3]{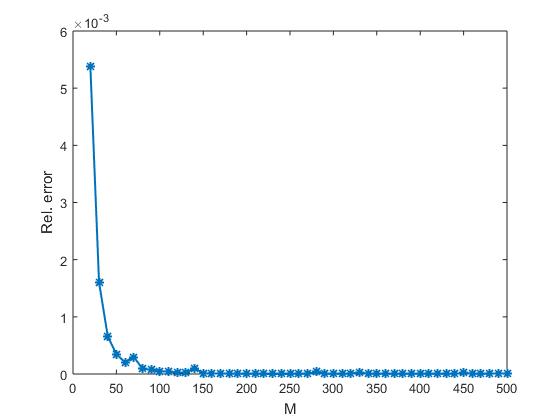}\\
  \small{\textsc{(a)} $q_t=c_t=0$ for all $t\in[0,T]$} & \small{\textsc{(b)} $q_t=0.02,c_t=0.05$ for all $t\in[0,T]$}  \\
	\end{tabular}
	\caption[Convergence pattern of the approximated CB prices using CTMC method under Black—Scholes--Hull--White model]{\small{Convergence pattern of the approximated CB prices using CTMC method, Algorihtm \ref{algoCBpriceTF_CTMCfast}, under Black—Scholes--Hull--White model. Except for the number of grid points $M$ of the auxiliary process, the model and CTMC parameters are as listed in Tables \ref{tblModelParamBSHW} and \ref{tblCTMCParamConvBondBSHW}, respectively. Contract specifications are $F=100$, $T=1$, and $\eta=1$, with an annual coupon rate $\alpha=0.05$ paid semi-annually.}}\label{figConvergenceCB_BSHW}
\end{figure}
\begin{table}[h]
	\begin{subtable}[c]{0.495\linewidth}
		\centering
			\begin{tabular}{ccccc}
		\hline
		\textbf{m}  & \textbf{Rel. error} & \textbf{Rate}  \\
		\hline 
		20	&5.50E-03 & -\\		
        50	&5.59E-04 &	2.494\\
        100	&2.21E-04 &		1.340\\
        150	&1.65E-04 &		0.714\\
        500	&1.25E-04 &		0.231\\
		\hline
	\end{tabular}
    \subcaption{$q_t=c_t=0$ for all $t\in[0,T]$}
	\end{subtable}
	\begin{subtable}[c]{0.495\linewidth}
		\centering
			\begin{tabular}{ccccc}
		\hline
		\textbf{m}  & \textbf{Rel. error} & \textbf{Rate}  \\
        \hline 
		20&	    5.38E-03 & - \\	
        50&	    3.43E-04 & 3.004\\
        100&	5.17E-05& 	2.731\\
        150&	1.30E-05&	3.396\\
        500&	2.45E-06&	1.390\\
		\hline
	\end{tabular}
    \subcaption{$q_t=0.02$, $c_t=0.05$ for all $t\in[0,T]$}
	\end{subtable}
 
	\caption[Convergence Rate Approximation Zero-Coupon Bond Call Options]{\small{Approximation of the convergence rate of the approximated CB prices using CTMC method, Algorithm \ref{algoCBpriceTF_CTMCfast}, under Black--Scholes--Hull-White model. Except for the number of grid points $M$ of the auxiliary process, the model and CTMC parameters are as listed in Tables \ref{tblModelParamBSHW} and \ref{tblCTMCParamConvBondBSHW}, respectively. Contract specifications are $F=100$, $T=1$, and $\eta=1$, with an annual coupon rate $\alpha=0.05$ paid semi-annually.}}
	\label{tblConvRateCBNonHomo}
\end{table} 
We observe that the approximated prices converge rapidly and smoothly to the benchmark prices, and the approximation achieves superquadratic convergence on average when accounting for both the dividend yield and credit risk.

Appendix \ref{appendixSupplMatNumExperiment}, available online as supplemental material, shows similar results under the Black--Scholes--Vasicek model. Analogous numerical analysis has also been performed under the Black--Scholes--CIR model with similar results, which are available upon request.  Finally, since Assumption \ref{assumpStateSpaceR} is not satisfied under the Black-Scholes--Hull--White model, the preceding experiments demonstrate that the results of Section \ref{sectionConvBond} can still be valid under less restrictive conditions under a certain set of parameters. Theoretical proof is left as future research.

\section{Conclusion}\label{sectConclu}
In this paper, we provide a general pricing framework based on continuous-time Markov chain approximations to value debt securities under general time-inhomogenous short-rate models. One advantage of CTMC methods over other commonly used numerical techniques is that they allow for a closed-form matrix expression for the price of zero-coupon bonds regardless of the complexity of the short-rate dynamics, making the calibration of the approximated model to the current market-term structure straightforward for a wide range of models. 
Closed-form matrix expressions are also obtained for the price of bond options and European CBs, 
and simple and efficient algorithms are developed to approximate the values of callable/putable bonds and CBs (American-style). Numerical results show the high accuracy and great efficiency of the methodology. Theoretical convergence is also discussed. 
The problem of CB pricing is also studied from a theoretical perspective. When both credit spread and dividend yield are set to nil and under some conditions on the parameters of the equity process (constant volatility), we show that early conversion is sub-optimal. When default risk is considered, we obtain lower and upper bounds for the value of American-style CBs. 

Alternative approaches to credit risk modeling, such as the one of \cite{hung2002pricing} and \cite{chambers2007tree} or \cite{milanov2012binomial}, can also be explored using CTMCs. The methodologies proposed by \cite{hung2002pricing} and \cite{chambers2007tree} can be easily integrated. However, \cite{milanov2012binomial} introduce default risk by incorporating a jump into the equity process. Jump processes in CTMCs have been studied in a one-dimensional setting by \cite{lo2014improved}, while  \cite{kirkby2022hybrid}, Appendix B, discusses some extensions to two-dimensional processes. Other developments to the present work can also be considered. For example,  building on insights from \cite{brigoMercurio2006}, Sections 2.6.1 and 3.3.2., the closed-form expression obtained for the price of bond options in \eqref{eqCallPutZero} and Remark \ref{rmkCouponBearingBond} can be used to approximate other derivatives like caps, floors, and swaptions. The pricing of these derivatives, as well as the associated model calibration, should be studied in greater detail. 

\section*{Disclosure statement}
The authors report there are no competing interests to declare.
\section*{Acknowledgement}
The first author acknowledges support from the Natural Science and Engineering Research Council of Canada (Grant number RGPIN-2024-05794).
The second author was supported by the Fonds de recherche du Qu\'ebec - Nature et Technologies (Grant number 257320) and Faculté des sciences de l'Université du Québec à Montréal (Grant number 1993-(A2022)).

\bibliographystyle{abbrvnat}
\bibliography{myBib}

\begin{thebibliography}{71}
\providecommand{\natexlab}[1]{#1}
\providecommand{\url}[1]{\texttt{#1}}
\expandafter\ifx\csname urlstyle\endcsname\relax
  \providecommand{\doi}[1]{doi: #1}\else
  \providecommand{\doi}{doi: \begingroup \urlstyle{rm}\Url}\fi

\bibitem[Abudy and Izhakian(2013)]{abudy2013pricing}
M.~Abudy and Y.~Izhakian.
\newblock Pricing stock options with stochastic interest rate.
\newblock \emph{International Journal of Portfolio Analysis and Management},
  1\penalty0 (3):\penalty0 250--277, 2013.

\bibitem[Ammann et~al.(2008)Ammann, Kind, and Wilde]{ammann2008simulation}
M.~Ammann, A.~Kind, and C.~Wilde.
\newblock Simulation-based pricing of convertible bonds.
\newblock \emph{Journal of Empirical Finance}, 15\penalty0 (2):\penalty0
  310--331, 2008.

\bibitem[Ayache et~al.(2003)Ayache, Forsyth, and Vetzal]{ayache2003valuation}
E.~Ayache, P.~A. Forsyth, and K.~R. Vetzal.
\newblock {Valuation of Convertible Bonds with Credit Risk}.
\newblock \emph{The Journal of Derivatives}, 11\penalty0 (1):\penalty0 9--29,
  2003.

\bibitem[Barone-Adesi et~al.(2003)Barone-Adesi, Bermudez, and
  Hatgioannides]{barone2003two}
G.~Barone-Adesi, A.~Bermudez, and J.~Hatgioannides.
\newblock Two-factor convertible bonds valuation using the method of
  characteristics/finite elements.
\newblock \emph{Journal of Economic Dynamics and Control}, 27\penalty0
  (10):\penalty0 1801--1831, 2003.

\bibitem[Battauz and Rotondi(2022)]{battauz2022american}
A.~Battauz and F.~Rotondi.
\newblock {American options and stochastic interest rates}.
\newblock \emph{Computational Management Science}, 19:\penalty0 567--604, 2022.

\bibitem[Batten et~al.(2014)Batten, Khaw, and Young]{batten2014convertible}
J.~A. Batten, K.~L.-H. Khaw, and M.~R. Young.
\newblock Convertible bond pricing models.
\newblock \emph{Journal of Economic Surveys}, 28\penalty0 (5):\penalty0
  775--803, 2014.

\bibitem[Batten et~al.(2018)Batten, Khaw, and Young]{batten2018pricing}
J.~A. Batten, L.-H.~K. Khaw, and M.~R. Young.
\newblock Pricing convertible bonds.
\newblock \emph{Journal of Banking \& Finance}, 92:\penalty0 216--236, 2018.

\bibitem[Ben-Ameur et~al.(2007)Ben-Ameur, Breton, Karoui, and
  L’Ecuyer]{ben2007dynamic}
H.~Ben-Ameur, M.~Breton, L.~Karoui, and P.~L’Ecuyer.
\newblock A dynamic programming approach for pricing options embedded in bonds.
\newblock \emph{Journal of Economic Dynamics and Control}, 31\penalty0
  (7):\penalty0 2212--2233, 2007.

\bibitem[Billingsley(1999)]{billingsley1999convergence}
P.~Billingsley.
\newblock \emph{{Convergence of Probability Measures}}.
\newblock John Wiley \& Sons, Inc., {2nd} edition, 1999.

\bibitem[Bj{\"o}rk(2009)]{bjork2009}
T.~Bj{\"o}rk.
\newblock \emph{{Arbitrage Theory in Continuous Time}}.
\newblock Oxford University Press, {3rd} edition, 2009.

\bibitem[Black and Karasinski(1991)]{blackKarasinski1991bond}
F.~Black and P.~Karasinski.
\newblock {Bond and Option Pricing when Short Rates are Lognormal}.
\newblock \emph{Financial Analysts Journal}, 47\penalty0 (4):\penalty0 52--59,
  1991.

\bibitem[Black et~al.(1990)Black, Derman, and Toy]{blackDermanToy1990one}
F.~Black, E.~Derman, and W.~Toy.
\newblock {A One-Factor Model of Interest Rates and Its Application to Treasury
  Bond Options}.
\newblock \emph{Financial Analysts Journal}, 46\penalty0 (1):\penalty0 33--39,
  1990.

\bibitem[Brennan and Schwartz(1977)]{brennan1977convertible}
M.~J. Brennan and E.~S. Schwartz.
\newblock {Convertible Bonds: Valuation and Optimal Strategies for Call and
  Conversion}.
\newblock \emph{The Journal of Finance}, 32\penalty0 (5):\penalty0 1699--1715,
  1977.

\bibitem[Brigo and Mercurio(2006)]{brigoMercurio2006}
D.~Brigo and F.~Mercurio.
\newblock \emph{{Interest Rate Models -- Theory and Practice: With Smile,
  Inflation and Credit}}.
\newblock Springer, {2nd} edition, 2006.

\bibitem[B{\"u}ttler and Waldvogel(1996)]{buttler1996pricing}
H.-J. B{\"u}ttler and J.~Waldvogel.
\newblock {Pricing Callable Bonds by Means of Green's Function}.
\newblock \emph{Mathematical Finance}, 6\penalty0 (1):\penalty0 53--88, 1996.

\bibitem[Chambers and Lu(2007)]{chambers2007tree}
D.~R. Chambers and Q.~Lu.
\newblock {A Tree Model for Pricing Convertible Bonds with Equity, Interest
  Rate, and Default Risk}.
\newblock \emph{The Journal of Derivatives}, 14\penalty0 (4):\penalty0 25--46,
  2007.

\bibitem[Chourdakis(2004)]{chourdakis2004non}
K.~Chourdakis.
\newblock {Non-Affine Option Pricing}.
\newblock \emph{The Journal of Derivatives}, 11\penalty0 (3):\penalty0 10--25,
  2004.

\bibitem[Cox et~al.(1979)Cox, Ross, and Rubinstein]{cox1979option}
J.~C. Cox, S.~A. Ross, and M.~Rubinstein.
\newblock {Option pricing: A simplified approach}.
\newblock \emph{Journal of Financial Economics}, 7\penalty0 (3):\penalty0
  229--263, 1979.

\bibitem[Cox et~al.(1985)Cox, Ingersoll~Jr., and Ross]{CIR1985}
J.~C. Cox, J.~E. Ingersoll~Jr., and S.~A. Ross.
\newblock {A Theory of the Term Structure of Interest Rates}.
\newblock \emph{Econometrica}, 53\penalty0 (2):\penalty0 385--407, 1985.

\bibitem[Crimaldi and Pratelli(2005)]{crimaldi2005convergence}
I.~Crimaldi and L.~Pratelli.
\newblock Convergence results for conditional expectations.
\newblock \emph{Bernoulli}, 11\penalty0 (4):\penalty0 737--745, 2005.

\bibitem[Cui(2013)]{cui2013martingale}
Z.~Cui.
\newblock \emph{{Martingale Property and Pricing for Time-Homogeneous Diffusion
  Models in Finance}}.
\newblock PhD thesis, University of Waterloo, 2013.

\bibitem[Cui et~al.(2018)Cui, Kirkby, and Nguyen]{cui2018general}
Z.~Cui, J.~L. Kirkby, and D.~Nguyen.
\newblock {A General Valuation Framework for SABR and Stochastic Local
  Volatility Models}.
\newblock \emph{SIAM Journal on Financial Mathematics}, 9\penalty0
  (2):\penalty0 520--563, 2018.

\bibitem[Cui et~al.(2019)Cui, Kirkby, and Nguyen]{cui2019continuous}
Z.~Cui, J.~L. Kirkby, and D.~Nguyen.
\newblock {Continuous-Time Markov Chain and Regime Switching Approximations
  with Applications to Options Pricing}.
\newblock In G.~Yin and Q.~Zhang, editors, \emph{{Modeling, Stochastic Control,
  Optimization, and Applications}}, pages 115--146. Springer, 2019.

\bibitem[Cui et~al.(2021)Cui, Kirkby, and Nguyen]{cui2021efficient}
Z.~Cui, J.~L. Kirkby, and D.~Nguyen.
\newblock {Efficient simulation of generalized SABR and stochastic local
  volatility models based on Markov chain approximations}.
\newblock \emph{European Journal of Operational Research}, 290\penalty0
  (3):\penalty0 1046--1062, 2021.

\bibitem[D'Halluin et~al.(2001)D'Halluin, Forsyth, Vetzal, and
  Labahn]{dHalluin2001numerical}
Y.~D'Halluin, P.~Forsyth, K.~Vetzal, and G.~Labahn.
\newblock {A numerical PDE approach for pricing callable bonds}.
\newblock \emph{Applied Mathematical Finance}, 8\penalty0 (1):\penalty0 49--77,
  2001.

\bibitem[Diener and Diener(2004)]{diener2004asymptotics}
F.~Diener and M.~Diener.
\newblock Asymptotics of the price oscillations of a european call option in a
  tree model.
\newblock \emph{Mathematical Finance}, 14\penalty0 (2):\penalty0 271--293,
  2004.

\bibitem[Ding et~al.(2012)Ding, Fu, and So]{ding2012pricing}
D.~Ding, Q.~Fu, and J.~So.
\newblock Pricing callable bonds based on monte carlo simulation techniques.
\newblock \emph{Technology and Investment}, 3:\penalty0 121--125, 2012.

\bibitem[Ding and Ning(2021)]{ding2021markov}
K.~Ding and N.~Ning.
\newblock Markov chain approximation and measure change for time-inhomogeneous
  stochastic processes.
\newblock \emph{Applied Mathematics and Computation}, 392:\penalty0 125732,
  2021.

\bibitem[Dothan(1978)]{dothan1978term}
L.~Dothan.
\newblock On the term structure of interest rates.
\newblock \emph{Journal of Financial Economics}, 6\penalty0 (1):\penalty0
  59--69, 1978.

\bibitem[Duffy(2006)]{duffy2013finite}
D.~J. Duffy.
\newblock \emph{{Finite Difference Methods in Financial Engineering: A Partial
  Differential Equation Approach}}.
\newblock John Wiley \& Sons, 2006.

\bibitem[Fu et~al.(2001)Fu, Laprise, Madan, Su, and Wu]{fu2001pricing}
M.~C. Fu, S.~B. Laprise, D.~B. Madan, Y.~Su, and R.~Wu.
\newblock {Pricing American options: A Comparison of Monte Carlo Simulation
  Approaches}.
\newblock \emph{Journal of Computational Finance}, 4\penalty0 (3):\penalty0
  39--88, 2001.

\bibitem[Geman et~al.(1995)Geman, El~Karoui, and Rochet]{geman1995changes}
H.~Geman, N.~El~Karoui, and J.-C. Rochet.
\newblock Changes of num{\'e}raire, changes of probability measure and option
  pricing.
\newblock \emph{Journal of Applied Probability}, 32\penalty0 (2):\penalty0
  443--458, 1995.

\bibitem[Glasserman(2003)]{glasserman2003monte}
P.~Glasserman.
\newblock \emph{{Monte Carlo Methods in Financial Engineering}}, volume~53.
\newblock Springer, 2003.

\bibitem[Goggin(1994)]{goggin1994convergence}
E.~M. Goggin.
\newblock {Convergence in Distribution of Conditional Expectations}.
\newblock \emph{The Annals of Probability}, pages 1097--1114, 1994.

\bibitem[Gushchin and Curien(2008)]{gushchin2008pricing}
V.~Gushchin and E.~Curien.
\newblock {The pricing of convertible bonds within the Tsiveriotis and
  Fernandes framework with exogenous credit spread: Empirical analysis}.
\newblock \emph{Journal of Derivatives \& Hedge Funds}, 14:\penalty0 50--65,
  2008.

\bibitem[Ho and Lee(1986)]{HoLee1986term}
T.~S.~Y. Ho and S.-B. Lee.
\newblock {Term Structure Movements and Pricing Interest Rate Contingent
  Claims}.
\newblock \emph{The Journal of Finance}, 41\penalty0 (5):\penalty0 1011--1029,
  1986.

\bibitem[Hull and White(1990)]{hull1990pricing}
J.~Hull and A.~White.
\newblock {Pricing Interest-Rate-Derivative Securities}.
\newblock \emph{The Review of Financial Studies}, 3\penalty0 (4):\penalty0
  573--592, 1990.

\bibitem[Hull and White(1994)]{hull1994numerical}
J.~Hull and A.~White.
\newblock {Numerical Procedures for Implementing Term Structure Models I:
  Single-Factor Models}.
\newblock \emph{The Journal of Derivatives}, 2\penalty0 (1):\penalty0 7--16,
  1994.

\bibitem[Hull and White(1996)]{hull1996using}
J.~Hull and A.~White.
\newblock {Using Hull-White Interest Rate Trees}.
\newblock \emph{The Journal of Derivatives}, 3\penalty0 (3):\penalty0 26--36,
  1996.

\bibitem[Hung and Wang(2002)]{hung2002pricing}
M.-W. Hung and J.-Y. Wang.
\newblock {Pricing Convertible Bonds Subject to Default Risk}.
\newblock \emph{The journal of Derivatives}, 10\penalty0 (2):\penalty0 75--87,
  2002.

\bibitem[Ingersoll~Jr.(1977)]{ingersoll1977contingent}
J.~E. Ingersoll~Jr.
\newblock A contingent-claims valuation of convertible securities.
\newblock \emph{Journal of Financial Economics}, 4\penalty0 (3):\penalty0
  289--321, 1977.

\bibitem[Jarrow and Turnbull(1995)]{jarrow1995pricing}
R.~A. Jarrow and S.~M. Turnbull.
\newblock {Pricing Derivatives on Financial Securities Subject to Credit Risk}.
\newblock \emph{The Journal of Finance}, 50\penalty0 (1):\penalty0 53--85,
  1995.

\bibitem[Jourdain(2004)]{jourdain2004loss}
B.~Jourdain.
\newblock Loss of martingality in asset price models with lognormal stochastic
  volatility.
\newblock {Preprint} 267, {CERMICS}, 2004.

\bibitem[Kirkby(2023)]{kirkby2022hybrid}
J.~L. Kirkby.
\newblock {Hybrid equity swap, cap, and floor pricing under stochastic interest
  by Markov chain approximation}.
\newblock \emph{European Journal of Operational Research}, 305\penalty0
  (2):\penalty0 961--978, 2023.

\bibitem[Kirkby et~al.(2020)Kirkby, Nguyen, and Nguyen]{kirkby2020general}
J.~L. Kirkby, D.~H. Nguyen, and D.~Nguyen.
\newblock {A general continuous time Markov chain approximation for multi-asset
  option pricing with systems of correlated diffusions}.
\newblock \emph{Applied Mathematics and Computation}, 386:\penalty0 125472,
  2020.

\bibitem[Kouritzin and Zeng(2005)]{kouritzin2005weak}
M.~A. Kouritzin and Y.~Zeng.
\newblock Weak convergence for a type of conditional expectation: application
  to the inference for a class of asset price models.
\newblock \emph{Nonlinear Analysis: Theory, Methods \& Applications},
  60\penalty0 (2):\penalty0 231--239, 2005.

\bibitem[Lamberton(1998)]{lamberton1998}
D.~Lamberton.
\newblock {American Options}.
\newblock In D.~J. Hand and S.~D. Jacka, editors, \emph{{Statistic in
  Finance}}. Arnold, 1998.

\bibitem[Li and Zhang(2018)]{li2018error}
L.~Li and G.~Zhang.
\newblock {Error analysis of finite difference and Markov chain approximations
  for option pricing}.
\newblock \emph{Mathematical Finance}, 28\penalty0 (3):\penalty0 877--919,
  2018.

\bibitem[Lin and Zhu(2020)]{lin2020numerically}
S.~Lin and S.-P. Zhu.
\newblock Numerically pricing convertible bonds under stochastic volatility or
  stochastic interest rate with an adi-based predictor--corrector scheme.
\newblock \emph{Computers \& Mathematics with Applications}, 79\penalty0
  (5):\penalty0 1393--1419, 2020.

\bibitem[Lin and Zhu(2022)]{lin2022pricing}
S.~Lin and S.-P. Zhu.
\newblock {Pricing callable--puttable convertible bonds with an integral
  equation approach}.
\newblock \emph{Journal of Futures Markets}, 42\penalty0 (10):\penalty0
  1856--1911, 2022.

\bibitem[Lo and Skindilias(2014)]{lo2014improved}
C.~C. Lo and K.~Skindilias.
\newblock {An improved Markov chain approximation methodology: Derivatives
  pricing and model calibration}.
\newblock \emph{International Journal of Theoretical and Applied Finance},
  17\penalty0 (7):\penalty0 1450047, 2014.

\bibitem[Longstaff and Schwartz(2001)]{longstaff2001valuing}
F.~A. Longstaff and E.~S. Schwartz.
\newblock {Valuing American Options by Simulation: A Simple Least-Squares
  Approach}.
\newblock \emph{The Review of Financial Studies}, 14\penalty0 (1):\penalty0
  113--147, 2001.

\bibitem[Lu and Xu(2017)]{lu2017simple}
L.~Lu and W.~Xu.
\newblock {A Simple and Efficient Two-Factor Willow Tree Method for Convertible
  Bond Pricing with Stochastic Interest Rate and Default Risk}.
\newblock \emph{The Journal of Derivatives}, 25\penalty0 (1):\penalty0 37--54,
  2017.

\bibitem[Ma et~al.(2020)Ma, Xu, , and Yuan]{ma2020valuation}
C.~Ma, W.~Xu, , and G.~Yuan.
\newblock {Valuation model for Chinese convertible bonds with soft call/put
  provision under the hybrid willow tree}.
\newblock \emph{Quantitative Finance}, 20\penalty0 (12):\penalty0 2037--2053,
  2020.

\bibitem[Ma et~al.(2022)Ma, Yang, and Cui]{ma2022convergence}
J.~Ma, W.~Yang, and Z.~Cui.
\newblock {Convergence analysis for continuous-time Markov chain approximation
  of stochastic local volatility models: Option pricing and Greeks}.
\newblock \emph{Journal of Computational and Applied Mathematics}, 404, 2022.

\bibitem[MacKay and Vachon(2023)]{mackay2023optimal}
A.~MacKay and M.-C. Vachon.
\newblock On an optimal stopping problem with a discontinuous reward.
\newblock \emph{arXiv preprint arXiv:2311.03538}, 2023.

\bibitem[MacKay et~al.(2023)MacKay, Vachon, and Cui]{mackay2023analysis}
A.~MacKay, M.-C. Vachon, and Z.~Cui.
\newblock {Analysis of VIX-linked fee incentives in variable annuities via
  continuous-time Markov chain approximation}.
\newblock \emph{Quantitative Finance}, 23\penalty0 (7-8):\penalty0 1055--1078,
  2023.

\bibitem[McConnell and Schwartz(1986)]{mcconnell1986lyon}
J.~J. McConnell and E.~S. Schwartz.
\newblock {LYON taming}.
\newblock \emph{The Journal of Finance}, 41\penalty0 (3):\penalty0 561--576,
  1986.

\bibitem[Mentink-Vigier()]{matlabFastExpm}
F.~Mentink-Vigier.
\newblock {Fast exponential matrix for Matlab (full/sparse), fastExpm}.
\newblock Available at Github \url{https://github.com/fmentink/fastExpm}.
\newblock Accessed: April 27, 2023.

\bibitem[Mercurio and Moraleda(2001)]{mercurio2001family}
F.~Mercurio and J.~M. Moraleda.
\newblock A family of humped volatility models.
\newblock \emph{The European Journal of Finance}, 7\penalty0 (2):\penalty0
  93--116, 2001.

\bibitem[Mijatovi{\'c} and Pistorius(2009)]{mijatovic2009continuously}
A.~Mijatovi{\'c} and M.~Pistorius.
\newblock {Continuously monitored barrier options under Markov processes
  (unabridged version with Matlab code)}.
\newblock Available at SSRN 1462822, 2009.

\bibitem[Mijatovi{\'c} and Pistorius(2013)]{mijatovic2013continuously}
A.~Mijatovi{\'c} and M.~Pistorius.
\newblock {Continuously monitored barrier options under Markov processes}.
\newblock \emph{Mathematical Finance}, 23\penalty0 (1):\penalty0 1--38, 2013.

\bibitem[Milanov et~al.(2013)Milanov, Kounchev, Fabozzi, Kim, and
  Rachev]{milanov2012binomial}
K.~Milanov, O.~Kounchev, F.~J. Fabozzi, Y.~S. Kim, and S.~T. Rachev.
\newblock {A Binomial-Tree Model for Convertible Bond Pricing}.
\newblock \emph{The Journal of Fixed Income}, 22\penalty0 (3):\penalty0 79--94,
  2013.

\bibitem[Ostrovski(2013)]{Ostrovski2013Efficient}
V.~Ostrovski.
\newblock {Efficient and Exact Simulation of the Hull-White Model}.
\newblock Available at SSRN 2304848, 2013.

\bibitem[Rindos et~al.(1995)Rindos, Woolet, Viniotis, and
  Trivedi]{rindos1995exact}
A.~Rindos, S.~Woolet, I.~Viniotis, and K.~Trivedi.
\newblock {Exact Methods for the Transient Analysis of Nonhomogeneous
  Continuous Time Markov Chains}.
\newblock In W.~J. Stewart, editor, \emph{Computations with Markov Chains},
  pages 121--133. Springer, 1995.

\bibitem[Sin(1998)]{sin1998complications}
C.~A. Sin.
\newblock {Complications with stochastic volatility models}.
\newblock \emph{Advances in Applied Probability}, 30\penalty0 (1):\penalty0
  256--268, 1998.

\bibitem[Song et~al.(2013)Song, Yin, and Zhang]{song2013weak}
Q.~Song, G.~Yin, and Q.~Zhang.
\newblock {Weak Convergence Methods for Approximation of the Evaluation of
  Path-Dependent Functionals}.
\newblock \emph{SIAM Journal on Control and Optimization}, 51\penalty0
  (5):\penalty0 4189--4210, 2013.

\bibitem[Tavella and Randall(2000)]{tavellapricing}
D.~Tavella and C.~Randall.
\newblock {Pricing Financial Instruments: The Finite Difference Method}.
\newblock \emph{John Willey \& Sons}, 2000.

\bibitem[Tsiveriotis and Fernandes(1998)]{tsiveriotis1998valuing}
K.~Tsiveriotis and C.~Fernandes.
\newblock {Valuing Convertible Bonds with Credit Risk}.
\newblock \emph{The Journal of Fixed Income}, 8\penalty0 (2):\penalty0 95--102,
  1998.

\bibitem[Vasicek(1977)]{vasicek1977equilibrium}
O.~Vasicek.
\newblock {An equilibrium characterization of the term structure}.
\newblock \emph{Journal of Financial Economics}, 5\penalty0 (2):\penalty0
  177--188, 1977.

\bibitem[Zhang and Li(2019)]{zhang2019analysis}
G.~Zhang and L.~Li.
\newblock {Analysis of Markov Chain Approximation for Option Pricing and
  Hedging: Grid Design and Convergence Behavior}.
\newblock \emph{Operations Research}, 67\penalty0 (2):\penalty0 407--427, 2019.

\end{thebibliography}

\appendix
\section{Proofs}\label{appendixProof}
\subsection{Proof of Lemma \ref{lemmaZeroCouponBondApprox}}
Without loss of generality, suppose that $i=0$ and that $\tN=N$. Since $R^{(m)}$ is a discrete random process whose transitional probabilities $p_{k_1k_2}(t_{n-1},t_{n})$, $1 \leq k_1,k_2\leq m$, $1\leq n\leq N$ are given by ${\mathbf{P}(t_{n-1},t_{n})=\exp\left\{\mathbf{Q}_{n}^{(m)}\Delta_N\right\}}$ as per \eqref{eqTransProb}, we have that
\begin{align*}
&\e\left[e^{-\sum_{n=1}^{N} R_{t_n}^{(m)}\Delta_N} \big| R_0^{(m)}=r_j\right] \\
&\quad= \sum_{k_1,k_2,\ldots, k_N =1}^{m} e^{-\sum_{l=1}^N r_{k_l}\Delta_N}p_{jk_1}(t_0,t_1)p_{k_1k_2}(t_1,t_2)\ldots p_{k_{N-1}k_N}(t_{N-1},t_N)\\
		&=\sum_{k_1,k_2,\ldots, k_N =1}^{m} p_{jk_1}(t_0,t_1)e^{-r_{k_1}\Delta_N}p_{k_1k_2}(t_1,t_2)e^{-r_{k_2}\Delta_N}\ldots
        p_{k_{N-1}k_N}(t_{N-1},t_N)e^{-r_{k_N}\Delta_N}\\
		&=\mathbf{e}_j \left(\prod_{n=1}^{N} \mathbf{P}(t_{n-1},t_{n})e^{-\mathbf{D}_m\Delta_N}\right)\mathbf{1}_{m\times 1}\\
        &= \mathbf{e}_j \left(\prod_{n=1}^{N} e^{\mathbf{Q}_n^{(m)}\Delta_N}e^{-\mathbf{D}_m\Delta_N}\right)\mathbf{1}_{m\times 1},
\end{align*}
where the fourth line is the matrix representation of the third line. 
The result follows similarly when $\tN=kN$.
The proof for \eqref{eqExpSumApproxInAdvance} follows using  analogous arguments.
\hfill\qed
\subsection{Proof of Proposition \ref{propZeroCouponBondApprox}}
Using the notation of Lemma \ref{lemmaZeroCouponBondApprox}, we have that
	\begin{align*}
		 & \e\left[e^{-\int_{t_i}^T R_s^{(m)}\diff s} \big| R_{t_i}^{(m)}=r_j\right] \\
   & \qquad\qquad= \e\left[\lim_{\tN\rightarrow\infty}e^{-\sum_{n=ki+1}^{\tN} R_{\tilde{t}_n}^{(m)}\Delta_{\tN}} \big| R_{\tilde{t}_{ki}}^{(m)}=r_j\right]\\
		& \qquad\qquad= \lim_{\tN\rightarrow\infty}\e\left[e^{-\sum_{n=ki+1}^{\tN} R^{(m)}_{\tilde{t}_n}\Delta_{\tN} } | R_{\tilde{t}_{ki}}^{(m)}=r_j\right]\textrm{ (by the dominated convergence theorem)}\\
		& \qquad\qquad = \lim_{\tN\rightarrow\infty}   \mathbf{e}_{j} \left(\prod_{n=i+1}^{N} \left(e^{\mathbf{Q}_{n}^{(m)}\Delta_{\tN}}e^{\mathbf{-D}\Delta_{\tN}}\right)^{\frac{\tN}{N}}\right)\mathbf{1}_{m\times 1} \textrm{ (by \eqref{eqExpSumApprox})}\\
		& \qquad\qquad = \mathbf{e}_j \left(\prod_{n=i+1}^{N} e^{(\mathbf{Q}_n^{(m)}-\mathbf{D})\Delta_N}\right) \mathbf{1}_{m\times 1}\textrm{ (by the Lie product formula).}  
  \end{align*}

\subsection{Extension of Proposition \ref{propCBtrivialTF} and Corollary \ref{corrCBsUpperBound}}\label{appendixExtensionPropTrivialCB}
This section shows that the results of Proposition \ref{propCBtrivialTF} and Corollary \ref{corrCBsUpperBound} still hold when periodic coupons are paid. 
We suppose that $\tilde{N}>0$ periodic coupons $\alpha>0$ are paid at time $0<t_{z}<t_{2z}<\ldots<t_{z\tilde{N}}=T$, with $z=N/\tilde{N}$. The time-$t$ risk-neutral value of a European-style CB is given by
\begin{equation}\label{eq:CBeuroTFcoupon}
\begin{split}
     v_e^\alpha(t,x,r)  &:=\e \Big[e^{-\int_t^T R_u\diff u} \eta S_T \ind_{\{S_T\geq F/\eta\}} 
     + e^{-\int_t^T R_u+c_u\diff u} F  \ind_{\{S_T< F/\eta\}} \big| S_t=x, R_t=r\Big] \\
      & \quad\quad\quad+ \alpha\sum_{n=1}^{\tilde{N}} \e\left[e^{-\int_t^{t_{zn}} R_u + c_u \diff u}\Big|S_t=x, R_t=r\right]\ind_{\{t_{zn}> t\}}\\
      &=v_e(t,x,r)+\alpha\sum_{n=1}^{\tilde{N}} e^{-\int_t^{t_{zn}} c_u\diff u} P_r(t,t_{zn})\ind_{\{t_{zn}> t\}},
 \end{split}
\end{equation}
where $P_r(t,t_{zn}):=\e\left[e^{-\int_t^{t_{zn}} R_u\diff u}\big|R_t=r\right]$ with $P(t,t_{zn})=1$ whenever $t_{zn}\leq t$. Hence, the value of European-style CBs with coupons is equal to the value of European-style CBs without coupons to which the present value of future coupons is added. When conversion can occur at any time prior to maturity, the time-$t$ risk-neutral value of the CB is given by
\begin{align}\label{eq:valueFctCBamericanCoupon}
    \begin{split}
    v^\alpha(t,x,r)&:=\sup_{\tau\in\mathcal{T}_{t,T}}\e\Bigg[e^{-\int_t^\tau R_u +c_u \ind_{\{\tau=T,S_T< F/\eta\}}\diff u}\varphi(\tau,S_{\tau})\\
    &\qquad\qquad\qquad+\alpha \sum_{n=1}^{\tilde{N}}e^{-\int_t^{t_{zn}} R_u + c_u \diff u}\ind_{\{t< t_{zn}\leq \tau\}}\Big|S_t=x, R_t=r\Bigg],
    \end{split}
\end{align}
with function $\varphi$ defined in \eqref{eq:rewardFctCB}.

The next corollary extends Proposition \ref{propCBtrivialTF} to coupon-bearing bonds. Under the assumption that no dividend yield is paid out and credit risk is nil, the value of coupon-bearing American-type CBs is equivalent to the value of coupon-bearing European-type CBs.
\begin{corollary}\label{corCBtrivialTFcoupon}
\begin{sloppypar}
    Assume $q_t=c_t=0$ for all $t\in[0,T]$ and $\sigma_S(r)=\tilde{\sigma}_S>0$ for all $r\in\mathcal{S}_R$. We have that ${v^\alpha(t,x,r)=v_e^\alpha(t,x,r)}$ for all $(t,x,r)\in [0,T]\times \reals_+^\star\times\mathcal{S}_R$. 
\end{sloppypar}
\end{corollary}
\begin{proof}
Define process $\left\{\tilde{Z}_s:=\alpha\sum_{n=1}^{\tilde{N}}  e^{-\int_t^{t_{zn}} R_u \diff u}\ind_{\{t< t_{zn}\leq s\}}\right\}_{t\leq s\leq T}$, representing the discounted value of future coupons, and note that
    \begin{align*}
         v^\alpha(t,x,r)&\leq \sup_{\tau\in\mathcal{T}_{t,T}}\e\Bigg[e^{-\int_t^\tau R(u) \diff u}\varphi(\tau,S_{\tau}) \Big|S_t=x, R_t=r\Bigg]\\
    &\quad\quad\quad+\sup_{\tau\in\mathcal{T}_{t,T}}\e\left[\tilde{Z}_\tau\Big|S_t=x, R_t=r\right]\\
    &= v_e(t,x,r)+ \e\left[ \tilde{Z}_T\Big|S_t=x, R_t=r\right]\\
    &=v_e^\alpha(t,x,r),
    \end{align*}
    where the second equality follows from Proposition \ref{propCBtrivialTF} and by noticing that $\tilde{Z}_s\uparrow \tilde{Z}_T$.
We conclude the proof by observing that $v^\alpha(t,x,r)\geq v_e^\alpha(t,x,r)$ for all $(t,x,r)\in [0,T]\times \reals_+^\star\times\mathcal{S}_R$, since $T\in\mathcal{T}_{t,T}$. 
\end{proof}
Let $\tilde{v}_e^{\alpha}:[0,T]\times\reals_+^\star\times\mathcal{S}_R\rightarrow \reals_+$ denote the value of European-style CBs (with periodic coupon $\alpha$) when $c_t=0$ for all $t\in[0,T]$. Using \eqref{eq:CBeuroTFcoupon}, it follows that
\begin{align*}
\tilde{v}_e^\alpha(t,x,r)&:=\e \Big[e^{-\int_t^T R_u\diff u} \max\left(\eta S_T,F\right) \big| S_t=x, R_t=r\Big]+ \alpha\sum_{n=1}^{\tilde{N}} P_r(t,t_{zn})\ind_{\{t_{zn}> t\}}\\
&=\tilde{v}_e(t,x,r)+\alpha\sum_{n=1}^{\tilde{N}} P_r(t,t_{zn})\ind_{\{t_{zn}> t\}},
\end{align*}
where $\tilde{v}_e$ is defined in \eqref{eqEuroCBnoCreditRisk} and represents the value of European-style CBs when no coupons are paid and when $c_t=0$ for all $t\in[0,T]$.

The next corollary extends the results of Corollary \ref{corrCBsUpperBound} to coupon-bearing bonds. It provides an upper bound for American-type CBs under the assumption that no dividends are distributed.
\begin{corollary}
\begin{sloppypar}
    Assume $q_t=0$ for all $t\geq 0$ and $\sigma_S(r)=\tilde{\sigma}_S>0$ for all $r\in\mathcal{S}_R$. We have that ${v_e^\alpha(t,x,r)\leq v^\alpha(t,x,r)\leq \tilde{v}_e^\alpha(t,x,r)}$ for all $(t,x,r)\in[0,T]\times\reals_+^\star\times\mathcal{S}_R$.
\end{sloppypar}
\end{corollary} \label{corrCBsUpperBoundcoupon}   
The proof is akin to that of Corollary \ref{corrCBsUpperBound}.

 \section{Time-Homogeneous Models} \label{appendixTimeHomoModels}
When the short-rate process is time-homogeneous (see, for instance, models listed in Table \ref{tblmodelsHomo}), the generator is time-independent, and the results obtained in Section \ref{sectionApplicationDebtSecurities} can be simplified. The construction of the generator remains the same as in Section \ref{subsectCTMC_R}, except that we now have 
    $\mathbf{Q}^{(m)}:=\mathbf{Q}^{(m)}_1=\mathbf{Q}^{(m)}_2=\ldots = \mathbf{Q}^{(m)}_{N},$
which follows because the drift and volatility parameters of $R$ are now time-independent. 
In this appendix, we summarize the different results obtained previously under the assumption that the short-rate process is time-homogeneous.

Recall that $\{\mathbf{e}_{k}\}_{k=1}^{m}$ denotes the standard basis in $\reals^{m}$, that is, $\mathbf{e}_{k}$ represents a row vector of size $1\times m$ with a value of $1$ in the $k$-th entry and $0$ elsewhere, $\mathbf{1}_{m\times1}$ denotes an $m\times 1$ unit vector, and ${\mathbf{D}_m:=\diag(\bm{r})}$, is an $m\times m$ diagonal matrix with vector $\bm{r}=(r_1,r_2,\ldots,r_m)$ on its diagonal, $r_k\in\mathcal{S}^{(m)}_R,$ $k=1,2\ldots,m$.

\subsection{Zero-Coupon Bond}\label{appendixZeroBondHomo}
The first result of this section concerns the price of zero-coupon bonds.
\begin{corollary}\label{corrZeroCouponBondHomo}
    Consider a time partition of $[0,T]$, $0=\tilde{t}_0<\tilde{t}_1<\ldots<\tilde{t}_{\tN}=T$, with $\tN\in\mathbb{N}$, $\Delta_{\tN}=T/{\tN}$ and $\tilde{t}_n=n\Delta_{\tN}$. Given that $R_{\tiT_i}^{(m)}=R_{\tiT_i}=r_j\in\mathcal{S}^{(m)}_R$, it holds that
     \begin{equation}
        \e\left[e^{-\sum_{n=i+1}^{\tN} R_{\tilde{t}_n}^{(m)}\Delta_{\tN}} \big| R_{\tilde{t}_{i}}^{(m)}=r_j\right]=  \mathbf{e}_{j} \left( e^{\mathbf{Q}^{(m)}\Delta_{\tN}}e^{-\mathbf{D}_m{\Delta_{\tN}}}\right)^{\tN-i} \times\mathbf{1}_{m\times1},\label{eqExpSumHomo1}
    \end{equation}
    and
    \begin{equation}
        \begin{split}\label{eqExpSumHomo2}
            \e\left[e^{-\sum_{n=i}^{\tN} R_{\tilde{t}_n}^{(m)}\Delta_{\tN}} \big| R_{\tilde{t}_i}^{(m)}
                  =r_j\right] & = \mathbf{e}_{j}  \left(e^{-\mathbf{D}_m\Delta_{\tN}} e^{\mathbf{Q}^{(m)}\Delta_{\tN}}\right)^{\tN-i} e^{-\mathbf{D}_m\Delta_{\tN}} \times\mathbf{1}_{m\times1}\\
                  & =\mathbf{e}_{j} e^{-\mathbf{D}_m\Delta_{\tN}}\left(e^{\mathbf{Q}^{(m)}\Delta_{\tN}}e^{-\mathbf{D}_m\Delta_{\tN}}\right)^{\tN-i}\times\mathbf{1}_{m\times1},
        \end{split}
    \end{equation}
Moreover, if Assumption \ref{assumpStateSpaceR} holds, the price at time $t\geq 0$ of a zero-coupon bond with maturity $T\geq t$ can be approximated by
	\begin{equation}
	    P_j^{(m)}(t,T):=\e\left[e^{-\int_t^T R_{s}^{(m)}\diff s} \big| R_{t}^{(m)}=r_j\right]
     = \mathbf{e}_{j} e^{\left(\mathbf{Q}^{(m)}-\mathbf{D}_m\right)(T-t)} \mathbf{1}_{m\times1},\label{eqZeroCouponHomo}
	\end{equation}
 given that $R_{t}^{(m)}=R_t =r_j\in\mathcal{S}^{(m)}_R$.
\end{corollary}
The proof follows by setting $\tN=N$ and $\mathbf{Q}^{(m)}=\mathbf{Q}^{(m)}_1=\mathbf{Q}^{(m)}_2=\ldots = \mathbf{Q}^{(m)}_{N}$ in Lemma \ref{lemmaZeroCouponBondApprox}, Remark \ref{rmkCommutZeroCoupon}, and Proposition \ref{propZeroCouponBondApprox}.
\eqref{eqZeroCouponHomo} was previously obtained by \cite{cui2018general} Proposition 8 (ii) and \cite{kirkby2022hybrid}, Proposition 3, whereas the first equality of \eqref{eqExpSumHomo2} is provided, in a more general form, in \cite{cui2018general} Proposition 8 (i). The proof in Appendix \ref{appendixProof} differs from these previous proofs and provides a simple and intuitive way of obtaining these results using basic probabilistic arguments. 
\subsection{Zero-Coupon Bond Option}\label{appendixZeroBondOptionHomo}
The next result concerns the price of European call and put options on zero-coupon bonds.
\begin{corollary}\label{corCallPutHomo}
   Let Assumption \ref{assumpStateSpaceR} hold. Given that $R^{(m)}_{t_{n_1}}=r_j\in\mathcal{S}_R^{(m)}$, the price at $t_{n1}\geq 0 $ of a European call (resp. put) option with maturity $t_{n_2}>t_{n_1}$ on a zero-coupon bond maturing at time $T>t_{n_2}$ with a strike $K>0$ can be approximated by
    \begin{equation}\label{eqCallPutHomo}
        \e\left[e^{-\int_{t_{n_1}}^{t_{n_2}} R_s^{(m)}\diff s }h\left(P^{(m)}(t_{n_2},T)\right) \Big|R^{(m)}_{t_{n_1}}=r_j\right]= \mathbf{e}_{j}\left(e^{\left(\mathbf{Q}^{(m)}-\mathbf{D}_m\right)(t_{n_2}-t_{n_1})}\right)\mathbf{H},
    \end{equation}
    \begin{sloppypar}
     where $h(x)=\max(x-K,0)$ (resp. $h(x)=\max(K-x,0)$) denotes the payoff function, ${P^{(m)}(t_{n_2},T):=\e\left[e^{-\int_{t_{n_2}}^TR_s^{(m)}\diff s}\Big| R_{t_{n_2}}^{(m)}\right]}$ denotes the approximated zero-coupon bond price at $t_{n_2}$, and $\mathbf{H}$ denotes a column vector of size $m\times 1$ whose $k$-th, $h_k$, entry is given by 
       $$h_k= h\left(P_k^{(m)}(t_{n_2},T)\right),$$
       with $P_k^{(m)}(t_{n_2},T)$ defined in \eqref{eqZeroCouponHomo}.
    \end{sloppypar}
\end{corollary}   

\section{Extended Models of \cite{brigoMercurio2006}}\label{appendixBrigoModels}
In this section, we summarize the results of Section \ref{sectionApplicationDebtSecurities} under particular time-inhomogeneous short-rate models, for which the short-rate process is obtained by a time-deterministic shift of a time-homogeneous auxiliary process. More precisely, we suppose that the short-rate process can be decomposed as 
\begin{equation}
R_t= Y_t + \theta(t),\label{eqShortRateExtendedModels}
\end{equation}
for $t\geq 0$, where $\theta$ denotes a continuous deterministic function of time and $Y$ denotes an auxiliary time-homogeneous diffusion process with the following dynamics:
\begin{equation}
  \diff Y_t = \mu_Y(Y_t)\diff t + \sigma_Y(Y_t)\diff W_t, \label{eqEDS_Y}  
\end{equation}
where $\mu_Y,\sigma_Y:\mathcal{S}_Y\rightarrow \reals$ are well-behaved functions such that \eqref{eqEDS_Y} has a unique in-law weak solution with $\mathcal{S}_Y$ the state-space of $Y$. Examples of such diffusion processes are listed in Table \ref{tblmodelsNonHomoBrigo}.


In this section, the auxiliary process $Y$ is approximated by a CTMC. We denote by $Y^{(m)}$ the CTMC approximation of $Y$ taking values in a finite state-space ${\mathcal{S}_Y^{(m)}=\{y_1,y_2,\ldots,y_m\}}$. The time-independent generator of $Y^{(m)}$, denoted by $\mathbf{Q}^{(m)}_Y$, is constructed as in \eqref{eq:generatorR}, with functions $\mu_R$ and $\sigma_R$ replaced by functions $\mu_Y$ and $\sigma_Y$ of \eqref{eqEDS_Y}, respectively. Moreover, since the auxiliary process is time-homogenous, the generator of $Y^{(m)}$ does not depend on time, such that
$\mathbf{Q}^{(m)}_Y:=\mathbf{Q}^{(m)}_1=\mathbf{Q}^{(m)}_2=\ldots=\mathbf{Q}^{(m)}_N.$
Using \eqref{eqShortRateExtendedModels}, the CTMC approximation of the short-rate process $R^{(m)}$ is given by
     $$ R^{(m)}_t=  Y^{(m)}_t+\theta(t),$$ 
  for $t\geq 0$, with $\theta(0)=0$.
\begin{remark}[Weak convergence of the approximation]\label{rmkConvCTMCBrigo}
     The weak convergence of $Y^{(m)}$ to $Y$ follows from Theorem 5.1 of \cite{mijatovic2013continuously}. Then, since $R$ is a continuous transformation of $Y$, we conclude that $R^{(m)}\Rightarrow R$ by the continuous mapping Theorem.
\end{remark} 

Recall that $\{\mathbf{e}_{k}\}_{k=1}^{m}$ denotes the standard basis in $\reals^{m}$, that is, $\mathbf{e}_{k}$ represents a row vector of size $1\times m$ with a value of $1$ in the $k$-th entry and $0$ elsewhere, $\mathbf{1}_{m\times1}$ denotes an $m\times 1$ unit vector, and $\mathbf{D}_Y:=\diag(\bm{y})$ denotes an $m\times m$ diagonal matrix with vector $\bm{y}=(y_1,y_2,\ldots,y_m)$ on its diagonal, with $y_k\in\mathcal{S}_Y^{(m)}$, $k=1,2,\ldots,m$.

\subsection{Zero-Coupon Bond}\label{appendixZeroBondBrigo}
The next corollary shows that the price of zero-coupon bonds inherits the analytical tractability of the homogeneous auxiliary process $Y$. 
\begin{corollary}\label{corrZeroCouponBrigo}
Suppose that there exists $y^\star\in\mathcal{S}_Y$ such that $Y_t\geq y^\star$ for all $t\geq 0$. Then, the price at time $t\geq 0$ of a zero-coupon bond with maturity $T\geq t$ can be approximated by   
 \begin{equation}
     P_j^{(m)}(t,T):=\e[e^{-\int_t^TR^{(m)}_s\diff s}\big | R_t^{(m)}=y_j+\theta(t)]=e^{-\int_t^T \theta(s)\diff s}\tilde{P}_j^{(m)}(t,T),\label{eqZeroCouponBrigo}
 \end{equation}
 with 
 \begin{equation}
	    \tilde{P}_{j}^{(m)}(t,T):=\e\left[e^{-\int_t^T Y_{s}^{(m)}\diff s} \big| Y_{t}^{(m)}=y_j\right]
     = \mathbf{e}_{j} e^{\left(\mathbf{Q}^{(m)}_Y-\mathbf{D}_Y\right)(T-t)} \mathbf{1}_{m\times 1},\label{eqZeroCouponHomoBrigo}
	\end{equation}
where $Y_t=Y_{t}^{(m)}=y_{j}\in\mathcal{S}_Y^{(m)}$.
\end{corollary}
The proof follows directly from Corollary \ref{corrZeroCouponBondHomo}.

\subsection{Zero-Coupon Bond Option}\label{appendixZeroBondOptionBrigo}
The next result concerns the price of European call and put options on zero-coupon bonds. Using the homogeneous property of the auxiliary process, we can obtain a simplified expression for \eqref{eqCallPutZero}. 
\begin{corollary}\label{corrCallPutBrigo}
 Suppose that there exists 
  $y^\star\in\mathcal{S}_Y$ such that $Y_t\geq y^\star$ for all $t\geq 0$. Given that $R^{(m)}_{t_{n_1}}=y_j+\theta(t)$, with $y_j\in\mathcal{S}_Y^{(m)}$, the price at $t_{n_1}\geq 0 $ of a European call (resp. put) option with maturity $t_{n_2}>t_{n_1}$ on a zero-coupon bond maturing at time $T>t_{n_2}$ with a strike $K>0$ can be approximated by
 \begin{multline}\label{eqCallPutBrigo}
       \e\left[e^{-\int_{t_{n_1}}^{t_{n_2}} R_s^{(m)}\diff s }h\left(P^{(m)}(t_{n_2},T)\right) \Big|R^{(m)}_{t_{n_1}}=y_j +\theta(t)\right] \\
         \qquad\qquad\qquad= \mathbf{e}_{j}e^{-\int_{t_{n_1}}^{t_{n_2}}\theta(s) \diff s}\left(e^{\left(\mathbf{Q}^{(m)}_Y-\mathbf{D}_Y\right)(t_{n_2}-t_{n_1})}\right)\mathbf{H},
 \end{multline}

    \begin{sloppypar}
         where $h(x)=\max(x-K,0)$ (resp. $h(x)=\max(k-x,0)$) denotes the payoff function with ${P^{(m)}(t_{n_2},T):=\e\left[e^{-\int_{t_{n_2}}^{T}R_s^{(m)}\diff s}\Big| R_{t_{n_2}}^{(m)}\right]}$ and $\mathbf{H}$ denotes a column vector of size $m\times 1$ whose $k$-th entry, $h_k$, is given by 
       $$h_k= h\left(P_k^{(m)}(t_{n_2},T)\right),$$
   with $P_k^{(m)}(t_{n_2},T)$ defined in \eqref{eqZeroCouponBrigo}.
    \end{sloppypar}
\end{corollary}
\subsection{Calibration to the Initial Term Structure of Interest Rates}\label{appendixBrigoThetaFit}
When the short-rate process is of the form $R_t=Y_t+\theta(t)$, as in Table \ref{tblmodelsNonHomoBrigo}, the fitting to the term structure of interest rates can be greatly simplified, since the function $\theta$ now appears explicitly in the zero-coupon bond formula \eqref{eqZeroCouponBrigo}.

Proposition \ref{propCalibrationBrigo} provides an explicit expression for the function $\theta$ that makes the model zero-coupon bond prices equal to the market prices when the short-rate process is of the form \eqref{eqShortRateExtendedModels}. As in Section \ref{subsecFitTermStruc}, we assume that the time deterministic function $\theta$ is piecewise constant in time, such that
    $$\theta(t)= \sum_{n=1}^N\theta_n\ind_{[t_{n-1},t_n)}(t),$$ 
for some $\bm{\theta}=(\theta_1,\theta_2,\ldots,\theta_N)\in\reals^N$. The objective is thus to find parameters $\bm{\theta}$ such that $P^{(m)}_j(0,t_n)=P^\star(0,t_n)$ for $n=1,2,\ldots,N$, where $P^\star$ represents the market zero-coupon bond prices. Those parameters are called the calibrated parameters and are denoted by a star $\bm{\theta}^\star$.
\begin{prop}\label{propCalibrationBrigo}
 Suppose that there exists $y^\star\in\mathcal{S}_Y$ such that $Y_t\geq y^\star$ for all $t\geq 0$. Given $R_{0}=R^{(m)}_{0}=y_j\in\mathcal{S}_Y^{(m)}$, we have that
\begin{equation}
    \theta^\star_n:= \left\{\begin{array}{ll}
     -\frac{1}{t_n}\ln\left(\frac{P^\star(0,t_n)}{\tilde{P}^{(m)}_j(0,t_n)}\right) & \text{if } n=1 \\
     -\frac{1}{t_n-t_{n-1}}\ln\left(\frac{P^\star(0,t_n)}{P^\star(0,t_{n-1})}\frac{\tilde{P}_j^{(m)}(0,t_{n-1})}{\tilde{P}_j^{(m)}(0,t_{n})}\right) & \text{if } n=2,3,\ldots,N,
    \end{array}\right.\label{eqThetasBrigo}
\end{equation}
where $\tilde{P}_j^{(m)}(0, \cdot)$ is defined in \eqref{eqZeroCouponHomoBrigo}. 
\end{prop}
The proof is straightforward from Corollary \ref{corrZeroCouponBrigo}.
\begin{algorithm}[h]
	\caption{Calibration of $\bm{\theta}$ to the Current Market Term-Structure - Extended Models of \cite{brigoMercurio2006}}
	\label{algoThetasCalibrationBrigo}
	\DontPrintSemicolon
	\KwInput{Initialize $\mathbf{Q}^{(m)}_Y$ and let $t\mapsto P^\star(0,t)$ be the current market zero-bond curve\;
		$N\in\mathbb{N}$, the number of time steps \;
		$\Delta_N\leftarrow T/N$, the size of a time step}
    Set $t_n=n\Delta_N$, $n=1,2,\ldots,N$\;
    Set $\mathbf{D}_Y\leftarrow \diag(\bm{y})$ with $\bm{y}=(y_1,y_2,\ldots,y_m)$, $y_k\in\mathcal{S}^{(m)}_Y$, $k=1,2,\ldots,m$\;
	\tcc{Adjusted transition probability matrix of $Y$ over a period of length $\Delta_N$}
	$\mathbf{A}_{\Delta_N}\leftarrow e^{(\mathbf{Q}^{(m)}_Y- \mathbf{D}_Y)\Delta_N}$\;\label{lineAlgoExpMat}
	\tcc{Calibration to the current market zero-bond curve $t\mapsto P^\star(0,t)$}
     Set $\tilde{\mathbf{P}}(t_1)\leftarrow \mathbf{A}_{\Delta_N}\mathbf{1}_{m\times 1}$\;
     Set $\theta_1^\star=  -\frac{1}{t_1}\ln\left(\frac{P^\star(0,t_1)}{\mathbf{e}_j\tilde{\mathbf{P}}(t_1)}\right)$\;
	\For{$n=2,\ldots, N$}{
        $\tilde{\mathbf{P}}(t_n)\leftarrow \mathbf{A}_{\Delta_N}\tilde{\mathbf{P}}(t_{n-1})$\;
        $\theta_n^\star= -\frac{1}{t_n-t_{n-1}}\ln\left(\frac{P^\star(0,t_n)}{P^\star(0,t_{n-1})}\frac{\mathbf{e}_j\tilde{\mathbf{P}}(t_{n-1})}{\mathbf{e}_j\tilde{\mathbf{P}}(t_{n})}\right)$
	}	
	\KwRet $\{\theta_n^\star\}_{n=1}^N$\;
\end{algorithm}

The calibrated parameters in \eqref{eqThetasBrigo} can easily be obtained by calculating the zero-coupon bond price at each time $\{t_1,t_2,\ldots,t_n\}$. This procedure involves calculating matrix exponentials at each time step, which can slow down the execution considerably. However, by taking advantage of the homogeneous property of $Y$, the matrix exponentials can be calculated only once at the beginning of the procedure, which makes it highly efficient. This is illustrated in Algorithm \ref{algoThetasCalibrationBrigo}. In Algorithm \ref{algoThetasCalibrationBrigo}, $\tilde{\mathbf{P}}(t_n):=[\tilde{P}_j^{(m)}(0,t_n)]_{j=1}^m$ is a column vector of size $m\times 1$, 
with $\tilde{P}_j^{(m)}(0,t_n)$ defined in \eqref{eqZeroCouponHomoBrigo}.
\section{Algorithms} \label{appendixCallablePutableBondAlgo}
 This section presents the results of Section \ref{sectionConvBond} into an algorithm format. More precisely, Proposition \ref{propConvBondCTMC} is provided in Algorithm \ref{algoCBpriceTF_CTMC}. Using the results of \cite{mackay2023analysis} Proposition 4.3, a fast version of Proposition \ref{prop:vanillaCB_CTMC} for the pricing of European-style CBs is also provided in Algorithm \ref{algoVanillaCBfast}, whereas the fast-version of Proposition \ref{propConvBondCTMC} for the pricing American-style CBs is provided in Algorithm \ref{algoCBpriceTF_CTMCfast}.
\subsection{European-Style Convertible Bond}\label{appendixAlgoEuroCBfast}
Using the tower property of conditional expectations and Proposition \ref{propExpectionApprox} inspired from the work of \cite{mackay2023analysis}, we present a new algorithm (Algorithm \ref{algoVanillaCBfast}) that speeds up the pricing of European-style CBs.
\begin{prop}[\cite{mackay2023analysis}, Proposition 4.3]\label{propExpectionApprox}
Let $h>0$ with $h\ll T$ and $0\leq t\leq T-h$. For any function $\phi$ for which the expectation on the left-hand side of \eqref{eqCondExpapprox} is finite, we have that
\begin{multline}
\e\left[\phi\left(t+h,\XmM_{t+h},\Rm_{t+h}\right)\big|\XmM_t=x_i,\Rm_t=r_j\right]\\
=\sum_{k=1}^m\e\left[\phi\left(t+h,\XmM_{t+h},\Rm_{t+h}\right)\big|\Rm_t=\Rm_{t+h}=r_k,\XmM_t=x_i\right]\\
\times\prob{P}\left(\Rm_{t+h}=r_k |\Rm_{t}=r_{j}\right)+\hat{c}(h),\label{eqCondExpapprox}
\end{multline}
where $\hat{c}(h)$ denotes a function satisfying $\lim_{h\rightarrow 0}\frac{\hat{c}(h)}{h}=0$.
\end{prop}
In their paper, \cite{mackay2023analysis} work with stochastic volatility models. However, the reasoning behind the proof is the same for stochastic interest rate models as in the present context.

The following notation is used in Algorithm \ref{algoVanillaCBfast}.
\begin{enumerate}
	\item $\mathbf{B}=[b_{kl}]_{k,l=1}^{m,M}$ denotes a matrix of size $m\times M$, containing the value of the CB. 
	\item $\mathbf{B}_{*,l}=[b_{kl}]_{k=1}^{m}$ denotes the $l$-th column of $\mathbf{B}$, $l=1,2,\ldots, M$,
	\item  $\mathbf{B}_{k,*}=[b_{kl}]_{l=1}^{M}$ denotes the $k$-th row of $\mathbf{B}$, $k=1,2,\ldots, m$.
	\item The symbol $\top$ indicates the matrix (vector) transpose operation.
\end{enumerate}  
\begin{algorithm}[h]
	\caption{European-style CB -- Fast Algorithm}
	\label{algoVanillaCBfast}
	\DontPrintSemicolon
	\KwInput{Initialize $\mathbf{Q}_n^{(m)}$ as in \eqref{eq:generatorR}, and $\mathbf{\Lambda}^{(n,M)}_{k}$ as in \eqref{eq:generatorX}, for $k=1,2,\ldots,m$, $n=1,2,\ldots,N$\;
		$N\in\mathbb{N}$, the number of time steps, \;
		$\Delta_N\leftarrow T/N$, the size of a time step}
		For each $k\in\{1,2,\ldots,m\}$, set 
        $\mathbf{B}_{k,*}\leftarrow\left[\eta e^{x_l+\rho f(r_k)}\ind_{\{e^{x_l+\rho f(r_k)}\geq F/\eta\}} +e^{-\int_0^Tc_u\diff u}F\ind_{\{e^{x_l+\rho f(r_k)}< F/\eta\}}\right]_{l=1}^{M}$\;  
	\For{$n=N-1,\ldots,0$}{
		\For{$k=1,2,\ldots, m$}
		{ $\mathbf{P}^X_{n,k}\leftarrow e^{\mathbf{\Lambda}_{k}^{(n+1,M)} \Delta_N} e^{-r_k\Delta_N}$\;
        $\mathbf{E}_{*,k}\leftarrow \mathbf{P}^X_{n,k}\mathbf{B}_{k,*}^{\top}$}\label{algoEuroCBFastLineEk}
        $\mathbf{P}_n^R\leftarrow e^{\mathbf{Q}^{(m)}_{n+1} \Delta_N}$\;
		\For{$l=1,2,\ldots, M$}
		{ $\mathbf{B}_{*,l}\leftarrow \mathbf{P}_n^R\mathbf{E}_{l,*}^{\top}$}
	}	
	\KwRet $b_{ji}$\;
\end{algorithm}
\begin{remark}[Extension to coupon-bearing bonds]\label{}
    Algorithm \ref{algoVanillaCBfast} is set up for zero-coupon CBs. However, extension to coupon-bearing bonds is straightforward. More precisely, when a coupon $\alpha>0$ is paid at time $t_{n+1}$, then the column vector $\mathbf{E}_{*,k}$ at time $t_n$, line \ref{algoEuroCBFastLineEk} of the algorithm,  must be modified as follow
    $$\mathbf{E}_{*,k}\leftarrow \mathbf{P}^X_{n,k}\left(\mathbf{B}_{k,*}^{\top}+\alpha\mathbf{1}_{M\times1}\right),$$ 
    for each $k \in \{1,2,\ldots,m\}$.
\end{remark}

 Proposition \ref{propExpectionApprox} allows the separation of the matrices $\mathbf{\Lambda}^{(n,M)}_{k}$ and $\mathbf{Q}^{(m)}_n$ at each time step $n\in\{1,2,\ldots,N\}$. Hence, the matrix exponential of a large sparse matrix $\mathbf{G}_n^{(mM)}$ of size $mM\times mM$ is replaced by $m$ calculations of the exponential of an $M\times M$ matrix and one calculation of the exponential of an $m\times m$ matrix. Numerical experiments in Section \ref{sectNumResults} show the accuracy and efficiency of the fast algorithm empirically. 

\subsection{Convertible Bond (American-style)}
 Based on Proposition \ref{propConvBondCTMC}, Algorithm \ref{algoCBpriceTF_CTMC} provides the CTMC approximation for the value of a CB given that $X_0^{(m,M)}=\ln(S_0)-\rho f(R_0)=x_i\in\mathcal{S}^{(M)}_X$ and $R_0^{(m)}=R_0=r_j\in\mathcal{S}^{(m)}_R$. The algorithm is set up for zero-coupon CBs with no additional features, such as call and put options. However, such extensions are straightforward and are discussed further below and in Remark \ref{rmk:AlgoUpdateCBAmCoupon}. 

 \begin{sloppypar}
Similar to the European-style CB, the performance of Algorithm \ref{algoCBpriceTF_CTMC} can be increased by assuming that the short-rate process is constant over small time periods (Proposition \ref{propExpectionApprox}). 
        Let ${\tilde{\mathbf{H}}:=[\eta e^{x_l+\rho f(r_k)}]_{k,l=1}^{m,M}}$ be an $m\times M$ matrix representing the conversion value.
          At each time step, matrices $\tilde{\mathbf{B}}^{E}$, $\tilde{\mathbf{B}}^{CO}$, and $\tilde{\mathbf{B}}$, of size $m\times M$, contain the equity part, cash-only part, and the whole value of the CB, respectively. Furthermore, we denote by $\tilde{b}_{ij}$ (resp $\tilde{h}_{ij}$), the $(i,j)$-entry of $\tilde{\mathbf{B}}$ (resp. $\tilde{\mathbf{H}}$), and define the matrix indicator $\ind_{\{\tilde{\mathbf{B}}=\tilde{\mathbf{H}}\}}$, where each element $(i,j)$ of the matrix, denoted $\ind_{\{\tilde{\mathbf{B}}=\tilde{\mathbf{H}}\}}(i,j)$, is given by
$\mathbf{1}_{\{\tilde{\mathbf{B}}=\tilde{\mathbf{H}}\}}(i,j)= \ind_{\{\tilde{b}_{ij}=\tilde{h}_{ij}\}},$
for $1\leq i\leq m$, $1\leq j\leq M$.
The fast Algorithm to value CBs is provided in Algorithm \ref{algoCBpriceTF_CTMCfast}.
\end{sloppypar} 

 \begin{algorithm}[h]
	\caption{American-style CB}
	\label{algoCBpriceTF_CTMC}
	\DontPrintSemicolon
	\KwInput{Initialize $\mathbf{G}_n^{(mM)}$ as in \eqref{eq:GeneratorY} for $n=1,2,\ldots,N$, $\mathbf{H}^{CO}$ as in \eqref{eq:defH_CO_ctmc}, and $\mathbf{H}^{E}_n$ as in \eqref{eq:defH_SO_ctmc}, for $n=0,1,\ldots,N$\;
	$N\in\mathbb{N}$, the number of time steps, \;
	$\Delta_N\leftarrow T/N$, the size of a time step}
    Set $\mathbf{D}_{mM}\leftarrow \diag\left(\bm{d}\right)$ with $\bm{d}=(d_1,d_2,\ldots,d_{mM})$, and ${d_{(k-1)M+l}=r_k\in\mathcal{S}_R^{(m)}}$, $k=1,2,,\ldots,m$, $l=1,2,,\ldots,M$\;
	Set $\mathbf{B}_N^{CO}\leftarrow\mathbf{H}^{CO}$, $\mathbf{B}_N^{E}\leftarrow\mathbf{H}_N^{E}$, $\mathbf{B}_N \leftarrow \mathbf{B}^{E}_N +\mathbf{B}^{CO}_N$\;
	\For{$n=N-1,N-2,\ldots,0$}
     {$\mathbf{A}_{n+1}\leftarrow\exp\left\{\Delta_N \left(\mathbf{G}^{(mM)}_{n+1}-\mathbf{D}_{mM}\right)\right\}$,\; $\mathbf{B}_n^{CO}\leftarrow e^{-\int_{t_n}^{t_{n+1}}c_u\diff u} \mathbf{A}_{n+1}\mathbf{B}_{n+1}^{CO}$\label{line:algoAmCBregCoupon},
      $\mathbf{B}_n^{E}\leftarrow \mathbf{A}_{n+1}\mathbf{B}_{n+1}^{E}$,\;
      $\mathbf{B}_n\leftarrow\max\left(\mathbf{H}_n^{E}, \mathbf{B}_n^{E}+ \mathbf{B}_n^{CO}\right)$\;
      $\mathbf{B}_n^{CO}\leftarrow \mathbf{B}^{CO}_n\left(\mathbf{1}_{mM\times 1}-\ind_{\{\mathbf{B}_n=\mathbf{H}_n^{E}\}} \right)$,
      $\mathbf{B}_n^{E}=\mathbf{B}_n-\mathbf{B}_n^{CO}$\;
      }
    
	$v^{(m,M)}(0,S_0, R_0)\leftarrow\mathbf{e}_{ji} \mathbf{B}_0$\;
	\KwRet $v^{(m,M)}(0,S_0, R_0)$
\end{algorithm}

\begin{algorithm}[h]
	\caption{American-style CB -- Fast Algorithm}
	\label{algoCBpriceTF_CTMCfast}
	\DontPrintSemicolon
	\KwInput{Initialize $\mathbf{Q}_n^{(m)}$ as in \eqref{eq:generatorR} and $\mathbf{\Lambda}^{(n,M)}_{k}$ as in \eqref{eq:generatorX}, for $k=1,2,\ldots,m$, $n=1,2,\ldots,N$\;
		$N\in\mathbb{N}$, the number of time steps, \;
		$\Delta_N\leftarrow T/N$, the size of a time step}
      Set ${\tilde{\mathbf{H}}:=[\eta e^{x_l+\rho f(r_k)}]_{k,l=1}^{m,M}}$\;
		Set $\tilde{\mathbf{B}}_{k,*}^{E}\leftarrow\left[\eta e^{x_l+\rho f(r_k)}\ind_{\{e^{x_l+\rho f(r_k)}\geq F/\eta\}}\right]_{l=1}^{M}$,  $k=1,2,\ldots,m$\; 
            Set $\tilde{\mathbf{B}}_{k,*}^{CO}\leftarrow\left[F\ind_{\{e^{x_l+\rho f(r_k)}< F/\eta\}}\right]_{l=1}^{M}$,  $k=1,2,\ldots,m$\;
        $\tilde{\mathbf{B}}\leftarrow \tilde{\mathbf{B}}^{E}+\tilde{\mathbf{B}}^{CO}$\;
	\For{$n=N-1,\ldots,0$}{
		\For{$k=1,2,\ldots, m$}
		{$\mathbf{P}^X_{n,k}\leftarrow  e^{\mathbf{\Lambda}_{k}^{(n+1,M)} \Delta_N} e^{-r_k\Delta_N}$\; 
        $\mathbf{E}^{CO}_{*,k}\leftarrow  e^{-\int_{t_n}^{t_{n+1}}c_u\diff u}\mathbf{P}^X_{n,k}(\tilde{\mathbf{B}}^{CO}_{k,*})^{\top}$, \label{line:algoAmCBfastCoupon}
        $\mathbf{E}^{E}_{*,k}\leftarrow \mathbf{P}^X_{n,k}(\tilde{\mathbf{B}}^{E}_{k,*})^{\top}$
        }
        $\mathbf{P}_n^R\leftarrow e^{\mathbf{Q}^{(m)}_{n+1} \Delta_N}$\;
		\For{$l=1,2,\ldots, M$}
		{ $\tilde{\mathbf{B}}_{*,l}^{CO}\leftarrow \mathbf{P}_n^R(\mathbf{E}_{l,*}^{CO})^{\top}$,
         $\tilde{\mathbf{B}}_{*,l}^{E}\leftarrow \mathbf{P}_n^R(\mathbf{E}_{l,*}^{E})^{\top}$,
         $\tilde{\mathbf{B}}\leftarrow \max\left(\tilde{\mathbf{H}},\tilde{\mathbf{B}}^{E}+\tilde{\mathbf{B}}^{CO}\right)$\;
         $\tilde{\mathbf{B}}^{CO}\leftarrow \tilde{\mathbf{B}}^{CO}\left(\mathbf{1}_{m\times M}-\ind_{\{\tilde{\mathbf{B}}=\tilde{\mathbf{H}}\}} \right)$,
      $\tilde{\mathbf{B}}^{E}=\tilde{\mathbf{B}}-\tilde{\mathbf{B}}^{CO}$\;}
	}	
	\KwRet $b_{ji}$\;
\end{algorithm}

\newpage

\begin{remark}\label{rmk:AlgoUpdateCBAmCoupon}
The extension of Algorithms \ref{algoCBpriceTF_CTMC} and \ref{algoCBpriceTF_CTMCfast} to coupon-bearing bonds is straightforward. When a coupon $\alpha>0$ is paid at time $t_{n+1}$, the continuation value at time $t_n$ of the cash-only part must be adjusted accordingly. Specifically, line \ref{line:algoAmCBregCoupon} of Algorithm \ref{algoCBpriceTF_CTMC} should be updated to 
$$\mathbf{B}_n^{CO}\leftarrow e^{-\int_{t_n}^{t_{n+1}}c_u\diff u} \mathbf{A}_{n+1}\left(\mathbf{B}_{n+1}^{CO}+\alpha\mathbf{1}_{mM\times 1}\right),$$  
and line \ref{line:algoAmCBfastCoupon} of Algorithm \ref{algoCBpriceTF_CTMCfast} should be changed to
$$\mathbf{E}^{CO}_{*,k}\leftarrow  e^{-\int_{t_n}^{t_{n+1}}c_u\diff u}\mathbf{P}^X_{n,k}\left((\tilde{\mathbf{B}}^{CO}_{k,*})^{\top}+\alpha\mathbf{1}_{M\times1}\right),$$ 
for each $k\in\{1,2,\ldots,M\}$.
\end{remark}

\newpage
\newpage
\setcounter{page}{1} 
\section{Supplemental Material}
This document provides supplemental material to A Unifying Approach for the Pricing of Debt Securities.
\subsection{Closed-Form Expression for European-style Convertible Bonds under TF approach}\label{appendixCBeuroTF}
In the following, we derive a closed-form analytical formula for European-style CBs (or when the conversion option can only be exercised at maturity). We suppose that credit risk is incorporated into the model using the approach of \cite{tsiveriotis1998valuing}. 

Accordingly, we make some simplifying assumptions. We suppose that the dynamics of the stock process are given by
\begin{equation}
        \begin{aligned}
         \diff S_t &=(R_t-q_t)S_t\diff t+\tilde{\sigma}_S S_t \left(\rho \diff \widetilde{W}^{(2)}_t + \sqrt{1-\rho^2} \diff \widetilde{W}^{(1)}_t\right),\\
	  \diff R_t &=(\theta(t)-\kappa R_t)\diff t+\tilde{\sigma}_R \diff \widetilde{W}_t^{(2)},\label{eqEDS_Sv2}
         \end{aligned}
	\end{equation}
    with $\tilde{\sigma}_S,\tilde{\sigma}_R>0$, $\rho\in[-1,1]$ and $\widetilde{W}=\{(\widetilde{W}^{(1)}_t, \widetilde{W}^{(2)}_t)\}_{t\geq0}$ is a standard bi-dimensional Brownian motion\footnote{The formulation in \eqref{eqEDS_Sv2} in terms of independent Brownian motion is equivalent to that in \eqref{eqEDS_S} in terms of correlated Brownian motion. Indeed, define $Z_t= \rho \widetilde{W}^{(2)}_t + \sqrt{1-\rho^2} \widetilde{W}^{(1)}_t$, $t\geq 0$. From Cholesky decomposition, the process $(Z,\widetilde{W}^{(2)})$ is a correlated Brownian motion with cross-variation $[Z,\tilde{W}^{(2)}]_t=\rho t$.}. Under this assumption, the short-rate process corresponds to the Hull--White model in Table \ref{tblmodelsNonHomo}. When function $\theta(\cdot)$ is constant over time, then the short-rate dynamics collapsed to the Vasicek model\footnote{It suffices to set $\theta(t)=\tilde{\theta}\kappa$, for $0\leq t\leq T$, for some constant $\tilde{\theta}>0$ to obtain the Vasicek model with a long-term mean parameter equal to $\tilde{\theta}$.}, 
    whereas when $\kappa=0$, the Ho--Lee model is obtained. Under these three models, the short-rate process is Gaussian, and the price of the zero-coupon bond can be obtained explicitly by
    \begin{equation}
    P(t,T)=e^{A(t,T)-B(t,T)R_t},\label{eqZBpriceAnalytic}
    \end{equation}
    for some time-deterministic functions $A$ and $B$ given in Table \ref{tblAandBparamZB}, see \cite{bjork2009}, Section 24.4 for details. Finally, note that when $\sigma_R(\cdot)=\kappa=\theta(\cdot)=0$, then the short-rate process is constant to $R_0$ and \eqref{eqEDS_Sv2} collapsed to the Black--Scholes model.
    \begin{table}[h]
	\begin{tabular}{ccc}
		\hline
		Model  & $A(t,T)$ & $B(t,T)$\\
		\hline 
		Vasicek & $\left(\theta-\frac{\tilde{\sigma}_R^2}{2\kappa^2}\right)\left[B(t,T)-(T-t)\right]-\frac{\tilde{\sigma}_R^2}{4\kappa}B^2(t,T)$ & $\frac{1}{\kappa}[1-e^{-\kappa(T-t)}]$ \\
        Ho--Lee & $\int_t^T\theta(s)(s-T)\diff s+\frac{\tilde{\sigma}_R^2}{2}\frac{(T-t)^3}{3}$ & $T-t$\\
        Hull--White & $\int_t^T\frac{1}{2}\tilde{\sigma}_R^2 B(s,T)-\theta(s)B(s,T)\diff s$ &  $\frac{1}{\kappa}[1-e^{-\kappa(T-t)}]$  \\[3pt]
		\hline
	\end{tabular}
 
	\caption{\small{Definition of $A(t,T)$ and $B(t,T)$ for different models of Tables \ref{tblmodelsHomo} and \ref{tblmodelsNonHomo}.}}
	\label{tblAandBparamZB}
\end{table} 

The following proposition provides a closed-form pricing formula for European-style CBs under general stochastic short-rate models \eqref{eqEDS_Sv2}. The pricing of European call options under stochastic interest rate are discussed in \cite{geman1995changes} (Theorem 2 and Section 3.2), \cite{bjork2009} (Section 26.5), and \cite{brigoMercurio2006} (Appendix B), among others. 
The general proof uses the change of numéraire techniques developed by \cite{geman1995changes}.  
An alternative derivation is also provided in \cite{abudy2013pricing} when the short-rate dynamics is given by the Ornstein--Uhlenbeck process (or, equivalently, the Vasicek model). 
    \begin{prop}\label{prop:ExactFormulaVanillaCB_TF}
        Given $S_t=x>0$ and $R_t=r\in\reals$, the price at time $t$ of a European-style CB 
        with maturity $T>0$, face value $F>0$, and conversion ratio $\eta>0$ is given by
    \begin{align}\label{eqCBeuroTFgenFormulaFinal}
        v_e(t,r,x) 
       = \eta x e^{-\int_t^Tq_s\diff s}\Phi(d_1)+e^{-\int_t^{T} c_u\diff u}P(t,T)F\Phi(d_2)
    \end{align}
    \begin{sloppypar}
          where $\Phi(\cdot)$ denotes the cumulative distribution of a standard normal distribution, ${d_1 =\frac{\ln\left(\frac{\eta x }{F P(t,T)}\right)-\int_t^Tq_s\diff s+\frac{1}{2}V(t,T)}{\sqrt{V(t,T)}}}$, and ${d_2 =\sqrt{V(t,T)}-d_1}$,  with
          \begin{equation*}
              V(t,T)=\begin{cases}
                  \tilde{\sigma}_S^2(T-t)+\frac{2\rho\tilde{\sigma}_S\tilde{\sigma}_R}{\kappa}\left[(T-t)-B(t,T)\right]+\frac{\tilde{\sigma}_R^2}{\kappa^2}\left[(T-t)-\frac{\kappa}{2}B^2(t,T)-B(t,T)\right] & \textrm{if } \kappa\neq 0\\
                   \tilde{\sigma}_S^2(T-t)+\rho\tilde{\sigma}_S\tilde{\sigma}_R(T-t)^2+\frac{\tilde{\sigma}_R^2}{3}(T-t)^3 & \textrm{if } \kappa=0,
              \end{cases}
          \end{equation*}
          and ${P(t,T) =e^{A(t,T)+B(t,T)r}}$, where functions $A$ and $B$ are defined in Table \ref{tblAandBparamZB}. 
    \end{sloppypar}
    \end{prop}
    \begin{proof}
         First, assume first that $q_t=0$ for all $t\geq 0$ and recall from \eqref{eq:CBeuroTF} that
         \begin{equation}
            v_e(t,x,r) =\e \Big[e^{-\int_t^T R_u\diff u} \eta S_T \ind_{\{S_T\geq F/\eta\}} 
            + e^{-\int_t^T R_u+c_u\diff u} F  \ind_{\{S_T< F/\eta\}} \big| S_t=x, R_t=r\Big]. \label{eqCBeuroTFv2}
     \end{equation}
    To solve this problem and avoid working with the joint density of $(\int_t^T R_u\diff u, S_T)$, we use the change of numéraire technique discussed in \cite{geman1995changes}, Theorem 2 and Section 3.2 (a). Consequently, we introduce the $T$-forward measure, denoted by $\mathbb{Q}^T$, which has the zero-coupon bond price process $P=\{P(t,T)\}_{0\leq t\leq T}$ as numéraire. We also introduce the measure $\mathbb{Q}_S$ under which the risky asset $S$ is the chosen numéraire. Following \cite{geman1995changes}, Theorem 2, we have that
    \begin{align}
      v_e(t,x,r) &=P(t,T)\e^{\mathbb{Q}^T}\left[ \eta S_T \ind_{\{S_T\geq F/\eta\}}  \big| S_t=x, R_t=r\right]\nonumber \\
     &\quad\quad + e^{-\int_t^T c_u\diff u}P(t,T) F \mathbb{Q}^T\left( S_T< F/\eta\big| S_t=x, R_t=r\right)\nonumber\\
     \begin{split}
     &=\eta x\mathbb{Q}_S\left(S_T\geq F/\eta \big| S_t=x, R_t=r\right) \\
     & \quad\quad +e^{-\int_t^T c_u\diff u}P(t,T) F \mathbb{Q}^T\left( S_T< F/\eta\big| S_t=x, R_t=r\right),
     \end{split}\label{CBeuroTF_genFormula}
    \end{align}
    where the first term of the second equality follows from Corollary 2 of \cite{geman1995changes}, and $\e^{\mathbb{Q}^T}[\cdot]$ denotes the expectation under the $T$-forward measure. 
     To solve \eqref{CBeuroTF_genFormula}, we need to find the distribution of $S_T$ under measures $\mathbb{Q}^T$ and $\mathbb{Q}^S$, respectively. This is what we do in the following. 
     
     The first step to obtain the dynamics of $S$ under $\mathbb{Q}^T$ consists of finding the dynamics of the zero-coupon bond price process $P$ under $\mathbb{Q}$. Accordingly, recall from \eqref{eqZBpriceAnalytic} that
     $$P(t,T)=e^{A(t,T) - B(t,T)R_t},$$
     for some time-deterministic functions $A$ and $B$ defined in Table \ref{tblAandBparamZB}, and note that 
     $$A_t(t,T):=\frac{\partial A}{\partial t}(t,T) = \theta (t) B(t,T) -\frac{1}{2} \tilde{\sigma}_R^2 B^2(t,T),\textrm{ and } B_t(t,T) :=\frac{\partial B}{\partial t}(t,T) = \kappa B(t,T)-1,$$
     see \cite{bjork2009}, Section 24.4. Hence, applying Ito's formula to $\ln P(t,T)$, we obtain that
     \begin{align}
         \diff \ln P(t,T) & = \left(A_t(t,T) - B_t(t,T)R_t \right) \diff t - B(t,T) \diff R_t\nonumber\\
         & = \left[ \theta(t) B(t,T)-\frac{1}{2}\tilde{\sigma}_R^2 B^2(t,T) -  B_t(t,T)R_t - B(t,T) \left(\theta(t)-\kappa R_t\right)\right] \diff t - \tilde{\sigma}_R B(t,T) \diff  \widetilde{W}^{(2)}_t\nonumber\\
         & = \left[R_t - \frac{1}{2}\tilde{\sigma}_R^2 B^2(t,T)\right] \diff t  - \tilde{\sigma}_R B(t,T) \diff  \widetilde{W}^{(2)}_t,\label{eqLNZBpriceEDS}
     \end{align}
     so that,
     \begin{align}
         \diff  P(t,T) 
         & = R_tP(t,T) \diff t  - \tilde{\sigma}_R B(t,T)P(t,T) \diff  \widetilde{W}^{(2)}_t.\label{eqZBpriceEDS2}
     \end{align}
   Now, observe that
    \begin{equation}\label{eqProbST_QT}
         \mathbb{Q}^T\left( S_T< F/\eta\big| S_t=x, R_t=r\right)=\mathbb{Q}^T\left( \frac{S_T}{P(T,T)}< F/\eta\big| S_t=x, R_t=r\right).
    \end{equation}
    From Theorem 1 (i) of \cite{geman1995changes}, we know that process $f=\{f_t:=\frac{S_t}{P(t,T)}\}_{0\leq t\leq T}$ is a $\mathbb{Q}^T$-local martingale (thus, it as no drift). Moreover, since a change of measures only affects the drift of a process, we can deduce the dynamics of process $f$ under $\mathbb{Q}^T$ from its dynamics under $\mathbb{Q}$ by setting the drift term to nil, which we do next. Using \eqref{eqEDS_Sv2} , \eqref{eqZBpriceEDS2} and Itô's formula, we find that 
    \begin{align}
     \diff f_t& =\frac{1}{P(t,T)}\diff S_t -\frac{S_t}{P^2(t,T)}\diff P(t,T) - \frac{1}{P^2(t,T)}\diff\langle S_t,P(t,T)\rangle +\frac{S_t}{P^3(t,T)}\diff \langle P(t,T) \rangle\nonumber\\
     & =\left(\rho\tilde{\sigma}_S\tilde{\sigma}_R B(t,T) +\tilde{\sigma}_R^2 B^2(t,T)\right)f_t\diff t + \tilde{\sigma}_S\sqrt{1-\rho^2}f_t \diff \widetilde{W}^{(1)}_t +\left[\rho\tilde{\sigma}_s+\tilde{\sigma}_R B(t,T)\right]f_t\diff \widetilde{W}^{(2)}_t.\label{eqEDS_ft_QT}
     \end{align}
   Because $f$ must be a local martingale under $\mathbb{Q}^T$, we deduce 
   from Girsanov theorem that process $(\widehat{W}^{(1)}, \widehat{W}^{(2)})$, defined by
      \begin{equation}
          \begin{split}
              \diff \widehat{W}^{(1)}_t & = \diff \widetilde{W}^{(1)}_t,\\
              \diff \widehat{W}^{(2)}_t & = \diff \widetilde{W}^{(2)}_t +\tilde{\sigma}_R B(t,T) \diff t,
          \end{split}
      \end{equation}
      for $0\leq t\leq T$, is a standard bi-dimensional Brownian motion under $\mathbb{Q}^T$, which implies that
       \begin{align}
     \diff f_t
     & = \tilde{\sigma}_S\sqrt{1-\rho^2}f_t \diff \widehat{W}^{(1)}_t +\left[\rho\tilde{\sigma}_s+\tilde{\sigma}_R B(t,T)\right]f_t\diff \widehat{W}^{(2)}_t.\label{eqEDS_ft_QT2}
     \end{align}
     Hence, $\ln f_T$ given $\ln f_t$ is normally distributed with a mean $\hat{\mu}:=\ln f_t -\frac{1}{2}V(t,T)$ (see for instance, \cite{bjork2009}, Lemma 4.15), and variance $\hat{\sigma}^2 :=V(t,T)$, where
     \begin{align}
     V(t,T)&= \int_t^T \tilde{\sigma}_S^2(1-\rho^2)+ \left[\rho\tilde{\sigma}_S+\tilde{\sigma}_RB(s,T)\right]^2\diff s\nonumber\\
        &= \tilde{\sigma}_S^2(T-t)+2\rho\tilde{\sigma}_S\tilde{\sigma}_R \int_t^T B(s,T)\diff s +\tilde{\sigma}_R^2 \int_t^T B^2(s,T)\diff s.\label{eqVarianceGenFormulaEuroCBTF}
     \end{align}
     The probability \eqref{eqProbST_QT} can thus be calculated explicitly using the property of the normal distribution, and we obtain that,
       \begin{equation}\label{eqPhid2}
         \mathbb{Q}^T\left(f_T< F/\eta\big|f_t=x/P(t,T)\right)=\Phi\left(\frac{\ln\left(\frac{FP(t,T)}{\eta x}\right)+\frac{1}{2}V(t,T)}{\sqrt{V(t,T)}}\right)= \Phi(d_2).
    \end{equation}
     This completes the proof for the second term of \eqref{CBeuroTF_genFormula}.
    
    The challenge now consists of finding an expression for the first term of \eqref{CBeuroTF_genFormula}. We thus need the distribution of $S_T$ under $\mathbb{Q}_S$. To this end, we observe that 
         \begin{equation}
         \begin{split}
         \mathbb{Q}_S\left( S_T\geq F/\eta \big| S_t=x, R_t=r\right)& = \mathbb{Q}_S\left( 1/S_T\leq \eta/F \big| S_t=x, R_t=r\right)\\
         &=\mathbb{Q}_S\left( P(T,T)/S_T\leq \eta/F \big| S_t=x, R_t=r\right) \\
         &=\mathbb{Q}_S\left( 1/f_T\leq \eta/F \big| f_t=x/P(t,T)\right)\label{eqProbY_Qs2}
         \end{split}
     \end{equation}
       \begin{sloppypar} 
       Similarly, as above, we use the results of Theorem 1 of \cite{geman1995changes} to conclude that process ${Y=\{Y_t:=P(t,T)/S_t=1/f_t\}_{0\leq t\leq T}}$ is a local martingale under $\mathbb{Q}_S$. Thus, it has no drift, and because a change of measures only affects the drift of a process, we can deduce the dynamics of process $Y$ under $\mathbb{Q}_S$ from its dynamics under $\mathbb{Q}^T$, by setting the drift term to nil. This what we do in the following. Using \eqref{eqEDS_ft_QT2} and Itô's formula, we find that
        \end{sloppypar}
    \begin{align*}
        \diff Y_t &= -\frac{1}{f_t^2}\diff f_t + \frac{1}{f_t^3}\diff \langle f_t \rangle \\
        &= \left[\tilde{\sigma}_S^2(1-\rho^2) +(\rho\tilde{\sigma}_S+\tilde{\sigma}_R B(t,T))^2\right] Y_t \diff t \\
        &\quad\quad\quad-\tilde{\sigma}_S\sqrt{1-\rho^2}Y_t \diff \widehat{W}^{(1)}_t -\left[\rho\tilde{\sigma}_s+\tilde{\sigma}_R B(t,T)\right]Y_t\diff \widehat{W}^{(2)}_t.
    \end{align*}
    From the local martingale property of $Y$ under $\mathbb{Q}_S$, we deduce from the Girsanov  theorem 
    that process $(\bar{W}^{(1)}, \bar{W}^{(2)})$, defined by
     \begin{equation}
          \begin{split}
              \diff \bar{W}^{(1)}_t & = \diff \widehat{W}^{(1)}_t-\tilde{\sigma}_S\sqrt{1-\rho^2}\diff t,\\
              \diff \bar{W}^{(2)}_t & = \diff \widehat{W}^{(2)}_t -\rho\tilde{\sigma}_S+\tilde{\sigma}_R B(t,T)\diff t,
          \end{split}
      \end{equation}
      for $0\leq t\leq T$, is a standard bi-dimensional Brownian motion under $\mathbb{Q}_S$. Thus, the dynamics of $Y$ under the new measure are given by
    \begin{align*}
        \diff Y_t = -\tilde{\sigma}_S\sqrt{1-\rho^2}Y_t \diff \bar{W}^{(1)}_t -\left[\rho\tilde{\sigma}_s+\tilde{\sigma}_R B(t,T)\right]Y_t\diff \bar{W}^{(2)}_t.
    \end{align*}
    \begin{sloppypar}
   From there, we conclude that
    $\ln Y_T$ given $\ln Y_t$ (or $\ln f_t$, since $\ln Y_t=\ln 1/f_t$) is normally distributed with mean ${\bar{\mu}:=-\ln f_t-\frac{1}{2}V(t,T)}$ and variance ${\bar{\sigma}^2:=\hat{\sigma}^2=V(t,T)}$, and the probability \eqref{eqProbY_Qs2} can be calculated explicitly as
    \begin{equation}\label{eqPHid1}
          \mathbb{Q}_S\left( 1/f_T\leq \eta/F \big| f_t=x/P(t,T)\right)=\Phi\left(\frac{\ln\left(\frac{\eta x}{ F P(t,T)}\right) +\frac{1}{2}V(t,T)}{\sqrt{V(t,T)}}\right)=\Phi(d_1).
    \end{equation}
    The final assertion then follows from \eqref{CBeuroTF_genFormula},\eqref{eqVarianceGenFormulaEuroCBTF}, \eqref{eqPhid2}, \eqref{eqPHid1}, and the expression for function $B$ in Table \ref{tblAandBparamZB}. 

    When $q_t>0$ for some $t\geq 0$, the results can be derived using the relationship $S_t=\tilde{S}_t e^{-\int_0^tq_s\diff s}$, $t\geq 0$, where $\tilde{S}=\{\tS_t\}_{t\geq 0}$ represents the stock price process when the dividend is assumed to be nil\footnote{The value of $\tilde{S}$ can also be interpreted as the value of the stock price when dividends are continuously reinvested in the stock.}.
    \end{sloppypar}
    \end{proof}
\begin{remark}
    When $\sigma_R=\kappa=\theta(t)=0$ for all $t\in[0,T]$, the short-rate is constant to $R_0$ over time and the model in \eqref{eqEDS_Sv2} collapsed to the Black--Scholes model. In that case, $A(t,T)=0$, $B(t,T)=T-t$, such that $P(t,T)=e^{-(T-t)R_0}$, and $V(t,T)=\sigma_S^2(T-t)$. The general formula \eqref{eqCBeuroTFgenFormulaFinal} then becomes
     \begin{align}\label{eqCBeuroTF_BS}
        v_e(t,r,x) 
       = \eta x e^{-\int_t^Tq_s\diff s}\Phi(d_1)+e^{-(T-t)R_0}e^{-\int_t^{T} c_u\diff u}F\Phi(d_2),
    \end{align}
    with $d_1=\frac{\ln\left(\frac{\eta x }{F}\right)+\left(r-\int_t^Tq_s\diff s+\frac{1}{2}\sigma_S^2\right)(T-t)}{\sigma_S\sqrt{T-t}}$, and $d_2=\sigma_S\sqrt{T-t}-d_1$.
\end{remark}

\subsection{Additional Numerical Experiments}\label{appendixSupplMatNumExperiment}
We now investigate the accuracy, convergence, and efficiency of the CTMC method to approximate various debt securities in the Vasicek, CIR, and Dothan models. Note that Assumption \ref{assumpStateSpaceR} is not respected under the Vasicek model. 

Unless stated otherwise, we use the following model and CTMC parameters in all numerical experiments, summarized in Table \ref{tblModelCTMCParamHomo}.
\begin{table}[h!]
      \centering
      \begin{tabular}{ccccc|ccccc}
		\hline
		 & ${r_0}$ & ${\kappa}$ & ${\theta}$ & ${\sigma}$& ${m}$  & ${r_1}$ & ${r_m}$ & $\tilde{\alpha}$& $\Delta_N$ \\
		\hline 
		Vasicek  & $0.04$ & $1$ & $0.04$ & $0.20$ &$160$ & $-30 R_0$& $25 R_0$ & $0.5$ & 1/252  \\
        CIR  & $0.04$ & $2$ & $0.035$ & $0.20$ &$160$ & $R_0/100$& $7 R_0$ & $0.5$  & 1/252 \\
        Dothan & 0.02 & 0 & N/A & 0.15&$160$ & $R_0/100$& $7 R_0$ & $0.5$ & 1/252 \\
		\hline
	\end{tabular}
 \caption{Model and CTMC parameters}\label{tblModelCTMCParamHomo}
\end{table}

To construct the state-space of the CTMC, $\mathcal{S}^{(m)}_R=\{r_1,r_2,\ldots,r_m\}$ with $m\in\mathbb{N}$, we use the non-uniform grid proposed by Tavella and Randall (\cite{tavellapricing}, Chapter 5.3), as in Section \ref{sectNumResults}. 

\subsubsection{Approximation of Zero-Coupon Bond Prices}\label{appendixNumZBpriceHomo}
We now examine the accuracy of the CTMC methods in calculating zero-coupon bond prices, Corollary \ref{corrZeroCouponBondHomo}. Under the Vasicek and CIR models, the price of zero-coupon bonds has a closed-from expression, which can be found, for instance, in \cite{brigoMercurio2006}, Section 3.2. The analytical value of the zero-coupon bond price can thus be used as a benchmark in our experiment. We test the accuracy of the approximated prices across different values of model parameters. The results are summarized in Table \ref{tblAccuracyZBpriceHomo}. Column ``CTMC'' reports the CTMC approximated value using the results of Corollary \ref{corrZeroCouponBondHomo}. 
\begin{table}[b]
	\begin{subtable}[c]{0.495\linewidth}
		\centering
		\scalebox{0.90}{
			\begin{tabular}{cccc}
		\hline
		$\mathbf{\kappa}$ & \textbf{CTMC} & \textbf{Benchmark} & \textbf{Abs. error} \\
		\hline 
		\textbf{0.5}&      0.9625591 &	 0.9625609 &	1.77E-06 \\
        \textbf{1}  &   0.8964870 &	 0.8964877 &	7.12E-07 \\
        \textbf{2} &   0.8661056 &	 0.8661057 &	9.47E-08 \\
        \textbf{3} &  0.8587974 &	 0.8587974 &	2.27E-08 \\
        \textbf{4} &  0.8560138 &	 0.8560138 &	7.77E-09 \\
		\hline
	\end{tabular}}
	\end{subtable}
	\begin{subtable}[c]{0.495\linewidth}
		\centering
		\scalebox{0.90}{
			\begin{tabular}{cccc}
		\hline
		$\mathbf{\kappa}$ & \textbf{CTMC} & \textbf{Benchmark} & \textbf{Abs. error} \\
		\hline 
		\textbf{0.5}&       0.8656670  & 	 0.8656663 & 	7.09E-07 \\
        \textbf{1}  &    0.8676884 	&  0.8676884 	& 1.11E-08 \\
        \textbf{2} &    0.8676884 	 & 0.8676884 	& 1.11E-08 \\
        \textbf{3} &   0.8681491 	 & 0.8681491 	& 3.96E-10 \\
        \textbf{4} &  0.8684110	&  0.8684110	& 1.98E-11 \\
		\hline
	\end{tabular}}
	\end{subtable}
 \begin{subtable}[c]{0.49\linewidth}
		\centering
		\scalebox{0.90}{
			\begin{tabular}{cccc}
		\hline
		$\mathbf{R_0}$ & \textbf{CTMC} & \textbf{Benchmark} & \textbf{Abs. error} \\
		\hline 
		\textbf{0.02}&	 0.9142555 	& 0.9142630 &	7.43E-06\\	 
        \textbf{0.03} &	 0.9053312 	& 0.9053317 &	5.12E-07\\	 
        \textbf{0.04}&	 0.8964870 	& 0.8964877 &	7.12E-07\\	 
        \textbf{0.05}&	 0.8877292 	& 0.8877301 &	8.23E-07\\	 
		\hline
	\end{tabular}}
	\end{subtable}
 \begin{subtable}[c]{0.49\linewidth}
		\centering
		\scalebox{0.90}{
			\begin{tabular}{cccc}
		\hline
		$\mathbf{R_0}$ & \textbf{CTMC} & \textbf{Benchmark} & \textbf{Abs. error} \\
		\hline 
        \textbf{0.02} &	 0.8763624& 	 0.8763627& 	2.82E-07\\
        \textbf{0.03}&	 0.8720147 &	 0.8720147 &	2.21E-09\\
        \textbf{0.04}&	 0.8676884 	& 0.8676884 	&1.11E-08\\
        \textbf{0.05}&	 0.8633835 	 &0.8633835 	&2.90E-08\\
		\hline
	\end{tabular}}
	\end{subtable}
  \begin{subtable}[c]{0.49\linewidth}
		\centering
		\scalebox{0.90}{
			\begin{tabular}{cccc}
		\hline
		$\mathbf{\theta}$ & \textbf{CTMC} & \textbf{Benchmark} & \textbf{Abs. error} \\
		\hline 
        \textbf{0.01} &	 0.9814527 &	 0.9814529 &	1.93E-07\\
        \textbf{0.02} &	 0.9522718 &	 0.9522722 &	3.77E-07\\
        \textbf{0.03} &	 0.9239585 &	 0.9239590 &	5.50E-07\\
        \textbf{0.04} &	 0.8964870 &	 0.8964877 &	7.12E-07\\
		\hline
	\end{tabular}}
	\end{subtable}
  \begin{subtable}[c]{0.49\linewidth}
		\centering
		\scalebox{0.90}{
			\begin{tabular}{cccc}
		\hline
		$\mathbf{\theta}$ & \textbf{CTMC} & \textbf{Benchmark} & \textbf{Abs. error} \\
		\hline 
        \textbf{0.015} &	     0.9303530 &		 0.9303518 &		1.20E-06\\
        \textbf{0.025} &		 0.8984741 &		 0.8984740 &		1.82E-07\\
        \textbf{0.035}&		 0.8676884 &		 0.8676884 &		1.11E-08\\
        \textbf{0.045} &		 0.8379576 &		 0.8379576 &		5.63E-10\\
		\hline
	\end{tabular}}
	\end{subtable}
  \begin{subtable}[c]{0.49\linewidth}
		\centering
		\scalebox{0.90}{
			\begin{tabular}{cccc}
		\hline
		$\mathbf{\sigma}$ & \textbf{CTMC} & \textbf{Benchmark} & \textbf{Abs. error} \\
		\hline 
        \textbf{0.1}	& 0.8630196 &	 0.8630198 & 	1.70E-07 \\
        \textbf{0.2}	& 0.8964870 &	 0.8964877 &	7.12E-07 \\
        \textbf{0.3}	& 0.9551747 &	 0.9551765 &	1.77E-06 \\
        \textbf{0.4}& 1.0438353 &	 1.0438513 &	1.60E-05 \\
		\hline
	\end{tabular}}
		\subcaption{Vasicek model}
	\end{subtable}
  \begin{subtable}[c]{0.49\linewidth}
		\centering
		\scalebox{0.90}{
			\begin{tabular}{cccc}
		\hline
		$\mathbf{\sigma}$ & \textbf{CTMC} & \textbf{Benchmark} & \textbf{Abs. error} \\
		\hline 
        \textbf{0.1} &	 0.8673140 &	 0.8673140 &	1.93E-10\\
        \textbf{0.2}	& 0.8676884 	& 0.8676884 	&1.11E-08\\
        \textbf{0.3}	& 0.8683033 	& 0.8683025 	&7.69E-07\\
        \textbf{0.4} & 0.8691384 	& 0.8691428 	&4.36E-06\\

		\hline
	\end{tabular}}
		\subcaption{CIR model}
	\end{subtable}
	\caption[Accuracy of the zero-coupon bond price approximation under Vasicek and CIR models]{\small{Accuracy of the zero-coupon bond price approximation \eqref{eqZeroCouponHomo} under Vasicek and CIR models. Model and CTMC parameters are as listed in Table \ref{tblModelCTMCParamHomo}, and zero-coupon bond parameters are $t=0$ and $T=4$.}}
	\label{tblAccuracyZBpriceHomo}
\end{table} 
\begin{figure}[b]
	\centering
		\includegraphics[scale=0.35]{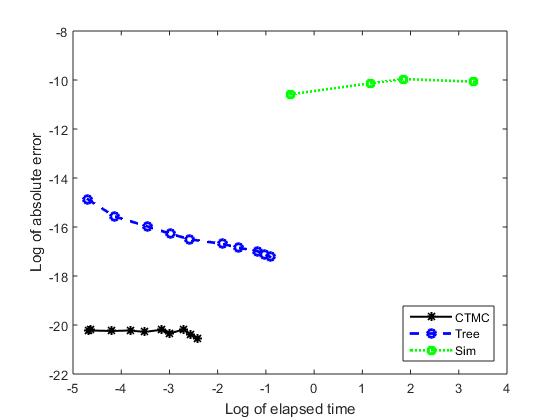} 
	\caption[Efficiency of the CTMC method in approximating zero-coupon bond prices under the Dothan model]{\small{Efficiency of the CTMC method in approximating zero-coupon bond prices under the Dothan model. Except for the number of grid points $m$, which range from $100$ to $1000$, model and CTMC parameters are as listed in Table \ref{tblModelCTMCParamHomo}. Zero-coupon bond price parameters are $t=0$ and $T=4$.}}\label{figEfficiencyZBPriceDothan}
\end{figure}
\begin{figure}[t]
	\centering
	\begin{tabular}{cc}
		\includegraphics[scale=0.3]{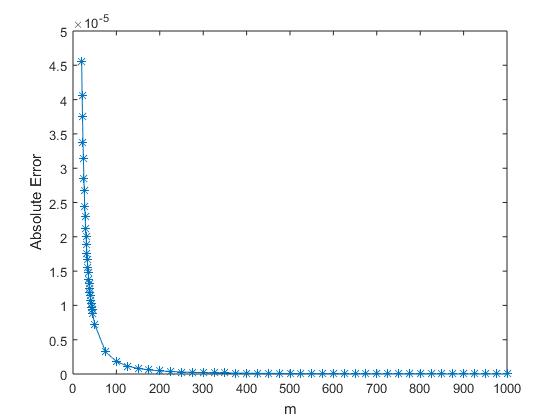} & \includegraphics[scale=0.3]{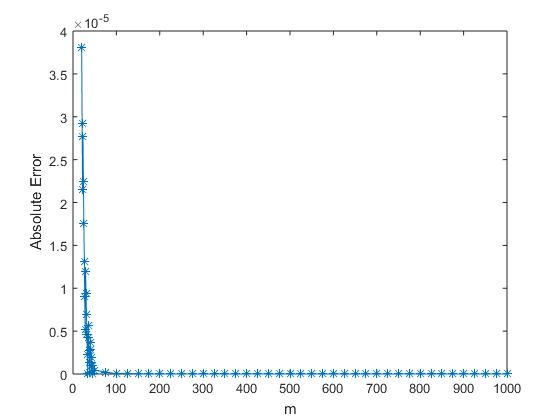}\\
         \small{\textsc{(a)} Vasicek} & \small{\textsc{(b)} CIR}  
	\end{tabular}
	\caption[Convergence pattern of the approximated zero-coupon bond prices]{\small{Convergence pattern of the approximated zero-coupon bond prices. Except for the number of grid points $m$, model and CTMC parameters are as listed in Table \ref{tblModelCTMCParamHomo}. Zero-coupon bond price parameters are $t=0$ and $T=4$.}}\label{figConvergenceZB_CIRVa}
\end{figure}

\begin{table}[h]
	\begin{subtable}[c]{0.495\linewidth}
		\centering
			\begin{tabular}{ccccc}
		\hline
		\textbf{m}  & \textbf{Abs. error} & \textbf{Rate}  \\
		\hline 
		50	&7.24E-06	& - \\
        100	&1.82E-06 &	 1.99\\ 
        200	&4.56E-07 &	 2.00\\ 
        300	&2.03E-07 &	 2.00 \\
        900	& 2.19E-08&	 2.03\\
	\hline
	\end{tabular}
    \subcaption{Vasicek}
	\end{subtable}
	\begin{subtable}[c]{0.495\linewidth}
		\centering
			\begin{tabular}{ccccc}
		\hline
		\textbf{m}  & \textbf{Abs. error} & \textbf{Rate}  \\
        \hline 
        50	&4.57E-07& -\\	
        100	&3.90E-08&	 3.55\\ 
        200	&4.63E-09&	 3.08\\ 
        300	&1.55E-09&	 2.70\\ 
        900	&9.18E-10&	 0.48\\ 
		\hline
	\end{tabular}
    \subcaption{CIR}
	\end{subtable}
 
	\caption[Convergence Rate Approximation Zero-Coupon Bonds]{\small{Approximation of the convergence rate of the zero-coupon bond prices. Except for the number of grid points $m$, model and CTMC parameters are as listed in Table \ref{tblModelCTMCParamHomo}. Zero-coupon bond price parameters are $t=0$ and $T=4$.}}
	\label{tblConvRateZBpriceHomo}
\end{table}

We observe that the approximation achieves a high level of precision across all parameters, with an average calculation time of $0.009$ seconds, illustrating the speed of the methodology.

Figures \ref{figEfficiencyZBPriceDothan} demonstrate the efficiency of CTMC methodology in valuing zero-coupon bond prices when the short-rate process follows a geometric Brownian motion (the Dothan model). Thus, alternative methods such as binomial trees (``Tree'') or Monte Carlo simulation (``Sim'') can be easily implemented. We compare the performance of the CTMC approximation to the Cox--Ross--Rubinstein binomial tree (\cite{cox1979option}) and to Monte Carlo simulation. For Monte Carlo simulation, we use an exact scheme with a number of simulations ranging from 10,000 to 200,000 with the same number of antithetic variables and 500-time steps per year. The benchmark is calculated using the CTMC method with $m=5,000$.
Figure \ref{figEfficiencyZBPriceDothan} shows the high efficiency of the CTMC method compared to these other numerical techniques. Indeed, CTMC approximation clearly outperforms other methods in terms of both calculation time and precision. 

Finally, Figure \ref{figConvergenceZB_CIRVa} shows the convergence pattern of the approximated zero-coupon bond prices as the number of grid points $m$ increases, whereas Table \ref{tblConvRateZBpriceHomo} shows the convergence rate. For the Vasicek (resp. CIR) model, we note that the approximations achieve quadratic (resp. superquadratic) convergence on average. We also observe that the two models converge smoothly and rapidly to their analytical values. Moreover, since Assumption \ref{assumpStateSpaceR} is not satisfied under the Vacisek model, the results show that theoretical convergence is possible under less restrictive conditions for a certain set of parameters. Theoretical proof is left as future research.
\subsubsection{Approximation of the Zero-Coupon Bond Option Prices}
In this section, we study the accuracy and the numerical convergence of the zero-coupon bond option prices provided in \eqref{eqCallPutHomo}, under the Vasicek and CIR models. Under these two models, the price of zero-coupon bond options has a closed-form expression, which can be found in \cite{brigoMercurio2006} Section 3.2., and thus, can serve as a benchmark in our example. We test the accuracy of the approximated option prices for different levels of moneyness and volatilities. The results are summarized in Table \ref{tblAccuracyZBoptionHomo}. Column ``price-to-strike'' shows the price-to-strike ratio, calculated as the actual zero-coupon bond price over the option strike price $K>0$. 
\begin{table}[b]
	\begin{subtable}[c]{0.495\linewidth}
		\centering
		\scalebox{0.70}{
			\begin{tabular}{ccccc}
		\hline
		$\mathbf{\sigma}$&\textbf{price-to-strike} & \textbf{CTMC} & \textbf{Benchmark} & \textbf{Abs. error} \\
		\hline 
		\multirow{5}{*}{$\mathbf{0.1}$}& 1.67&	 0.38319569 &	 0.38319569 &	4.29E-09\\
        &1.25&	 0.22325433 &	 0.22325433 &	4.25E-09 \\
        &1.00	 &     0.06581737& 	 0.06581712 & 2.56E-07 \\
        &0.83&	 0.01014165 &	 0.01014179 &	1.40E-07\\
        &0.71&	 0.00000011 &	 0.00000011 &	6.79E-10\\

		\hline
	\end{tabular}}
	\end{subtable}
	\begin{subtable}[c]{0.495\linewidth}
		\centering
		\scalebox{0.70}{
			\begin{tabular}{ccccc}
		\hline
		$\mathbf{\sigma}$&\textbf{price-to-strike} & \textbf{CTMC} & \textbf{Benchmark} & \textbf{Abs. error} \\
		\hline 
		\multirow{5}{*}{$\mathbf{0.1}$} &1.67&	0.38326833& 	0.38326833& 	5.81E-10\\
                                        &1.25&	0.22191975& 	0.22191976& 	5.04E-10\\
                                        &1.00&	0.06057118& 	0.06057118& 	4.27E-10\\
                                        &0.95&	0.02020450& 	0.02020450& 	3.55E-09\\
                                        &0.92&	0.00000000& 	0.00000000& 	1.22E-11\\
		\hline
	\end{tabular}}
	\end{subtable}
 \begin{subtable}[c]{0.49\linewidth}
		\centering
		\scalebox{0.70}{
			\begin{tabular}{ccccc}
		\hline
		$\mathbf{\sigma}$&\textbf{price-to-strike} & \textbf{CTMC} & \textbf{Benchmark} & \textbf{Abs. error} \\
		\hline 
		\multirow{5}{*}{$\mathbf{0.2}$}&1.67&	 0.39232996& 	 0.39232997 &	1.76E-08\\
        &1.25&	 0.22455160& 	 0.22455166 &	5.84E-08	\\
        &1.00&	 0.07590525 &	 0.07590491 &	     3.39E-07	\\
        &0.83&	 0.01014165 &	 0.01014179 &	1.40E-07	\\
        &0.71&	 0.00053736 &	 0.00053749 &	1.33E-07	\\

		\hline
	\end{tabular}}
	\end{subtable}
 \begin{subtable}[c]{0.49\linewidth}
		\centering
		\scalebox{0.70}{
			\begin{tabular}{ccccc}
		\hline
		$\mathbf{\sigma}$&\textbf{price-to-strike} & \textbf{CTMC} & \textbf{Benchmark} & \textbf{Abs. error} \\
		\hline 
		\multirow{5}{*}{$\mathbf{0.2}$}&1.67&	0.38334989& 	0.38334989& 	3.29E-09\\
                                        &1.25&	0.22190374& 	0.22190374& 	2.86E-09\\
                                        &1.00&	0.06045761& 	0.06045761& 	2.43E-09\\
                                        &0.95&	0.02020450& 	0.02020450& 	3.55E-09\\
                                        &0.92&	0.00000690& 	0.00000690& 	5.39E-09\\

		\hline
	\end{tabular}}
	\end{subtable}
  \begin{subtable}[c]{0.49\linewidth}
		\centering
		\scalebox{0.70}{
			\begin{tabular}{ccccc}
		\hline
		$\mathbf{\sigma}$&\textbf{price-to-strike} & \textbf{CTMC} & \textbf{Benchmark} & \textbf{Abs. error} \\
		\hline 
		\multirow{5}{*}{$\mathbf{0.3}$}&1.67& 	 0.40772961& 0.40772968 &	7.30E-08\\
                                    &1.25& 	 0.22983971 	& 0.22983969 &	1.39E-08\\
                                    &1.00& 	 0.09108337 	 &0.09108360 &	2.27E-07\\
                                    &0.83& 	 0.01014165 	 &0.01014179 &	1.40E-07\\
                                    &0.71 &	 0.00465377 	 &0.00465387 &	1.06E-07\\
		\hline
	\end{tabular}}
	\end{subtable}
  \begin{subtable}[c]{0.49\linewidth}
		\centering
		\scalebox{0.70}{
			\begin{tabular}{ccccc}
		\hline
		$\mathbf{\sigma}$&\textbf{price-to-strike} & \textbf{CTMC} & \textbf{Benchmark} & \textbf{Abs. error} \\
		\hline 
		\multirow{5}{*}{$\mathbf{0.3}$}& 1.67 &	 0.38348301 &	 0.38348311 	&1.01E-07\\
                                      &  1.25& 	 0.22187654& 	 0.22187663 	&9.48E-08\\
                                      &  1.00& 	 0.06027987 &	 0.06028002 	&1.53E-07\\
                                       &  0.95& 	 0.02020450 &	 0.02020450 	&3.55E-09\\
                                        &0.92 &	 0.00006003 	& 0.00006008 	&5.13E-08\\
		\hline
	\end{tabular}}
	\end{subtable}
  \begin{subtable}[c]{0.49\linewidth}
		\centering
		\scalebox{0.70}{
			\begin{tabular}{ccccc}
		\hline
		$\mathbf{\sigma}$&\textbf{price-to-strike} & \textbf{CTMC} & \textbf{Benchmark} & \textbf{Abs. error} \\
		\hline 
		\multirow{5}{*}{$\mathbf{0.4}$}& 1.67 &	 0.43035317 &	 0.43036293 &	9.76E-06\\
                                        &1.25 &	 0.24315058 &	 0.24315678 &	6.20E-06\\
                                        &1.00 &	 0.10988334 &	 0.10988798 &	4.64E-06\\
                                        &0.83 &	 0.01014165 &	 0.01014179 &	1.40E-07\\
                                        &0.71 &	 0.01314694 &	 0.01315072 &	3.77E-06\\
		\hline
	\end{tabular}}
		\subcaption{Vasicek model}
	\end{subtable}
  \begin{subtable}[c]{0.49\linewidth}
		\centering
		\scalebox{0.70}{
			\begin{tabular}{ccccc}
		\hline
		$\mathbf{\sigma}$&\textbf{price-to-strike} & \textbf{CTMC} & \textbf{Benchmark} & \textbf{Abs. error} \\
		\hline 
		\multirow{5}{*}{$\mathbf{0.4}$}& 1.67 	&0.38366164 &	 0.38366831 &	6.68E-06\\
                                        &1.25& 	 0.22183726 &	 0.22184350 &	6.23E-06\\
                                        &1.00 &	 0.06011562 &	 0.06012477 &	9.15E-06\\
                                        &0.95& 	 0.02020450 &	 0.02020450 &	3.55E-09\\
                                        &0.92& 	 0.00011258 &	 0.00011395 &	1.37E-06\\
		\hline
	\end{tabular}}
		\subcaption{CIR model}
	\end{subtable}
	\caption[Accuracy of the zero-coupon bond call option approximation under Vasicek and CIR models]{\small{Accuracy of the zero-coupon bond call option approximation \eqref{eqCallPutHomo} under Vasicek and CIR models. Benchmark is calculated using closed-form analytical formulas. Except for the number of grid points set to $m=1,000$, model and CTMC parameters are as listed in Table \ref{tblModelCTMCParamHomo}. Zero-coupon bond call option parameters using the notation of Corollary \ref{corCallPutHomo}: $t_{n_1}=0$, $t_{n_2}=2$, and $T=4$.}}
	\label{tblAccuracyZBoptionHomo}
\end{table} 
We observe that the approximation achieves a high level of accuracy across all volatilities and strikes. 
The convergence of the approximated call prices to the analytical formulas is illustrated in Figure \ref{figConvergenceZBoption_CIRVa} for the two models, whereas the approximated convergence rates are shown in Table \ref{tblConvRateZBoptionHomo}.
\begin{figure}[t]
	\centering
	\begin{tabular}{cc}
		\includegraphics[scale=0.3]{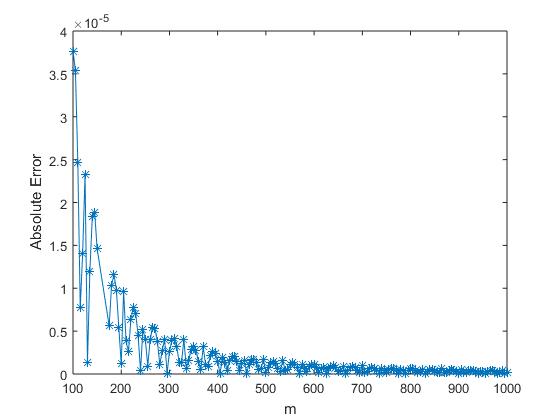} & \includegraphics[scale=0.3]{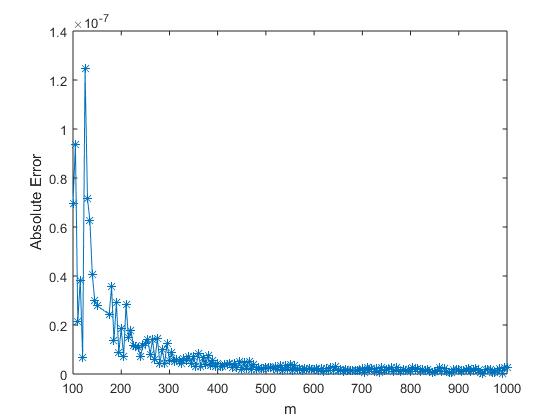}\\
  \small{\textsc{(a)} Vasicek} & \small{\textsc{(b)} CIR}  
	\end{tabular}
	\caption[Convergence pattern of the approximated zero-coupon bond call option prices under Vasicek and CIR models]{\small{Convergence pattern of the approximated zero-coupon bond call option prices as the number of grid points $m$ increases. Benchmark is calculated using closed-form analytical formulas. Except for the number of grid points $m$, model and CTMC parameters are as listed in Table \ref{tblModelCTMCParamHomo}. Zero-coupon bond call option parameters using the notation of Corollary \ref{corCallPutHomo}: $t_{n_1}=0$, $t_{n_2}=2$, $T=4$, $K=0.9$.}}\label{figConvergenceZBoption_CIRVa}
\end{figure}
\begin{table}[h]
	\begin{subtable}[c]{0.495\linewidth}
		\centering
			\begin{tabular}{ccccc}
		\hline
		\textbf{m}  & \textbf{Abs. error} & \textbf{Rate}  \\
		\hline 
		100&	3.77E-05&	-\\
        250&	4.05E-06&	2.43\\
        400&	1.49E-06&	2.12\\
        550&	1.07E-06&	1.06\\
        800&	5.90E-07&	1.58\\

	\hline
	\end{tabular}
    \subcaption{Vasicek}
	\end{subtable}
	\begin{subtable}[c]{0.495\linewidth}
		\centering
			\begin{tabular}{ccccc}
		\hline
		\textbf{m}  & \textbf{Abs. error} & \textbf{Rate}  \\
        \hline 
        100&	6.97E-08&	- \\
        250&	1.25E-08&	1.88\\
        400&	4.09E-09&	2.37\\
        550&	4.00E-09&	0.07\\
        800&	1.12E-09&	3.40\\
		\hline
	\end{tabular}
    \subcaption{CIR}
	\end{subtable}
 
	\caption[Convergence Rate Approximation Zero-Coupon Bond Options]{\small{Approximation of the convergence rate of the zero-coupon bond option prices. Benchmark is calculated using closed-form analytical formulas. Except for the number of grid points $m$, model and CTMC parameters are as listed in Table \ref{tblModelCTMCParamHomo}. Zero-coupon bond call option parameters using the notation of Corollary \ref{corCallPutHomo}: $t_{n_1}=0$, $t_{n_2}=2$, $T=4$, $K=0.9$.}}
	\label{tblConvRateZBoptionHomo}
\end{table} 
We note that the approximated call prices converge rapidly to their analytical values but exhibit a sawtooth pattern. As mentioned in Section \ref{sectionNumExpBond}, such oscillatory behavior has been observed in other research (see, for instance, \cite{zhang2019analysis}). However, the technique proposed by the authors to remove oscillation and improve convergence is not directly applicable in the present context. Further investigation into how grid design can improve convergence is left for future research.
\subsubsection{Approximation of Callable/Putable Bond Prices}\label{subsectCallableDebtHomo}
We now examine the accuracy and the convergence of Proposition \ref{propCallableputableDebt} in approximating callable/putable bonds under the Vasicek model. Accordingly, we consider a coupon-bearing bond with semi-annual coupons that mature in $4$ years $T=4$. The coupon rate, denoted below by $\alpha$, is set to $4\%$ per annum compounded semi-annually. The notional of the debt is set to $F=100$, and we assume that it can be called at any time between the second and the fourth year for no additional cost, that is, $K^c_t=100$ for $2\leq t\leq T$ and we let $K_t^c\rightarrow\infty$ when $t<2$ (since exercise is not allowed). Moreover, as there is no put feature, $K^p:=K^p_t=0$ for $0\leq t\leq T$. Finally, we assume that accrued interest is paid to the bondholder upon redemption. The contract specifications are summarized in Table \ref{tblCallableBondParam}. 
 \begin{table}[h]
	\begin{tabular}{cccccc}
		\hline
		$F$  &  $\alpha$  & $T$ & $K_t^c$& $K^p$    \\
           $100$ & $0.04$ & $4$ & $100$ & $0$  \\
		\hline
	\end{tabular}
	\caption{Callable bond contract specifications\tablefootnote{$K_t^c$=100 for $2\leq t\leq T$ and $K_t^c$ is set to a large constant when $t<2$.}}
	\label{tblCallableBondParam}
\end{table} 

Proposition \ref{propCallableputableDebt} is also used to calculate the value of the straight bond (i.e., the value of the coupon-bearing bond with no optionality, when $K^c_t\rightarrow \infty$ and $K^p_t=0$ for all $t\in[0,T]$). 
The results are summarized in Table \ref{tblAccuracyStraightBond}, whereas those for the callable bond are outlined in Table \ref{tblAccuracyCallableBondOnly}. The value of the optionality is obtained from the difference between the value of the callable and the straight bonds. 
The call option has a negative value because it is in favor of the issuer and, thus, reduces the value of the bond. 
For the straight debt, a benchmark can be obtained using a closed-form analytical formula since it can be decomposed in a series of zero-coupon bonds (see Remark \ref{rmkCouponBearingBond}). For the callable debt, the benchmark is calculated using CTMC approximation with $m=2,000$. 
\begin{table}[h]
	\begin{subtable}[c]{0.495\linewidth}
		\centering
		\scalebox{0.90}{
			\begin{tabular}{cccc}
		\hline
		$\mathbf{\kappa}$ & \textbf{CTMC} & \textbf{Benchmark} & \textbf{Rel. error} \\
		\hline 
		\textbf{0.5}&	111.5779419&	111.5781246&	1.64E-06\\
        \textbf{1}&	    104.6007798&	104.6008544&	7.13E-07\\
        \textbf{2}&	    101.3606496&	101.3606597&	9.94E-08\\
        \textbf{3}&	    100.5741608&	100.5741632&	2.42E-08\\
        \textbf{4}&	    100.2732264&	100.2732272&	8.40E-09\\

		\hline
	\end{tabular}}
    \end{subtable}
	\begin{subtable}[c]{0.495\linewidth}
		\centering
		\scalebox{0.90}{
			\begin{tabular}{cccc}
		\hline
		$\mathbf{\kappa}$ & \textbf{CTMC} & \textbf{Benchmark} & \textbf{Rel. error} \\
		\hline 
		 \textbf{0.5}&94.4289102& 	 94.4293068 &	4.20E-06\\
          \textbf{1}&  95.5617290& 	 95.5616095 &	1.25E-06\\
          \textbf{2}&	96.9174873 &	 96.9173130 &	1.80E-06\\
          \textbf{3}&	97.6890278 &	 97.6890198 &	8.14E-08\\
          \textbf{4}&	98.1735503&	98.1735203	& 3.06E-07\\

		\hline
	\end{tabular}}
	\end{subtable}
 \begin{subtable}[c]{0.49\linewidth}
		\centering
		\scalebox{0.90}{
			\begin{tabular}{cccc}
		\hline
		$\mathbf{R_0}$ & \textbf{CTMC} & \textbf{Benchmark} & \textbf{Rel. error} \\
		\hline 
		\textbf{0.02}&	106.6205602&	106.6213440&	7.35E-06\\
        \textbf{0.03}&	105.6061649&	105.6062188&	5.11E-07\\
        \textbf{0.04}&	104.6007798&	104.6008544&	7.13E-07\\
        \textbf{0.05}&	103.6050710&	103.6051563&	8.23E-07\\
		\hline
	\end{tabular}}
	\end{subtable}
 \begin{subtable}[c]{0.49\linewidth}
		\centering
		\scalebox{0.90}{
			\begin{tabular}{cccc}
		\hline
		$\mathbf{R_0}$ & \textbf{CTMC} & \textbf{Benchmark} & \textbf{Rel. error} \\
		\hline 
        \textbf{0.02}&97.3030355& 	 97.3030555& 	2.05E-07\\
        \textbf{0.03}&96.4287924& 	 96.4287773& 	1.57E-07\\
        \textbf{0.04}&95.5617290& 	 95.5616095& 	1.25E-06\\
        \textbf{0.05}&94.7018444&	 94.7020397& 	2.06E-06\\
		\hline
	\end{tabular}}
	\end{subtable}
  \begin{subtable}[c]{0.49\linewidth}
		\centering
		\scalebox{0.90}{
			\begin{tabular}{cccc}
		\hline
		$\mathbf{\theta}$ & \textbf{CTMC} & \textbf{Benchmark} & \textbf{Rel. error} \\
		\hline 
        \textbf{0.01}&	113.7506391&	113.7506574&	1.61E-07\\
        \textbf{0.02}&	110.6102083&	110.6102465&	3.46E-07\\
        \textbf{0.03}&	107.5611515&	107.5612085&	5.30E-07\\
        \textbf{0.04}&	104.6007798&	104.6008544&	7.13E-07\\
		\hline
	\end{tabular}}
	\end{subtable}
  \begin{subtable}[c]{0.49\linewidth}
		\centering
		\scalebox{0.90}{
			\begin{tabular}{cccc}
		\hline
		$\mathbf{\theta}$ & \textbf{CTMC} & \textbf{Benchmark} & \textbf{Rel. error} \\
		\hline 
      \textbf{0.01}&100.5832613& 	 100.5844116 &	1.14E-05\\
      \textbf{0.02}&98.9373817 &	 98.9375213 &	1.41E-06\\
      \textbf{0.03}&97.2623005 &	 97.2631690 &	8.93E-06\\
      \textbf{0.04}&95.5617290 &	 95.5616095 &	1.25E-06\\
		\hline
	\end{tabular}}
	\end{subtable}
  \begin{subtable}[c]{0.49\linewidth}
		\centering
		\scalebox{0.90}{
			\begin{tabular}{cccc}
		\hline
		$\mathbf{\sigma}$ & \textbf{CTMC} & \textbf{Benchmark} & \textbf{Rel. error} \\
		\hline 
        \textbf{0.1}&	101.0176528&	101.0176706&	1.76E-07\\
        \textbf{0.2}&	104.6007798&	104.6008544&	7.13E-07\\
        \textbf{0.3}&	110.8775749&	110.8777598&	1.67E-06\\
       \textbf{0.4}&	120.3458095	& 120.3474922&	1.40E-05\\
		\hline
	\end{tabular}}
		\subcaption{Straight bond}\label{tblAccuracyStraightBond}
	\end{subtable}
  \begin{subtable}[c]{0.49\linewidth}
		\centering
		\scalebox{0.90}{
			\begin{tabular}{cccc}
		\hline
		$\mathbf{\sigma}$ & \textbf{CTMC} & \textbf{Benchmark} & \textbf{Rel. error} \\
		\hline 
        \textbf{0.1}&97.0423712& 	 97.0419581& 	4.26E-06\\
        \textbf{0.2}&95.5617290& 	 95.5616095& 	1.25E-06\\
        \textbf{0.3}&95.3070805& 	 95.3073647& 	2.98E-06\\
        \textbf{0.4}&96.1965169& 	 96.1965745& 	6.00E-07\\
		\hline
	\end{tabular}}
		\subcaption{Callable bond}\label{tblAccuracyCallableBondOnly}
	\end{subtable}
	\caption[Accuracy of the CTMC method to approximate the values of straight and callable bonds under Vasicek model]{\small{Accuracy of Proposition \ref{propCallableputableDebt} to approximate the values of straight and callable bonds under Vasicek model. Model and CTMC parameters are as listed in Table \ref{tblModelCTMCParamHomo}. Contract specifications are as listed in Table \ref{tblCallableBondParam}, with the call option exercise window starting from $t=2$ to maturity.}}
	\label{tblAccuracyCallableDebtHomo}
\end{table} 
\begin{figure}[h]
	\centering
	\begin{tabular}{cc}
		\includegraphics[scale=0.3]{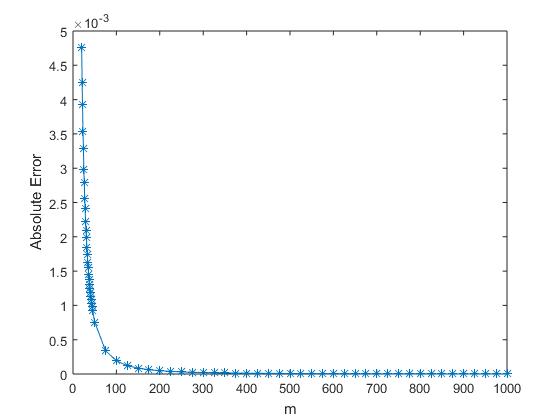} & \includegraphics[scale=0.3]{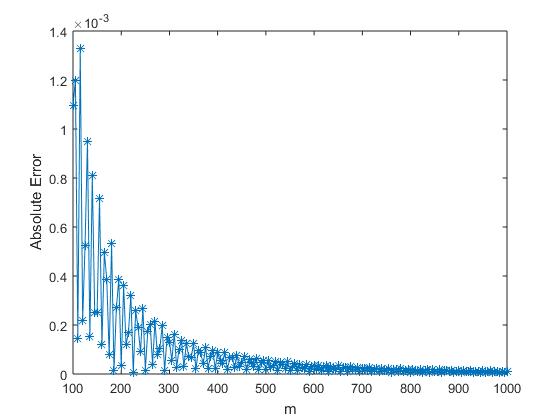}\\
  \small{\textsc{(a)} Straight bond} & \small{\textsc{(b)} Callable bond} 
	\end{tabular}
	\caption[Convergence pattern of the approximated prices of straight and callable bonds under Vasicel model]{\small{Convergence pattern of the approximated prices of straight and callable bonds as the number of grid points $m$ increases under the Vasicek model. Except for the number of grid points $m$, model and CTMC parameters are as listed in Table \ref{tblModelCTMCParamHomo}. Contract specifications are as listed in Table \ref{tblCallableBondParam}, with the call option exercise window starting from $t=2$ to maturity.}}\label{figConvergenceCallableDebt_Vasicek}
\end{figure}
\begin{table}[h]
	\begin{subtable}[c]{0.495\linewidth}
		\centering
			\begin{tabular}{ccccc}
		\hline
		\textbf{m}  & \textbf{Abs. error} & \textbf{Rate}  \\
		\hline 
		100&	1.91E-04&	- \\
        125&	1.22E-04&	2.02\\
        150&	8.47E-05&	1.99\\
        175&	6.23E-05&	1.99\\
        200&	4.78E-05&	1.99\\
	\hline
	\end{tabular}
    \subcaption{Straight bond}
	\end{subtable}
	\begin{subtable}[c]{0.495\linewidth}
		\centering
			\begin{tabular}{ccccc}
		\hline
		\textbf{m}  & \textbf{Abs. error} & \textbf{Rate}  \\
        \hline 
        100&	1.10E-03&	- \\
        125&	5.25E-04&	3.29\\
        150&	2.49E-04&	4.09\\
        175&	7.96E-05&	7.41\\
        200&	3.57E-05&	6.00\\
		\hline
	\end{tabular}
    \subcaption{Callable bond}
	\end{subtable}
 
	\caption[Convergence Rate Approximation Callable Bonds]{\small{Approximation of the convergence rate of straight and callable bond prices under Vasicek model. Except for the number of grid points $m$, model and CTMC parameters are as listed in Table \ref{tblModelCTMCParamHomo}. Contract specifications are as listed in Table \ref{tblCallableBondParam}, with the call option exercise window starting from $t=2$ to maturity.}}
	\label{tblConvRateCallableBondHomo}
\end{table} 
Again, we note that the approximation achieves a high level of accuracy in a fraction of a second across all model parameters. The average calculation time is 0.02 seconds. 

The convergence pattern of straight and callable bonds is displayed in Figure \ref{figConvergenceCallableDebt_Vasicek}, and the approximated convergence rates are shown in Table \ref{tblConvRateCallableBondHomo}. For the straight bond (resp. callable bond), we note that the approximations achieve quadratic (resp. superquadratic) convergence on average. We also observe that the straight bond converges smoothly to the analytical price. However, the callable debt exhibits a sawtooth pattern. For the two securities, the absolute error decreases rapidly to $0$.

\subsubsection{Approximation of Convertible Bond Prices}\label{subsectionNumExConvBond}
We now investigate the accuracy of Algorithm \ref{algoCBpriceTF_CTMCfast} in approximating CB prices under the Black--Scholes--Vasicek model, along with the numerical convergence of the price estimates. 

That is, we suppose that the stock price process follows a geometric Brownian motion with stochastic interest rates satisfying
 \begin{equation}
        \begin{aligned}
         \diff S_t &=(R_t-q_t) S_t\diff t+\tilde{\sigma}_S S_t \diff W^{(1)}_t,\\
    	   \diff R_t &=\kappa(\theta-R_t)\diff t+\tilde{\sigma}_R\diff W^{(2)}_t,\label{eqEDS_BSvasicek}
         \end{aligned}
	\end{equation}
  with $\kappa$, $\theta$, $\tilde{\sigma}_S$, $\tilde{\sigma}_R$>0, and $[W^{(1)},W^{(2)}]_t=\rho t$, $\rho\in[-1,1]$. 
  
  From Lemma \ref{lemmaStoX}, we find that $f(r)=\frac{\tilde{\sigma}_S}{\tilde{\sigma}_R}r$. The dynamics of the auxiliary process ${X_t=\ln(S_t)-\rho f(R_t)}$ can then be derived as
\begin{equation}
    \begin{split}\label{eqEDS_XBSvasicek}
        \diff X_t & =\mu_X(t,R_t)\diff t+\sigma_X(R_t)\diff W^\star_t\\
        \diff R_t & = \kappa(\theta-R_t)\diff t+\tilde{\sigma}_R\diff W^{(2)}_t,
    \end{split}
\end{equation}
with $\mu_X(t,R_t)=R_t -q_t-\frac{\tilde{\sigma}_S^2}{2}-\rho\frac{\tilde{\sigma}_S}{\tilde{\sigma}_R}\kappa(\theta-R_t)$, $\sigma_X=\tilde{\sigma}_S\sqrt{1-\rho^2}$, and $X_0=\ln(S_0)-\rho f(R_0)$.

  Unless stated otherwise, the model parameters for the short-rate process are the same as those used in previous examples, reported in Table \ref{tblModelCTMCParamHomo} under Vasicek model. We suppose further that $\tilde{\sigma}_S=0.2$, $q_t=0.02$ for all $t\in[0,T]$, and $\rho=-0.2$. The model parameters are summarized in Table \ref{tblModelParamBSvasicek}.
  \begin{table}[h!]
	\begin{tabular}{ccccccccc}
		\hline
		Model & ${R_0}$ & ${\kappa}$ & ${\theta}$ & ${\tilde{\sigma}_R}$ & $S_0$ & $q_t$ & $\tilde{\sigma}_S$ & $\rho$\\
		\hline 
		Black--Scholes--Vasicek  & $0.04$ & $1$ & $0.04$ & $0.20$  &$100$& $0.02$& $0.2$ & $-0.2$\\
		\hline
	\end{tabular}
	\caption[Model parameters for Black--Scholes--Vasicek model]{Model parameters}
	\label{tblModelParamBSvasicek}
\end{table} 
  
  The grid used to approximate the short-rate process, $\mathcal{S}^{(m)}_R=\{r_1,r_2,\ldots,r_m\}$, and the auxiliary process $\mathcal{S}^{(M)}_X=\{x_1,x_2,\ldots,x_M\}$, are constructed using the methodology of Tavella and Randall (\cite{tavellapricing}, Chapter 5), as explained in Section \ref{sectNumResults}, with $\tilde{\alpha}_R$ (resp.  $\tilde{\alpha}_X$) representing the non-uniformity parameter of the grid for $R^{(m)}$ (resp. $X^{(m)}$). Unless otherwise indicated, all numerical experiments are conducted using the CTMC parameters listed in Table \ref{tblCTMCParamConvBondHomo}.

  \begin{table}[h]
	\begin{tabular}{cccccccccc}
		\hline
		  & ${m}$  & $M$ & ${r_1}$ & ${r_m}$ & $\tilde{\alpha}_R$& ${x_1}$ & ${x_M}$ & $\tilde{\alpha}_X$& $\Delta_N$ \\
       \hline
		Black--Scholes--Vasicek  &$160$ & $100$ & $-30 R_0$& $25 R_0$ & $0.5$ & $0.64 X_0$ & $1.42 X_0$& $2$ &$1/100$\\
		\hline
	\end{tabular}
	\caption{CTMC parameters}
	\label{tblCTMCParamConvBondHomo}
\end{table}
 \begin{table}[h]
	\begin{tabular}{ccccc}
		\hline
		$F$  &  $\alpha$  & $T$ & $\eta$    \\
           100 & 0.05 & 1 & 1 \\
		\hline
	\end{tabular}
	\caption{CB contract specifications}
	\label{tblCBparam}
\end{table} 
The contract specifications are summarized in Table \ref{tblCBparam}. We consider a convertible bond that pays semi-annual coupons at an annual rate of $\alpha=0.05$ with a notional $F=100$. We suppose that the bond can be converted at any time from inception to maturity ($T=1$) at a conversion rate $\eta=1$.

Under this set of parameters and when both dividend yield and credit spread are assumed to be nil ($q_t=c_t=0$ for all $t\in[0,T]$), the valuation of American-style CBs is simplified to that of European-style CBs, see Corollary \ref{corCBtrivialTFcoupon}. The results of Proposition \ref{prop:ExactFormulaVanillaCB_TF}\footnote{The expected present value of future coupons should be added to the formula obtained in  Proposition \ref{prop:ExactFormulaVanillaCB_TF}.}, available online as supplemental material, can thus serve as a benchmark in our analysis. When $q_t,c_t>0$ for some $t\in [0,T]$, the benchmark is calculated using CTMC approximation with $M=160$ and $\Delta_N=1/252$, all other CTMC parameters are as listed in Table \ref{tblCTMCParamConvBondHomo}. 
The results are summarized in Table \ref{tblAccuracyCBHomo}. 
\begin{table}[h]
	\begin{subtable}[c]{0.495\linewidth}
		\centering
		\scalebox{0.90}{
			\begin{tabular}{cccc}
		\hline
		$\mathbf{S_0}$ & \textbf{CTMC} & \textbf{Benchmark} & \textbf{Rel. error} \\
		\hline 
		  \textbf{90}&   105.99171 &	 105.99224 &	4.97E-06\\
           \textbf{95}	&   108.28100 &	 108.28568 	& 4.33E-05\\
           \textbf{100}&   111.08883 &	 111.09580 	& 6.27E-05\\
           \textbf{105}&  114.37154 &	 114.37855 	&6.12E-05\\
           \textbf{110}&  118.06513 &	 118.07046 	&4.51E-05\\

		\hline
	\end{tabular}}
	\end{subtable}
	\begin{subtable}[c]{0.495\linewidth}
		\centering
		\scalebox{0.90}{
			\begin{tabular}{cccc}
		\hline
		$\mathbf{S_0}$ & \textbf{CTMC} & \textbf{Benchmark} & \textbf{Rel. error}  \\
		\hline 
		\textbf{90}&    101.80110 	& 101.80830 	&7.07E-05\\
         \textbf{95}&   104.31639 	& 104.32189 	&5.27E-05\\
         \textbf{100}&   107.35768 	& 107.35983 	&2.00E-05\\
          \textbf{105}&  110.84269 	& 110.84438 	&1.53E-05\\
          \textbf{110}&  114.70970 	& 114.70773 	&1.71E-05\\
		\hline
	\end{tabular}}
	\end{subtable}
 \begin{subtable}[c]{0.495\linewidth}
		\centering
		\scalebox{0.90}{
			\begin{tabular}{cccc}
		\hline
		$\mathbf{\sigma_S}$ & \textbf{CTMC} & \textbf{Benchmark} & \textbf{Rel. error} \\
		\hline 
		\textbf{0.10} &107.85312 	& 107.88135 &	2.62E-04\\
        \textbf{0.15}&  109.38108 	& 109.39318 &	1.11E-04\\
        \textbf{0.20}&111.08883 	 &111.09580 	&6.27E-05\\
        \textbf{0.30}&114.71935 	 &114.72313 	&3.30E-05\\
        \textbf{0.40}&118.44863 	 &118.45114 	&2.12E-05\\
		\hline
	\end{tabular}}
	\end{subtable}
 \begin{subtable}[c]{0.49\linewidth}
		\centering
		\scalebox{0.90}{
			\begin{tabular}{cccc}
		\hline
		$\mathbf{\sigma_S}$ & \textbf{CTMC} & \textbf{Benchmark} & \textbf{Rel. error}  \\
		\hline 
        \textbf{0.10} &104.28246 &	 104.29241& 	9.54E-05\\
        \textbf{0.15} &105.72587 &	 105.73050& 	4.38E-05\\
        \textbf{0.20} &107.35768 &	 107.35983 &	2.00E-05\\
        \textbf{0.30} &110.84483 &	 110.84572 &	7.97E-06\\
        \textbf{0.40} &114.43689 &	 114.43768 &	6.86E-06\\
		\hline
	\end{tabular}}
	\end{subtable}
  \begin{subtable}[c]{0.49\linewidth}
		\centering
		\scalebox{0.90}{
			\begin{tabular}{cccc}
		\hline
		$\mathbf{\rho}$ & \textbf{CTMC} & \textbf{Benchmark} & \textbf{Rel. error} \\
		\hline 
        \textbf{-0.3}&110.80352 & 	 110.81156 	&7.25E-05\\
        \textbf{-0.2}&111.08883 &	 111.09580 	&6.27E-05\\
        \textbf{0.2}&112.14557 	& 112.14307 	&2.22E-05\\
        \textbf{0.3}&112.39226 	& 112.38624 	&5.36E-05\\
		\hline
	\end{tabular}}
		\subcaption{$q_t=c_t= 0$ for all $t\in[0,T]$}
	\end{subtable}
  \begin{subtable}[c]{0.49\linewidth}
		\centering
		\scalebox{0.90}{
			\begin{tabular}{cccc}
		\hline
		$\mathbf{\rho}$ & \textbf{CTMC}  & \textbf{Benchmark} & \textbf{Rel. error} \\
		\hline 
         \textbf{-0.3}&107.08461 & 	 107.08698 	&2.22E-05\\
         \textbf{-0.2}&107.35768& 	 107.35983 	&2.00E-05\\
        \textbf{0.2}&108.37107 	& 108.36868 	&2.21E-05\\
        \textbf{-.3}&108.60805 	& 108.59625 	&1.09E-04\\
		\hline
	\end{tabular}}
		\subcaption{$q_t=0.02$, $c_t= 0.05$ for all $t\in[0,T]$}
	\end{subtable}
	\caption[Accuracy of the CB price approximations under Black--Scholes--Vasicek model]{\small{Accuracy of the CB price approximations, Algorithm \ref{algoCBpriceTF_CTMCfast}, under Black--Scholes--Vasicek model. Model, CTMC, and contract parameters are as listed in Tables \ref{tblModelParamBSvasicek} and \ref{tblCTMCParamConvBondHomo} and \ref{tblCBparam}, respectively.}}\label{tblAccuracyCBHomo}
\end{table} 
We note that the model achieves a high level of accuracy across all model parameters. 
The average calculation time for the CTMC approximated prices is less than 1.70 seconds. 
\begin{figure}[t]
	\centering
	\begin{tabular}{cc}
		\includegraphics[scale=0.3]{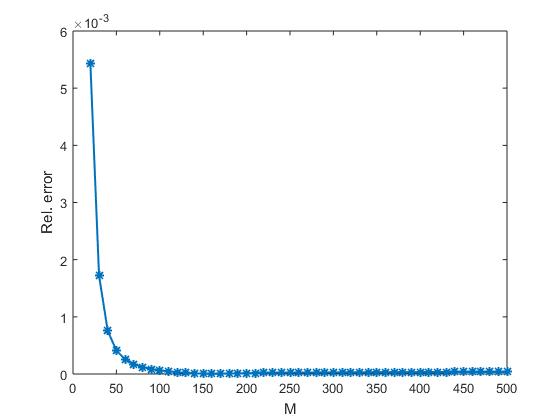} & \includegraphics[scale=0.3]{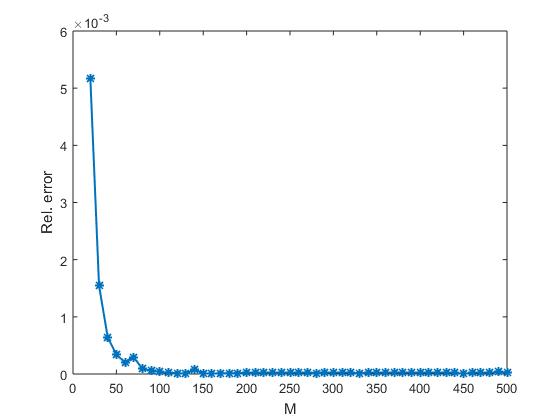}\\
  \small{\textsc{(a)} $q_t=c_t=0$ for all $t\in[0,T]$} & \small{\textsc{(b)} $q_t=0.02,c_t=0.05$ for all $t\in[0,T]$}  
	\end{tabular}
	\caption{\small{Convergence pattern of the CB price approximations, Algorithm \ref{algoCBpriceTF_CTMCfast}, under Black--Scholes--Vasicek model. Except for the number of grid points $M$ of the auxiliary process, the model, CTMC, and contract parameters are as listed in Tables \ref{tblModelParamBSvasicek}, \ref{tblCTMCParamConvBondHomo}, and \ref{tblCBparam}, respectively.}}\label{figConvergenceCB_BSvasicek}
\end{figure}
\begin{table}[h]
	\begin{subtable}[c]{0.495\linewidth}
		\centering
			\begin{tabular}{ccccc}
		\hline
		\textbf{m}  & \textbf{Rel. error} & \textbf{Rate}  \\
		\hline 
		20&	5.44E-03 &-	 \\
        50&	4.11E-04&	2.819 \\
        100&	6.27E-05&	2.711 \\
        120&	3.12E-05&	3.838 \\
        150&	5.39E-06&	7.861 \\
	\hline
	\end{tabular}
    \subcaption{$q_t=c_t=0$ for all $t\in[0,T]$}
	\end{subtable}
	\begin{subtable}[c]{0.495\linewidth}
		\centering
			\begin{tabular}{ccccc}
		\hline
		\textbf{m}  & \textbf{Rel. error} & \textbf{Rate}  \\
        \hline 
        20	&5.16E-03& - 	\\
        50	&3.36E-04&	2.982\\
        100	&3.56E-05&	3.240\\
        120	&1.46E-05&	4.895\\
        150	&8.16E-06&	2.598\\
		\hline
	\end{tabular}
    \subcaption{$q_t=0.02,c_t=0.05$ for all $t\in[0,T]$}
	\end{subtable}
 
	\caption[Convergence Rate Approximation CBs]{\small{Approximation of the convergence rate of CB prices, Algorithm \ref{algoCBpriceTF_CTMCfast}, under Black--Scholes--Vasicek model. Except for the number of grid points $M$ of the auxiliary process, the model, CTMC, and contract parameters are as listed in Tables \ref{tblModelParamBSvasicek}, \ref{tblCTMCParamConvBondHomo}, and \ref{tblCBparam}, respectively.}}
	\label{tblConvRateCBHomo}
\end{table} 

The convergence patterns of the approximation as $M$ increases are illustrated in Figure \ref{figConvergenceCB_BSvasicek}, whereas the approximated convergence rates are shown in Table \ref{tblConvRateCBHomo}. When both credit spread and dividend yield are set to nil ($q_t=c_t=0$ for all $t\in[0,T]$), the benchmark is calculated using the exact pricing formula of Proposition \ref{prop:ExactFormulaVanillaCB_TF}, available online as supplemental material. When credit risk and dividend yield are considered ($q_t,c_t>0$ for some $t\in[0,T]$), the benchmark is obtained using CTMC approximation with $M=1,000$ and $\Delta_N=1/252$, all other CTMC parameters from table \ref{tblCTMCParamConvBondHomo}.
Figure \ref{figConvergenceCB_BSvasicek} shows that the approximated prices converge rapidly and smoothly to the benchmark prices. 
\end{document}